\documentclass[
 reprint,floatfix,
 amsmath,amssymb,aps,longbibliography
]{revtex4-2}
\usepackage{hyperref}
\usepackage[utf8]{inputenc}
\usepackage{graphicx}
\usepackage{bm}% bold math
\usepackage{xcolor}
\usepackage{amsthm}
\newcommand{\nc}{\newcommand}
\nc{\nn}{\nonumber}
\nc{\txt}{\textrm}
\nc{\txtsup}{\textsuperscript}
\nc{\txtsub}{\textsubscript}
\nc{\calL}{\mathcal{L}}
\nc{\U}{\mathcal{U}}
\nc{\T}{\mathcal{T}}
\nc{\E}{\mathcal{E}}
\nc{\calH}{\mathcal{H}}
\usepackage{amsmath}
\usepackage{bbm}
\usepackage{graphicx}
\usepackage{amsfonts}
\usepackage{xcolor}
\usepackage{upgreek}

\usepackage{hyperref}
\newcommand{\CC}{\mathbb{C}}
\newcommand{\ii}{ {\rm i} }
\def\bra#1{\mathinner{\langle{#1}|}}
\def\ket#1{\mathinner{|{#1}\rangle}}

\newcommand{\ave}[1]{{\langle #1\rangle}}

\newcommand{\ZZ}{\mathbb{Z}}
\newcommand{\tr}{\rm{tr}}

\newcommand{\Aa}{\mathcal{A}}
\newcommand{\A}{\hat{\mathcal{A}}}
\newcommand{\sA}{\mathfrak{A}}
\def\one{\mathbbm{1}}
\def\bra#1{\mathinner{\langle{#1}|}}
\def\ket#1{\mathinner{|{#1}\rangle}}

\newtheorem{corollary}{Corollary}
\newtheorem{definition}{Definition}
\newtheorem{remark}{Remark}

\newcommand{\LL}{{{\cal L}}}

\def\tr{{{\rm tr}}}

\def\one{\mathbbm{1}}
\def\Re{{\,{\rm Re}\,}}
\def\Im{{\,{\rm Im}\,}}

\def\Re{{\,{\rm Re}\,}}

\newcommand{\RR}{\mathbb{R}}

\newtheorem{theorem}{Theorem}
\newtheorem{ass}{Assumption}

\begin{document}

\title{Unified theory of local quantum many-body dynamics: Eigenoperator thermalization theorems}

\author{Berislav Bu\v{c}a}
\email{berislav.buca@nbi.ku.dk}
\affiliation{Niels Bohr International Academy, Niels Bohr Institute, Copenhagen University, Universitetsparken 5, 2100 Copenhagen, Denmark}
\affiliation{Clarendon Laboratory, University of Oxford, Parks Road, Oxford OX1 3PU, United Kingdom}

%\affiliation{ }

\date{\today}
\begin{abstract}
Explaining quantum many-body dynamics is a long-held goal of physics.  A rigorous operator algebraic theory of dynamics in locally interacting systems in any dimension is provided here in terms of time-dependent equilibrium (Gibbs) ensembles. The theory explains dynamics in closed, open and time-dependent systems, provided that relevant pseudolocal quantities can be identified, and time-dependent Gibbs ensembles unify wide classes of quantum non-ergodic and ergodic systems. The theory is applied to quantum many-body scars, continuous, discrete and dissipative time crystals, Hilbert space fragmentation, lattice gauge theories, and disorder-free localization, among other cases. Novel pseudolocal classes of operators are introduced in the process: \emph{projected}-local, which are local only for some states, \emph{crypto}-local, whose locality is not manifest in terms of any finite number of local densities and \emph{transient} ones, that dictate finite-time relaxation dynamics. An immediate corollary is proving saturation of the Mazur bound for the Drude weight. This proven theory is intuitively the rigorous algebraic counterpart of the weak eigenstate thermalization hypothesis and has deep implications for thermodynamics: quantum many-body systems 'out-of-equilibrium' are actually \emph{always} in a time-dependent equilibrium state for any natural initial state. The work opens the possibility of designing novel out-of-equilibrium phases, with the newly identified \emph{scarring} and \emph{fragmentation} phase transitions being examples. 
\end{abstract}

\maketitle
%\tableofcontents

\section{Introduction}
In recent decades the eigenstate thermalization hypothesis (ETH) \cite{ETHReview} has become the cornerstone of our understanding of non-equilibrium quantum many-body dynamics. It deals with how an isolated quantum many-body system when prepared in a far-from-equilibrium initial state can relax to a state that is effectively in equilibrium for the purpose of determining expectation values of local observables. The \emph{weak} ETH in canonical form concerns the time-averaged dynamics of observables, and states, \cite{ETHReview},
\begin{align}
&\lim_{T \to \infty} \frac{1}{T}\int_0^T dt \bra{\psi(t)}O\ket{\psi(t)}=\frac{1}{N} \sum_j \bra{j}O\ket{j} \nonumber\\
&=\frac{1}{Z}\tr\left[O\exp(-\beta H+\sum_\alpha \mu_\alpha Q_\alpha)\right], \label{ETH}
\end{align}
where the sum in the first line goes over an appropriate \emph{microcanonical} window of the joint eigenstates $\ket{j}$ of the systems's Hamiltonian $H$ and a set of conservation laws $Q_\alpha$ $[Q_\alpha,H]=0$, and $Z$ and $N$ are normalization constants. The expression is usually written as a sum over the eigenstates, and should be replaced with an integral in the thermodynamic limit provided that the eigenstates are well defined there. Physically, the hypothesis states that the long-time expectation values of all observables is set only by the initial expectation values of $H$ and $Q_\alpha$ in terms of a (generalized) Gibbs ensemble via the Lagrange multiplies $\beta$ and $\mu_\alpha$. In this form the ETH also accounts for quantum integrable models that have an infinite set of conservation laws \cite{thermoreview}. These principles of local equilibriation have also been fundamental in the recent successful approach of generalized hydrodynamics \cite{castroalvaredo2016emergent,bertini2016transport,Doyon1}. The strong version of ETH \cite{ETHReview} is postulated as holding also without the time averaging in \eqref{ETH}.

However, since its formulation there has been an unmet desire to prove ETH. Only in recent years have systems been identified, beyond quantum integrable ones, that seemingly violate ETH. These include quantum many-body scars \cite{scars,scars2,SerbynScars,HoshoScars,JamirScars,ZhengScars1,ZhengScars2,ArnabScars1,ArnabScars2}, Hilbert space fragmented models \cite{Fragmentation1,Fragmentation2,fragmentationSanjay}, time crystals \cite{Wilczek2012,VedikaLazarides,liang2020time} and others \cite{Zhou_2022,JadDis,Dima}. Naturally, they have drawn lots of attention due to their non-ergodic and non-mixing dynamics seemingly defying ETH. 

Indeed, even though ETH may possibly be provable for certain quantum systems, there is no reason to expect that it will hold in full generality. The reason for this is that since ETH is a statement about thermodynamically large systems, eigenstates may be singular, or they may not exist \cite{operator2,spectraltheory}. Moreover, even eigenvalues and the spectra of $H$ and $Q_\alpha$ may display singularities. Such non-analyticies are important and cannot be ignored, e.g. they are the reason behind phase transitions \cite{operator1,spinsystsemsbook}. 

In order to therefore "prove" ETH an equivalent formulation in terms of a rigorous theory is needed. The standard choice is to describe local observables by $C^*$ algebras \cite{operator1,operator2}. In other words, one needs to move from eigenstates to eigenoperators. Here I will focus on locally interacting lattice models on hypercubic graphs in arbitrary dimension. I prove that for sufficiently low-entangled initial state, i.e. clustering (including physically realistic equilibrium states), the general long-time dynamics of all local observables in all such systems is described by a time-dependent generalized Gibbs ensemble determined by pseudolocal quantities (even without time-averaging in \eqref{ETH}). In particular, two immediate corollaries are the operator equivalent of \eqref{ETH} for the time average and saturation of the Mazur bound for the Drude weight. Moreover, it will be shown that a time-dependent Gibbs ensemble describes dynamics for all times (not just in the long-time limit). Hence, remarkably, a quantum many-body system is always in a state similar to an equilibrium state that contains exponentially decaying non-conserved quantities. In other words, the time-dependent generalized Gibbs ensemble is the effective state describing the dynamics of all local observables (the full state of the system, if defined, may be pure provided that it was initally pure). The reason for this is the bijection (one-to-one correspondance) between pseudolocal quantities and local observables \cite{Doyon}. It will turn out that knowing the dynamics of pseudolocal quantities gives all the information about the dynamics of local quantities, but the dynamics of pseudolocal quantities is much easier to find analytically. For instance, the Hamiltonian itself $H$ is a pseudolocal quantity and has trivial dynamics for an isolated system $\frac{d}{dt}H=\ii[H,H]=0$. Indeed, a priori to find dynamics of local observables one needs to find all the eigenstates of $H$. Here we will rather show how to find solutions for the dynamics of \emph{all} local observables provided one knows a \emph{much} smaller (usually finite) set of special pseudolocal quantities.

Later on, the theory will be explicitly used for recent topical non-ergodic and non-mixing cases, including quantum many-body scars, time crystals, projected Hamiltonians, Hilbert space fragmented models, disorder free localized models, lattice gauge theories and others. In the process I will introduce several kinds of pseudolocal quantities, including projected-local (local only for some states) and crypto-local whose locality is not manifest in terms of sums of local quantities. This eigenoperator thermalization theory holds for open (dissipative), time-dependent and independent quantum systems and does not rely on integrability, holds in any dimension and relies only on locality of interactions and clustering of the initial state. Note that the theory for open systems is non-trivial. One might assume that if we have a $D$-dimensional system we may treat any environment as a "fictive" enlargement of the system to $D+k$, and then studying only observables in the original system i.e. perform a Stinespring dilation \cite{paulsen_2003}, but such dilations are in general non-local and hence our theory would not apply without explicit generalizations. 

{\bf{Outline of the article--}} In Sec.~\ref{previouswork} I discuss relation between this article and the previous work done on pseudolocality and thermalization, provide the main definitions needed, and demonstrate the importance of pseudolocality with a specific example often overlooked in the literature. The overview of the main results and discussion of the proof of weak ETH, together with the descriptions of the long-time dynamics of several non-ergodic models is given in Sec.~\ref{overview}. Sec.~\ref{mainsection} contains the main theorems, which are proved in the Appendix. These are used to study various physical examples, including scars and fragmentation in Sec.~\ref{applications}. Finally, the Conclusion (Sec.~\ref{conclusion}) contains a list of possible immediate research directions stemming from the present work.

\section{Preliminaries and previous work}
\label{previouswork}
In order to rigorously study dynamics of quantum many-body systems in the thermodynamic limit we need to move to the framework of $C^*$ algebras. This is to be contrasted with standard ETH 
where one focuses on the eigenstates of the system's Hamiltonian $H$. However, eigenstates may not even exist in the thermodynamic limit and more broadly an inner product on the corresponding Hilbert space $\mathcal{H}$ is not well-defined \cite{operator2}. This is more than just a mathematical curiosity because related discontinuities are responsible for phase transitions \cite{operator1,operator2}. Indeed, operator algebraic approaches to dynamics and thermalization of quantum many-body systems have a long history - from the work of von Neumann on the ergodic theorem \cite{vonNeumann} and Robinson, Emch, Hume, Narnhofer and others, e.g. \cite{Emch1,Robinson1,Robinson2,Robinson3}. 

The physical systems studied here will be locally interacting $D$-dimensional lattice models on infinite hypercubic lattices, with a set of sites $\Gamma=\mathcal{Z}^D$, with each \emph{site} $x$ having a finite dimensional space of matrices and finite subsets of \emph{balls} $\Lambda$ of the full lattice of size $V:=|\Lambda|$. Correspondingly, local operators $O,P,Q\ldots$ will form a $C^*$-algebra $\mathfrak{U}_{loc}:=\bigotimes_{x\in \Lambda} M_d{(\CC)}$, where $d$ is the dimension of the local matrix of operators on site $x$. The algebra $\mathfrak{U}_{loc}$ is equipped with a norm that may be Cauchy completed to the full quasi-local algebra $\mathfrak{U}$ \cite{operator1,spinsystsemsbook}. More specifically, defining a standard state $\omega$ as a positive linear functional on the algebra $\mathfrak{U}_{loc}$ (with $\omega(\one)=1$), with the finite case having the familiar density matrix representation $\rho$, $\omega(O)=\tr(\rho O)$.  The standard Gelfand-Naimark-Segal (GNS) construction allows for Cauchy completion with respect to the norm induced by the (symmetrized) connected correlator inner product \cite{Doyon}, 
\begin{equation}
\ave{O,Q}^c_\omega:=\sum_{x \in \Gamma}\frac{1}{2}\omega(\{O_x^\dagger,Q\})-\omega(O_x^\dagger)\omega(Q) \label{cc}
\end{equation}
where $O_x$ is the displacement of $O$ by $x$, $\dagger$ denotes the conjugate, and $\{x,y\}:=xy+yx$. 

A crucial notion will be that of pseudolocal quantities, introduced in \cite{Tomaz_Prosen_1998,ProsenQL2} and \cite{QL4}, defined rigorously by Doyon and this is the framework we will use in the article. Under Doyon's framework the \eqref{cc} exists provided that the state $\omega$ is \emph{p-clustering}, i.e. essentially $|\omega(O Q)-\omega(O)\omega(Q)| \le C {\rm{dist}}(O,Q)^{-p}$, where $C$ is a constant that does not depend on distance and the operators $O,Q$. This defines a Hilbert space of local observables $\mathcal{H}_\omega$. Pseudolocal quantities \cite{footnote1} are linear functionals defined as limits of sequences of local operators $A_V \in \mathfrak{U}_{loc}$ satisfying the following conditions (A) $\omega(A_V^\dagger A_V) \le \gamma V$, for some $\gamma$ and $\forall V$, (B) the limit $\Aa_\omega(O):=\lim_{V \to \infty} \omega(A_V^\dagger O)$ exists for all $O \in \mathfrak{U}_{loc}$. Without loss of generality we take $\omega(A_V)=0$. One may also define two-sided pseudolocal quantities $\A_\omega(O):=\lim_{V \to \infty}\frac{1}{2} \omega(\{A_V, O\})$ and right pseudolocal quantities $\Aa^\dagger_\omega(O):=\lim_{V \to \infty} \omega(O A_V)$. The results generalize directly for all these three types of quantities and we consider them interchangeably.  Doyon demands that the whole construction is translationally invariant. We will relax this requirement partially and later fully (to be defined precisely later). It is important to note that pseudolocality is state dependent, a quantity may be pseudolocal for one state and not another. This is to be contrasted with the more standard and restrictive notion of (quasi-)locality for which one is used to thinking of pseudolocal quantities as sums of local terms. The present work will, apart from pseudolocal conserved charges that were studied extensively before (e.g. \cite{QL4}), also rely on recently introduced pseudolocal dynamical symmetries (e.g. \cite{BucaJaksch2019,Marko1}). We will see here that pseudolocal dynamical symmetries are enough to characterize essentially all non-ergodic and ergodic dynamics. 

One of Doyon's crucial results is that there is a bijection $\mathfrak{D}$ between the Hilbert space of local operators $\mathcal{H}_\omega$ and the set of $\Aa_\omega$ (denoted as $\mathfrak{A}_\omega$). More specifically, $\forall \Aa_\omega, O$ there exists a $Q \in \mathcal{H}_\omega$ such that
\begin{equation}
\Aa_\omega(O)=\ave{Q^\dagger,O}^c_\omega, \label{quantudef}
\end{equation}
and similarly for the two-sided and right pseudolocal quantities. 
Phyiscally, $Q$ is the local density of the quantity $\Aa_\omega$. Indeed, one may physically think of the limit $A=\lim_{V \to \infty} \sum_x Q_{x,V}$ as the pseudolocal quantity, as well as mathematically provided that $Q_{x,V}$ is a Cauchy sequence, which means intuitively that as $V \to \infty$ the "support" of $Q_{x,V}$ grows with $V$ in a well-defined way. 

Another important notion is that of a \emph{pseudolocal state} $\omega:=\omega_1$, defined via a set of $\{\Aa_s]\}$ and its corresponding flow over \emph{p-clustering} states,
\begin{equation}
\omega_{s_+(O)}-\omega_{s_-}(O)=\int_{s_-}^{s_+} ds \Aa_s(O), \label{pseudolocalstate}
\end{equation}
$\forall 0\le s_- < s_+\le 1$ and $\forall O$. In the very useful case when $\omega_s(O)$ is analytic we have simply,
\begin{equation}
\frac{d}{ds}\omega_s(O)=\Aa_s(O).
\end{equation}
The corresponding solution then may be thought of as a path-ordered exponential of $A_s$ provided that it exists. 

The rest of Doyon's framework concerns the long-time limit of closed many-body systems (with time propagator $\tau_t$) and deals with cases when the limit $\lim_{t\to \infty}\omega(\tau_t(O))$ exists for all local $O$, showing that the system relaxes to a linear functional counterpart of a generalized Gibbs ensemble \cite{thermoreview}. Here we will drop this requirement allowing for the study of the general dynamics for $\forall t$. This will allow us to give analytical solutions to general non-ergodic dynamics, provided that relevant pseudolocal quantities can be identified for a given model.

In \cite{Muller_2015,ETH1,ETH2,ETH3} weak ETH in the sense of typical eigenstates of the Hamiltonian being equal to the canonical ensemble has been proven. However, this does not immediately imply weak ETH in the canonical dynamical sense of Eq. \eqref{ETH} because atypical eigenstates may play an important role in the dynamics (as discussed in \cite{ETH2}) leading to violations of \eqref{ETH}. Furthermore, in the thermodynamic limit, eigenstates and eigenvalues do not provide full information about the physics of a system because eigenstates may not even exist and the spectrum of an operator need not equal its eigenavalues. 

A rigorous theory for open quantum systems (quantum Markovian semigroups) and their long-time properties has been established \cite{Frigerio,Evans,Frigerio2,Fagnola2} even for unbounded generators of the time evolution \cite{Fagnola1}. Much work has been devoted to showing that $w-\lim_{t \to \infty} \tau_t(O)=e^{\ii H t}Oe^{-\ii H t}$, where $H$ is the closed system Hamiltonian. This by itself does not imply anything for local observables or thermalization because for thermodynamically large $H$ local observables may thermalize independently of any external bath, like in closed systems.   

\subsection{Lack of relation between (restricted) spectrum generating algebras and weak ergodicity breaking}
\label{norelation}
Before moving on to the main results of the article, I wish to discuss relations between (restricted) spectrum generating algebras (projectors to eigenstates of a Hamiltonian $H$) and dynamical symmetries (conservation laws). 

A spectrum generating algebra (SGA) \cite{SGA} is defined as the existence of an operator $R$ satisfying,
\begin{equation}
[H,R]=\lambda R. \label{SGA}
\end{equation}
It is clear that for any Hamiltonian we have such operators $H\ket{E_\alpha}=E_\alpha \ket{E_\alpha}$, i.e. $R=\ket{E_\alpha}\bra{E_\beta}$ with $\lambda=E_\alpha-E_\beta$. These $R$ are \emph{not} pseudolocal in the sense from the previous subsection and play no direct role in the dynamics of physically relevant observables. In the literature \cite{SanjayReview} sometimes one extends the requirement demanding that $R^k \ket{E_0} \neq 0, \forall k<V$ (or specializes to \emph{commutant} algebras \cite{SanjayNew,SanjaySymCom}). This is then equivalent to the \emph{restricted} SGA (\cite{SanjayReview}) and has been linked with a phenomenom in \emph{quantum many-body scarred} models, called weak ergodicity breaking \cite{scars}, wherein local observables show persistent oscillations from special (but non-equilibrium and clustering) initial states. More precisely, this is a form of non-stationary dynamics, because weak ergodicity breaking can also occur for large random fluctuations around a thermal expectation value. However, the restricted SGA requirement is neither sufficient nor necessary for $R$ to have implications for dynamics of observables. This is easiest to observe in the $V \to \infty$ limit. A generic many-body $H$ will have a dense and extensive spectrum. Thus we may always find such an $R$ for any $\lambda$. More specifically, 
\begin{equation}
R=\lim_{V \to \infty} \int_{|E_\alpha-E_\beta-\lambda|<\varepsilon(V)} d\mu_\alpha d\mu_\beta \ket{E_\alpha}\bra{E_\beta}, 
\end{equation}
with $\lim_{V \to \infty } \varepsilon(V) \to 0$ being a suitable function used in taking the thermodynamic limit to avoid issues with existence of eigenstates corresponding to the continuous spectrum, and where the integral is taken over some suitable spectral measure $d\mu$. Then $R^k \ket{E_0} \neq 0$ will be fullfilled for some $\ket{E_0}$ and $\forall k$.

The fact that this requirement is not strictly necessary can be observed by considering a $D=1$ quadratic fermionic lattice model with periodic boundary conditions,
\begin{equation}
H= \sum_{j=0}^n J c^\dagger_j c_{j+1} + \mu c^\dagger_j c_{j}+h.c.,
\end{equation}
in the $J \to 0$ limit. Clearly, we have a SGA $\lim_{J \to 0}[H,c^\dagger_k]=\mu c^\dagger_k$ with $c_k=\sum_j e^{\ii k j} c_j$, with $(c^\dagger_k )^2 \ket{\psi}=0$. However, trivially, in the $J \to 0$ limit all local observables (in the fermion picture) that are off-diagonal in the number basis $[O,c^\dagger_j c_{j}] \neq 0$ persistently oscillate for arbitrarily long times for all (clustering) initial states that are not eigenstates of the total particle number operator. 

The restricted SGA, i.e. with $R^k \ket{E_0} \neq 0$  does \emph{not} imply weak ergodicity breaking, nor the converse, rather we will see that what is needed is a SGA with a pseudolocal $R$ at select frequencies $\lambda$, i.e. a pseudolocal \emph{dynamical symmetry} \cite{Buca,Marko1}.

This should not be so surprising. Indeed, setting $\lambda=0$ in \eqref{SGA} implies that $R$ commutes with $H$. Formally, every thermodynamically large Hamiltonian has a infinite number of such $R$ (projectors to its eigenstates), but this does not mean that $H$ is integrable, or even that $R$ is physically relevant conservation law. What is actually needed to have physically relevant conservation laws is locality. 

\section{Overview and main results}
\label{overview}
The assumptions of the article are listed here:

   \begin{ass}The system under question is either an open (Markovian) or closed, time-independent or time-dependent system with local (finite-range) interactions on some $D$-dimensional hypercubic lattice. \label{ass1} \end{ass}
    \begin{ass}
    The whole construction is space-translation invariant in some generalized sense, i.e. there exists an automorphism on the lattice $\mathcal{Z}^D$ denoted as $\iota_x$ for which we have $\iota_x \circ \iota_y=\iota_{x+y}$. Note that this does \emph{not} necessarily mean that the system is directly translationally invariant, and includes cases with (Bloch) translation invariance at finite momentum. This allows us to treat, e.g. modulated pseudolocal quantities \cite{Pollmann}. In any case, later on, we fully drop this requirement and allow that no such automorphism exists, i.e. we allow for arbitrary disorder. In this setup the set of pseudolocal quantities must be reduced to the set of \emph{pseudolocalized} quantities that have subextensive growth.  \label{ass2}
    \end{ass}
 \begin{ass}
    The expectation values $\ave{O}_t=\omega(\tau_t(O))$, $\forall O \in \mathfrak{U}_{loc}$ is bounded $\forall t$ and in the $t\to \infty$ limit. This is a physically reasonable assumption for most lattice models, except for perhaps bosonic ones (with infinite dimensional local Hilbert spaces) at infinite density, but such systems can be treated with standard semi-classical approaches \cite{haake1991quantum}.  \label{ass3}
    \end{ass}
   \begin{ass} The system is initially prepared in a \emph{pseudolocal} state, which essentially means that it does not have correlations that are \emph{too} long-ranged (thermal states in $D>1$ at high temperatures, ground states of gapped chains, etc. satisfy even stronger exponential clustering).  \label{ass4} \end{ass}

The main technical contribution of this article compared to the framework in \cite{Doyon} is dropping the requirement of the existence of the long-time limit and closed quantum many-body dynamics. In fact, we will discuss dynamics for \emph{general} times. Dropping this seemingly innocuous requirement will allow us to give \emph{analytical} solutions for the long-time limit of many non-ergodic and \emph{chaotic} systems (provided that all the pseudolocal quantities can be identified). This includes, but is not limited to, quantum many-body scars, Hilbert space fragmented models, time crystals and lattice gauge theories. I emphasize that even though quantum integrable systems are covered under the presented theory, the theory in no way relies on integrability. 

First, let us overview the results for closed time-independent systems with Hamiltonian $H$.

\subsection{Far-from-equilibrium states are \emph{always} in equilibrium for local observables}
\label{overview1a}
 Assume that the system is initially $t=0$ in a pseudolocal state given by the pseudolocal flow $\omega_s:=\omega_{s,t=0}$. These kinds of states can be written as exponentials of local extensive operators, for instance $\omega_{s,t=0}(O)=\tr(\exp(-s H_{t=0})O)$ for some initial Hamiltonian $H_{t=0}$ that we quench from. As we will see later, %this implies that the system remains in a pseudolocal state defined via pseudolocal quantities given by,
%\begin{equation}
%\frac{d}{dt} \A_{s,t}=\mathfrak{L}(\A_{s,t}):=\A_{s,t} \circ \adH, 
%\end{equation}
%where $\A_{s,t} \in \hat{\sA}_s$ and $\adH(\bullet):=\A_{s,t}([H,\bullet])$, and where we now view $\hat{\sA}_{s}$ as the predual of $\mathcal{H}_{s}$. The operator $\mathfrak{L}$ defining the dynamics is not self-adjoint, but its spectrum satisfies $\Re(\sigma(\mathfrak{L}))\le 0$. 
This implies that, essentially barring issues with path-ordering and existence that the system is \emph{always} (from $t=0$ to any other $t$) in a \emph{time dependent Gibbs ensemble} given by,
\begin{equation}
\rho(t)=\frac{ \exp(-\beta H +\int_0^1 du \mu_u e^{\lambda(u)t } A_u)}{\tr( \exp(-\beta H +\int_0^1 du \mu_u e^{\lambda(u) t} A_u))}, \label{tGE}
\end{equation}
where $\Re(\lambda(u))\le 0$ and we explicitly wrote the thermal part of the state. 
Note that the only time dependence is in the $e^{\lambda(u)t}$ term inside the exponential ($\lambda(u)$ does not depend on time). Although this result is rather formal, it gives us physical insight into the nature of equilibriation as illustrated in Fig.~\ref{illustration1}.

\begin{figure}[ht]
   \centering
    \includegraphics[width=\columnwidth]{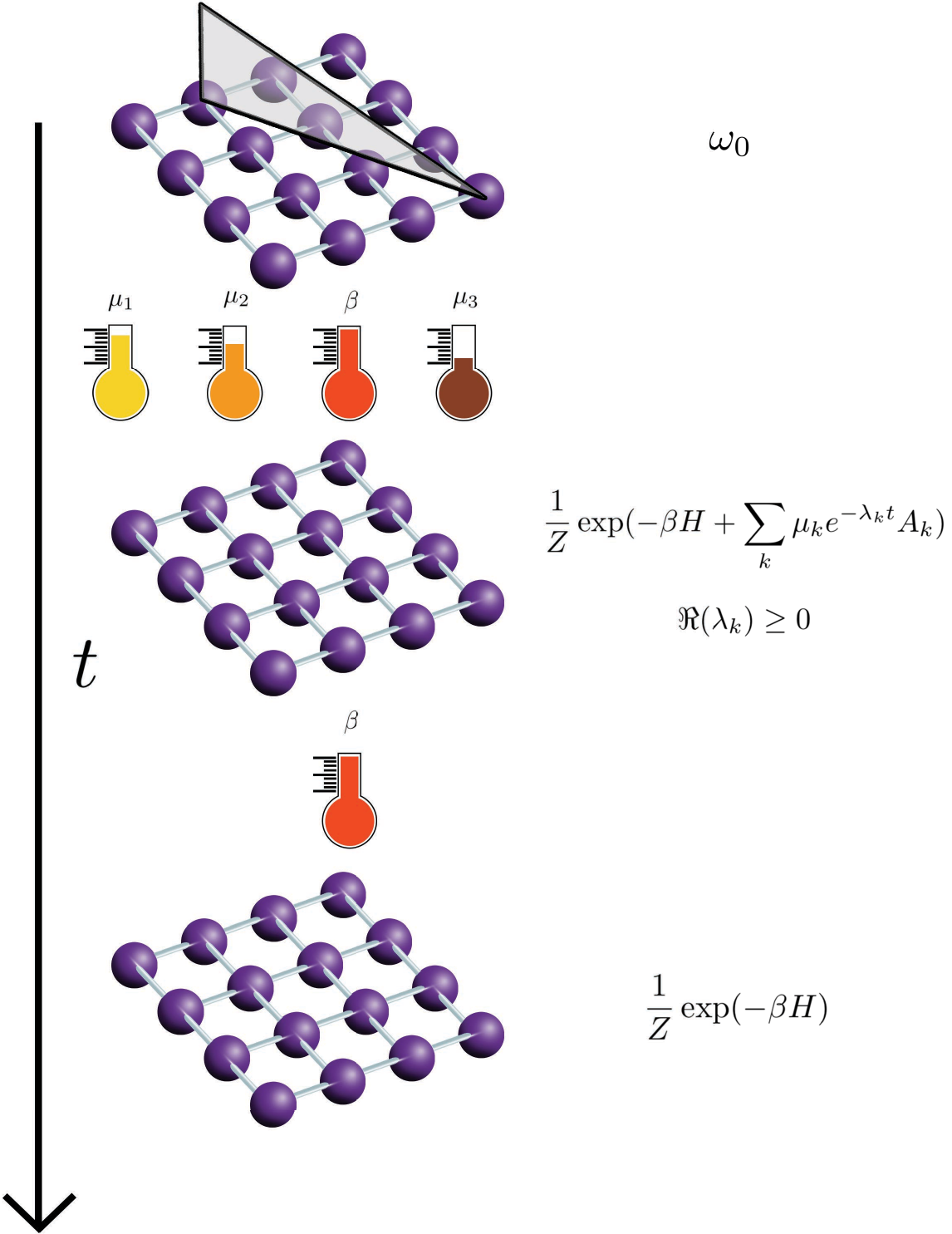}
    \caption{~\label{illustration1} An illustration of generic non-equilibrium quantum many-body dynamics. The system starts of in a clustering state (with finite power-law correlations at most) and then proceeds through an Gibbs-like state with transient pseudolocal quantities and corresponding temperatures (or chemical potentials) before thermalizing to a Gibbs ensemble.}
\end{figure}

A quantum many-body systems relaxes by starting and remaining in a Gibbs state, given by time-dependent, possibly exponentially decaying, chemical potentials $\mu_u e^{\lambda(u) t}$ with corresponding pseudolocal quantities $A_u=\sum_x Q_x(u)$. The values of $\mu_u $ are set by the initial state with flow $\omega_s$. Thus even a "far-from-equilibrium" state is a time-dependent equilibrium-like state as far as local observables are concerned. 

I also conjecture and provide evidence later, that the \emph{transient} pseudolocal quantities with $\Re(\lambda(u))<0$ are responsible for local diffusive relaxation. Note that these quantities are not the only relaxation part of \eqref{tGE}. In particular, the system may also relax algebraically in $t$ by dephasing of the purely imaginary $\Re(\lambda(u))=0$ as happens for quadratic models \cite{eigenstatedephasing}. 

\subsection{Long-time dynamics} 
\label{overview1b}
Assuming that clustering holds for long enough times and defining the Fourier transform of the expectation value $\ave{O}_{s,t}$ in the time evolved state $\omega_{s,t}$, 
\begin{equation}
\ave{O}_{s,\lambda}=\lim_{T \to \infty} \frac{1}{T} \int_0^T dt e^{\ii \lambda t} \ave{O}_{s,t}
\end{equation} 
the systems is in a time-dependent \emph{generalized} Gibbs ensemble, which is a pseudolocal state, defined via pseudolocal quantities $\A_{s,\lambda}$ that oscillate at frequencies $\lambda$, i.e. satisfying, 
\begin{equation}
\A_{s,\lambda}([H,O])=\lambda \A_{s,\lambda}(O), \quad \forall O \in \mathfrak{U}_{loc}, \quad \lambda \in \mathfrak{R}, \label{ogge} 
\end{equation}
where we define $\A_{s,\lambda}(O):=\lim_{T \to \infty}\frac{1}{T}\int_0^T e^{\ii \lambda t} \A_{s,t}(O)$ and recall that $\A_{s,t}(O)=\sum_x \ave{A_x^\dagger O}_{s,t}$ where we can subtract a constant from $A_x$ always such that $\ave{A^\dagger_x}_{s,t}=0$. Recall also that the subscript $x$ denotes translation across the lattice. 
In particular this implies the existence of corresponding (non-unique) pseudolocal sequences,
\begin{equation}
[H_V,A_V]=\lambda_V A_V, \qquad \lambda=\lim_{V \to \infty }\lambda_V \label{dynsym}
\end{equation}
where we dropped the explicit dependence on $s,\lambda$. I call these sequences \emph{dynamical symmetries} in analogy with previous work on such operators \cite{Buca,Marko1}. 

Let us assume that the following representation of the state is well defined in the $\Lambda \to \infty$ limit,
\begin{equation}
\rho^\Lambda(t)= \frac{\exp(-\beta (H_\Lambda+\sum_k e^{\ii \lambda_k t}\mu_k A^\Lambda_{k}))}{\tr_\Lambda(\exp(-\beta (H_\Lambda+\sum_k e^{\ii \lambda_k t}\mu_k A^\Lambda_{k})))}, \label{tgge1}
\end{equation}
and the set \eqref{dynsym} is countable. We will discuss when this is true later, but, essentially, the state is well defined in the thermodynamic limit provided that $|\beta|$ is large enough in $D>1$ and always in $D=1$. Then the following \emph{weak eigenoperator thermalization} result holds,
\begin{equation}
\lim_{T \to \infty} \frac{1}{T}\int_0^T dt \ave{O}_t=\lim_{\Lambda \to \infty }\tr_\Lambda(\rho^\Lambda_{\lambda=0} O),
\end{equation}
where $\rho^\Lambda_{\lambda=0}$ is the zero-frequency component, defined as before, i.e. as in \eqref{ETH}. 
This proves weak ETH in the canonical form \eqref{ETH}.

Likewise, it will be shown that the Mazur bound \cite{Mazur} for the Drude weight is saturated at all frequencies. 

\subsubsection{No translation invariance}
\label{transinv}
If we drop translation invariance even in the generalized sense, we may apply the theory provided that instead of pseudolocal quantities, we focus on pseudolocalized ones. We replace Doyon's sesquilinear form with displaced operators with the more typical inner product \cite{spinsystsemsbook},
\begin{equation}
\ave{O,Q}^{loc}_\omega:=\ave{O^\dagger Q}_\omega \label{ses},
\end{equation}
with quasi-local $O,Q \in \mathfrak{U}$. The algebra may be extended to a Hilbert space $\mathcal{H}^{loc}_{\omega}$ and its dual $(\mathcal{H}^{loc}_{\omega})^\dagger$. Pseudolocalized quantities $\Aa^{loc}_\omega(O)$ are now formed as $\lim_{\Lambda \to \infty}$ limits of $A^\dagger \in (\mathcal{H}^{loc}_{\omega})^\dagger$, which always exist according to Assumption~\ref{ass3}. The space of all $\Aa^{loc}_\omega \in \sA^{loc}_\omega$ is trivially in bijection with $\mathcal{H}^{loc}_{\omega}$ and the entire construction proceeds as before. 

This will allow us to treat Hilbert space fragmented systems that may not be translationally invariant in any sense (e.g. \cite{Abanin_MBL}). 

\subsubsection{Open and time-dependent systems}

Curiously, the preceding discussions work with very little modification for open quantum systems described by continuous quantum Markovian semigroups, i.e. the Lindblad master equation. Analogously to the Hamiltonian case, we look at \emph{local} quantum Liouvillians on $D$-dimensional hypercube lattices described by a generator of time evolution,
\begin{align} \label{lindbladeq}
&\LL O=\ii[H,O] \nonumber \\
&+\sum_{x} \int d\eta \left[2L^\dagger_x(\eta)OL_x(\eta)-\{L_x(\eta)L^\dagger_x(\eta),O\}\right],
\end{align}

where $L_x(\eta),O \in \mathfrak{U}_{loc}$ and the operator $\LL:\mathfrak{U}_{loc}^\dagger \otimes \mathfrak{U}_{loc} \to \mathfrak{U}_{loc}^\dagger \otimes \mathfrak{U}_{loc}$. The corresponding time evolution is given by a power series defined by the exponential map $\tau_t(O):=\exp(\LL t)(O)$. Physically, the Lindblad jump operators model the action of some external (memory-less) environment on the system and is applicable to a wide range of physical systems that have bath degrees of freedom that are much faster than the systems ones, e.g. quantum optics, cold atoms, etc. 

As we will see, the results in Sec.~\ref{overview1a} and~\ref{overview1b} extend directly to the quantum Markovian semigroup case if we replace $ad_H$ with $\LL$. More remarkably, we will see that, under very mild assumptions, the long-time dynamics is described by \emph{open} time-dependent (generalized) Gibbs ensembles (t-GGEs) generated by a flow given by pseudolocal quantities satisfying, 
\begin{align}
&\A_{s,\lambda}([H,O])=\lambda \A_{s,\lambda}(O), \quad \lambda \in \mathfrak{R}, \nonumber \\
&\A_{s,\lambda}([L_x(\eta),O])=\A_{s,\lambda}([L^\dagger_x(\eta),O])=0, \quad \forall O \in \mathfrak{U}_{loc}. \label{opentgge}   
\end{align}
In other words, the long-time dynamics of open quantum systems is determined exclusively by the Hamiltonian \emph{even for baths that drive the system inherently out-of-equilibrium}. The role of the Lindblad operators is mainly to select a smaller subset of pseudolocal dynamical symmetries that do so (i.e. those that commute with them in the sense of \eqref{opentgge}). 

For time dependent systems we likewise have an equivalent extension with $H \to H_\lambda$, where $H_\lambda$ is defined in the standard extended space framework \cite{howland1974stationary,pre-therm}. 

\subsection{Physics: Scars, Fragmentatation, etc.} 

It will turn out that, to the best of my knowledge, for all known cases of quantum non-ergodic dynamics, excluding integrability, it will be sufficient to consider a \emph{finite} set of dynamical symmetries \eqref{dynsym} $A^\Lambda_\mu$ that close a finite algebra with $H_\Lambda$ under commutation. This also includes the widely topical cases of quantum many-body scarred models \cite{scars,scars2,scars5,scarsdynsym2,scarsdynsym6,SerbynScars,HoshoScars,HoshoScars3}, Hilbert space fragmented models \cite{fragmentation3,Fragmentation1,fragmentationSanjay,strictlylocalfrag,ArnabFragmentation,ZanardiFRAG}, disorder-free localization models \cite{Dima}, lattice gauge theories \cite{JadDis,LGT1,LGT2,LGT3,Jad4,Jad5,Magnifico2020realtimedynamics,Michele1,Michele2} and all types of time crystals \cite{timecrystal2,Seibold,dissipativeTCobs,UedaTC,BucaJaksch2019,buca2021algebraic,Buca2,WangWang}. This allows us to use the representation \eqref{ogge} to describe the long-time dynamics of these models provided $|\beta|$ is small enough. Moreover, \eqref{ogge} defines the unique $\omega_{s,\lambda}$.  Specifically, the $\beta$ and $\mu_k$ are the only unknown parameters and these are set by the initial state. Thus, \eqref{ogge} provides an \emph{exact} solution to all such (chaotic) dynamics. 

In order to find the t-GGE for such models we need to, however, introduce two new classes of pseudolocal quantities. 

{\emph{Quantum many-body scars and the scarring phase transition.}} It will be shown that quantum many-body scarred models are defined by a novel class of pseudolocal dynamical symmetries that I call \emph{projected}-local. More specifically, in \eqref{tgge1} $A^\Lambda_{k}=\mathcal{P}_\Lambda(\sum_{x \in \Lambda} Q_x)$, where $Q_x$ is a translation of a local operator and $\mathcal{P}_\Lambda$ is a map $\mathcal{P}_\Lambda:\mathfrak{U}_\Lambda \to \mathfrak{U}_\Lambda$. In all the cases studied this will be simply a projector $\mathcal{P}_\Lambda x=PxP$, $P^2=P$. These dynamical symmetries, as sequences, satisfy the requirements of pseudolocalty, i.e. they are only pseudolocal for certain initial states $\omega_0$ and therefore lead to persistent oscillations and non-stationary dynamics only when the system is prepared in these initial states. This provides likewise an unambiguous definition of quantum many-body scars in line with the single-body definition \cite{haake1991quantum}. Crucially, the representation \eqref{tgge1} will not have clustering of correlations (and will not be a pseudolocal state) for certain values of $|\mu_k|$. In fact, there exists a critical value when this happens. This indicates the presence of a novel type of \emph{scarring} phase transition, distinct from standard thermodynamic phase transitions: in standard thermodynamic phase transitions the Gibbs ensemble no longer admits a unique representation above some value of $|\beta|$ indicating symmetry breaking. By contrast the scarring phase transitions happens because a pseudolocal quantity stops being pseudolocal. These quantities are also responsible for non-ergodicity in embedded Hamiltonians. 

{\emph{Hilbert space fragmentation, induced localization and fragmentation phase transitions ---}} Models with true Hilbert space fragmentation will turn out to be described by pseudolocalized quantities, like the ones we define for models without assuming translation invariance $\Aa_\omega^{loc}$ (Sec.~\ref{transinv}). These will lead to memory of the initial conditions \emph{locally}, i.e. a form of localization. They may be derived from the statistically localized integrals of motions \cite{SLIOM} or non-local commutant algebras of the models \cite{fragmentationSanjay}, but are distinct from them. Curiously, they cannot be written in a manifestly local manner as sums of translated local operators. Hence, we call them \emph{crypto-}local. Similar pseudolocalized charges are responsible for non-ergodic behaviour in disorder-free localization and lattice gauge theories. It can happen, in contrast to both thermodynamic phase transitions and the scarring phase transition above, that, for certain continuous changes of the chemical potentials corresponding to crypto-local quantities, the t-GGE state abruptly stops being clustered. The crypto-local quantities remain pseudolocal, however, unlike in the scarring phase transition. This still signals a change of phase and emergence of long-range order. This should be contrasted with phase transitions between thermalization (weak fragmentation) and non-ergodicity (strong fragmentation) \cite{Morningstar,Pozderac} because in the present work the phase transition is between two distinct non-ergodic phases. We will study an explicit example later.

\section{Dynamics of quantum many-body systems}
\label{mainsection}
We now turn to stating the main theorems and lemmas beginning with the ones we will need to prove the main results from the previous section. The proofs are in the Appendix. Assume that the dynamics is provided by a time-dependent Markovian \emph{closed and dense} generator,
\begin{equation}
\LL_{\Lambda,t}=\sum_{x \in \Lambda} \LL_{x,t},
\end{equation}
where $\LL_{x,t}$ is the time-dependent local Hamiltonian density (local Lindblad jump operator) translated by $x$ as in \eqref{lindbladeq}. The dynamics for $O \in \mathfrak{U}_{loc}$ is given as ($\LL_t:=\lim_{\Lambda \to \infty} \LL_{\Lambda,t}$), 
\begin{equation}
\tau_t(O):=\mathcal{T}\int_0^t dp \exp(dp \LL_p )(O), \label{propagator}
\end{equation}
where $\mathcal{T}$ is the time ordering operator. 
We implicitly defined a doubled $C^*$-algebra $\LL:\mathfrak{U}_{loc}^\dagger \otimes \mathfrak{U}_{loc} \to \mathfrak{U}_{loc}^\dagger \otimes \mathfrak{U}_{loc}$. This defines the equation of motion,
\begin{equation}
\frac{d\omega_t}{dt}(O)=\omega_t(\LL_t (O)), \label{eqmotion}
\end{equation}
where the time evolved state is $\omega_t:=\omega_0\circ \tau_t$. Note that the map is generally contractive and we have $||\tau_t(O)||\le ||O||$ and dissipative $||\tau_t(O^\dagger Q)-\tau_t(O^\dagger)\tau_t(Q)|| \ge 0$ \cite{Lindblad,Evans}. For closed systems, however, the equality in the relations holds, i.e. the map is an isometry and preserves the algebra. 

The convergence properties of the series in \eqref{propagator} has been extensively studied \cite{Frigerio}, but for our purposes what will be relevant is that $\lim_{\Lambda \to \infty}\tau^\Lambda_t(O)$ exists in the Hilbert space of local observables $\mathcal{H}_\omega$ defined previously in Sec.~\ref{previouswork}. 

Locality of time evolution for time-depedent quantum Liouvillians has been established \cite{KastoryanoEisert,Kliesch_2014,NachtergaeleOpen} using Lieb-Robinson bounds \cite{LiebRobinson} generalized for such dynamics \cite{Poulin,BarthelLR}. Define $()_\Lambda$ to be the projection to some sublattice (ball) $\Lambda$. There exists some $\varphi>0,v>0$ such that for $\Delta>2D-1$,
\begin{equation}
||\tau_t(O)-(\tau_t(O))_\Lambda||\le \varphi ||O_\Lambda|| \Delta^{D-1}\exp(-\Delta+v |t|), \label{locality}
\end{equation}
where $\Delta={\rm dist}({\rm supp}(O),\mathcal{Z}^D \backslash \Lambda)$, ${\rm dist}$ is the metric (distance) on the lattice $\mathcal{Z}^D$ and ${\rm supp}$ is the support of the operator $O$. The value $v$ is called the Lieb-Robinson velocity. Note that in contrast to the result for the isolated system \cite{BravyiLR} there is an extra polynomial dependence on distance $\Delta^{D-1}$.

Despite this extra dependence it is possible to generalize a theorem by Doyon to driven-dissipative dynamics. 

\begin{theorem}
\label{theorem1}
Let $\tau_t$ be the time evolution and $\omega$ be a pseudolocal state with flow $\omega_s$, then for $\forall O,Q \in \mathfrak{U}_{loc}$ and $t \in \RR$,
\begin{itemize} 
\item The limit $\lim_{\Lambda \to \infty} (\tau_t(O))_\Lambda$ exists in $\mathcal{H}_\omega$ and $\lim_{\Lambda \to \infty} ||(\tau_t(O))_\Lambda||_{\mathcal{H}_\omega}$ exists and is uniformly bounded with an induced form $\ave{O,Q}^c_{\omega \circ \tau_t}$
\item The state $\omega_{t}=\omega \circ \tau_t $ is pseudolocal. 
\end{itemize}
\end{theorem}
If the Lindblad jump operators terms are not zero, then it is important to note the following. 
\begin{remark}
In general we have for the form $\ave{O,Q}^c_{\omega \circ \tau_t} \neq \ave{\tau_t(O),\tau_t(Q)}^c_{\omega}.$
\end{remark}

\subsection{Eternal equilibrium}

In this abstract section we describe the finite time non-equilibrium dynamics. Although the construction is involved, it provides physical insight. for the sake of presentation, for the moment we specialize to time-independent cases and discuss generalizations in later sections. 

\begin{theorem}
If we initialize the system in a pseudolocal (non-equilibrium state) $\omega$ with flow $\omega_s$, then the state of the system $\forall t \in \RR^+$ is given by the pseudolocal state $\omega_t$ with flow,
\begin{equation}
\omega_{s,t}(O)=\omega_{t,0}(O)+\int_0^s du \Aa_{u,t}(O),
\end{equation}
with $\Aa_{u,t}:=\Aa_{\omega_u\circ\tau_t}$. The pseudolocal quantities solve for almost all $u$ the following well-defined Cauchy problem,
\begin{equation}
\frac{d}{dt}\Aa_{u,t}=\mathfrak{L}(\Aa_{u,t})=\Aa_{u,t} \circ \mathcal{L}, \label{cauchy}
\end{equation}
There exists a $M \ge 0$ such that $\mathfrak{L} :\mathcal{H}^\dagger_{u,t} \to \mathcal{H}^\dagger_{u,t}$ generates an strongly continuous contracting semigroup $\mathfrak{T}_t:=e^{(\mathfrak{L} ) t}$ solving \eqref{cauchy},
\begin{equation}
 \Aa_{u,t}=e^{M t}\mathfrak{T}_t \Aa_{u,0}    
\end{equation}admitting a resolvent and solving \eqref{cauchy}
\begin{equation}
\mathfrak{T}_t= \int_{\Gamma} d\lambda e^{\lambda t} (\lambda-\mathfrak{L})^{-1},
\end{equation}
where $\Re(\lambda)\le 0$ and $\Gamma$ is an appropriate path. \label{theorem2}
\end{theorem}

Note that $\mathfrak{L}$ is not self-adjoint in general even in the purely Hamiltonian case $L_{x}(\mu)=0$. This shows that the dynamics of local operators in even an isolated many-body system has a natural arrow of time induced by the semigroup $\mathfrak{T}_t$, i.e. $t\ge0$. Intutively, the generator $\mathcal{L}$ "intertwines" infinitesimally between $\mathcal{H}_{u,t}$ and $\mathcal{H}_{u,t+dt}$ 

\begin{remark}
The requirement about $M \ge 0$ in the generator of the semigroup is purely technical and may be dropped (i.e. set $M=0$) provided the infinite time limit $\lim_{t \to \infty}||\tau_t O||_{\mathcal{H}_\omega}$ exists $\forall O \in \mathfrak{U}_{loc}$. This should be the case in all physically reasonable examples. Otherwise, we could have unbounded values of local observables. Alternatively, we may set $M=0$ provided we are only interested in dynamics for all finite $t \in \RR$ 
\end{remark}

The result as given looks complicated even for the purely Hamiltonian case. However, if we assume that $\omega_{s,t}$ is analytic,
\begin{equation}
\frac{d}{ds}\omega_{s,t}(O)=\Aa_{s,t}(O)
\end{equation}
$\forall s \in [0,1]$ and that it admits a well-defined matrix representation, it directly follows from the linearity of the functional $\Aa_{u,t}$, that the thermodynamic state of the system is the limit of a time-dependent Cauchy sequence in $\Lambda$, which closely resembles a Gibbs ensemble, i.e. the state is of the form,
\begin{equation}
\rho^\Lambda(t)= \mathcal{P}\frac{\exp(-\beta (H_\Lambda+\int_0^1 du e^{\lambda(u) t}\mu_u A^\Lambda_{u}))}{\tr_\Lambda(\mathcal{P} \exp(-\beta (H_\Lambda+\int_0^1 du e^{\lambda(u) t}\mu_u A^\Lambda_{u})))}, \label{tge2}
\end{equation}
where $\mathcal{P}$ is a suitable path-ordering along the flow of the pseudolocal state.

Physically, the system proceeds from a state $\omega_0$ which is an eigenstate or thermal state of a local Hamiltonian that is distinct from the Hamiltonian driving the time-evolution of the system. This state admits a decomposition in terms of pseudolocal quantities $A_u$ for the given generator $\mathfrak{L}$ that may be dived into two main classes, 
\begin{enumerate}
    \item Those for which $\Re(\lambda(u))<0$, i.e. \emph{exponentially decaying} transient ones that disappear from $\rho(t)$ exponentially quickly and correspond to exponential decay of expectation values of local observables. \label{1}
    \item Those for which $\Re(\lambda(u))=0$, which may be further subdivided,
    \begin{enumerate}
        \item Those for which the spectrum is continuous around $\lambda(u)$. These may correspond to polynomial decay of expectation values of local observables, which may be seen by e.g. invoking the stationary phase approximation. \label{2a}
        \item Those $\lambda(u)$ around which the spectrum of $\mathfrak{L}$ is isolated. These are either $\lambda(u)=0$, i.e. these are the pseudolocal conservation laws and must include the Hamiltonian $H$ itself, or $\Im(\lambda(u))\neq0$ and these are pseudolocal dynamical symmetries that will be studied in the next section. The latter correpond to persistent oscillations at fundamental frequencies $\Im(\lambda(u))$.\label{2b}
    \end{enumerate}
\end{enumerate}

 Further intuition about the Th.~\ref{theorem2} can be gained for the closed infinitely large system $\Lambda=\infty$ case, by observing that the main statement of the Theorem in Eq.~\eqref{cauchy} can be formally and unrigorously written as $[H_\infty,A^\infty_u]=\lambda(u) A^\infty_u$, where $\lambda_u$ need not be real. This is because $H_\infty$ need not be self-adjoint on the entire Hilbert space. By contrast $H_\Lambda$ for finite system $\Lambda$ always has real and countable eigenstates $H_\Lambda \ket{E^\Lambda_n}=E^\Lambda_n \ket{E^\Lambda_n}$. However, most of the corresponding $A^\Lambda_{n,m}=\ket{E^\Lambda_n}\bra{E^\Lambda_m}$, and $[H_\Lambda,A^\Lambda_{n,m}]=(E^\Lambda_n-E^\Lambda_m) A^\Lambda_{n,m}$ become thermodynamically irrelevant for \emph{all} local observables $O$, i.e. $\lim_{\Lambda \to \infty}\ave{A^\Lambda_{n,m} O}_t=0, \forall O$ (similary to a well-known assumption of ETH \cite{ETHReview}). The pseudolocal $A^\Lambda_u$ are precisely the linear combinations of such $A^\Lambda_{n,m}$ (and the only linear combinations) that have finite overlap with (at least some) local observables in the thermodynamic limit. Moreover, they are such that they grow at most extensively with system size \emph{for the given state} and hence are well-defined in the thermodynamic limit. In other words, a naive solution of the dynamics of local observables requires knowning all the eigenstates and energies of $H$, but most of these are not thermodynamically relevant. The same information can be gained about dynamics of local observables if one knows a much smaller subset in terms of $A^\Lambda_u$. 

The physical relevance, if any, of the residual spectrum of $\mathfrak{L}$ is not immediately clear. Cases~\ref{2a} and \ref{2b} have been studied previously and will be discussed in the next section. To the best of my knowledge there are no known constructions of the transient pseudolocal quantities in~\ref{1}.  

However, here I give a physically motivated conjecture that they are responsible for diffusive relaxation that may be studied by quantum hydrodynamics \cite{quantumhydro}. 

To see this, consider a 1D lattice with a (unrigorous) local conservation law $[H,Z_0]=0$ at finite momentum $Z_k=\sum_x e^{\ii k x} z(x)$, where $z(x)$ is the translated local density. The diffusion equation for this charge is,
\begin{equation}
\frac{\partial}{\partial t}z(x,t)=\kappa \frac{\partial^2}{\partial x^2} z(x,t),
\end{equation}
where $z(x,t):=\tau_t(z(x))$.
The finite momentum solution is $Z(k,t)=Z(k,0) e^{-k^2 \kappa t}$. Studying $Z(k,t)=\tr{\rho(t)Z_k}$ and comparing with \eqref{tge2} indicates that the exponential decay of $Z_k$ at finite momentum proceeds because the overlap between the transient pseudolocal quantities $\tr(A_u Z_k)\neq 0$ for $k\neq 0$. Hence, it is reasonable to assume that the transient pseudolocal quantities are responsible for difussion. 

The relation between the pseudolocal quantities and transport is more complicated for other types of relaxation. Consider the 1D convection equation, solving for, e.g. ballistic transport in integrable models \cite{castroalvaredo2016emergent,bertini2016transport} for the simplest linear case,
\begin{equation}
\frac{\partial }{\partial t} z(x,t)=v \frac{\partial }{\partial x} z(x,t).
\end{equation}
For an initial condition $z(x,0)=z_0(x)$, the equation is solved by any $z(x,t)=z_0(x+vt)$. So depending on the initial condition it could correspond to a local faster-than-exponential relaxation (e.g. for a Gaussian wave packet $z_0(x)=C \exp(-a x^2)$) to persistent oscillations for an initial condition with finite momentum $k$.

\subsection{Asymptotic dynamics}

In this section we develop a general theory in terms of pseudolocal quantities for the long-time dynamics. 

\begin{theorem}[General eigenoperator thermalization]
Assume that the system in (non-equilibrium) pseudolocal state $\omega_0$ with flow $\omega_s$. In the long-time limit the state of the system is a pseudolocal (open) \emph{time-dependent} generalized Gibbs ensemble (t-GGE) $\omega_t$ with flow $\omega_{s,t}$,
\begin{enumerate} \item \label{case1} The t-GGE satisfies, 
\begin{equation} \label{fullcond}
\omega_{s,\lambda}(\LL_t(O))=-\ii\lambda \omega_{s,\lambda}(O) \qquad \forall O \in \mathfrak{U}_{loc}
\end{equation}
where $\omega_{s,\lambda}$ is the component of the state at frequency $\lambda$, i.e. $\omega_{s,\lambda}:=\lim_{T \to \infty}\frac{1}{T}\int^T_0 dt e^{\ii\lambda t} \omega_{s,t}$ with $\lambda \in \RR$. 

\item \label{case2} Provided that the dual map $\tau_t^\dagger$ has a faithful stationary state and a time-independent generator then we also have, 
\begin{align}
&\omega_{s,\lambda}([H,O])=-\lambda \omega_{s,\lambda}(O) \label{hamcond}\\
&\omega_{s,\lambda}([L_x(\eta),O])=\omega_{s,\lambda}([L^\dagger_x(\eta),O])=0, \label{lindcond}\qquad \forall O \in \mathfrak{U}_{loc}
\end{align}
\end{enumerate}
In both case~\ref{case1} and \ref{case2} the corresponding quantities $\A_{s,\lambda}(O):=\lim_{T \to \infty}\frac{1}{T}\int_0^T e^{\ii \lambda t} \A_{s,t}(O)$ satisfy the same relations as the state (with $\omega_{s,\lambda} \to \A_{s,\lambda}$ in \eqref{fullcond}, resp. \eqref{hamcond} and \eqref{lindcond}) for almost all $s$, i.e. in case~\ref{case2} they satisfy \eqref{opentgge}. The quantities $\A_{s,t}(O)$ are called \emph{pseudolocal dynamical symmetries}. 
\label{theorem3}
\end{theorem}
Naturally, for the Hamiltonian case the requirement of the faithful stationary state is trivial (e.g. the tracial state is always a faithful stationary state) and $L_x(\eta)=0$ so only the Hamiltonian condition is relevant. 

Note that the functional $\omega_{\lambda}$ is a pseudolocal functional (in the sense of \eqref{pseudolocalstate}) with the flow $\omega_{s,t}$ from above, but it is not necessarily positive for $\lambda \neq 0$. But we will abuse terminology and still refer to it as a "state". The flow $\omega_{s,t}$ is two-dimensional and we may reduce it to a single parameter flow $\omega_{\lambda,s}$ with flow along the time direction being infinite and deformed by $e^{\ii \lambda t}$. The pseudolocal quantities $\A_{s,\lambda}(O)$ are defined across a family of pseudolocal states parameterized by $t$. The equations \eqref{hamcond} and \eqref{lindcond} need not be finite. Indeed, if there is no frequency $\lambda$ in the dynamics, they give $0=0$ identically.  

These results may appear to be daunting to apply to any given Hamiltonian, but as we will see, they have necessary and sufficient interpretations in terms of standard theoretical physics concepts - eigenoperators and equilibrium states. Physical intuition about Th.~\ref{theorem3} may be gained by observing that the criteria of the theorem essentially state that for finite system $\Lambda$ there exist $A^\Lambda_\lambda$ such that $[H_\Lambda,A^\Lambda_\lambda]=\lambda A^\Lambda_\lambda$ where $A^\Lambda_\lambda$ are pseudolocal, in the sense discussed previously, for the finite frequency state $\omega_{s,\lambda}$, i.e. such that $\lim_{\Lambda \to \infty} \ave{A^\Lambda_\lambda O}$ exists and is non-zero for at least some local $O$. Moreover, the Lindblad part states that these $A^\Lambda_\lambda$ are invisible to the dissipation (unaffected by it) when the dynamics of local observables at frequency $\lambda$ is described by the same $\omega_{s,\lambda}$ state. 

Indeed, any operator sequence $A_\Lambda$ that remains pseudolocal during the time evolution, i.e. there exists $\gamma>0$ such that it satisfies $\omega_{t}(A^\dagger_\Lambda A_\Lambda) \le \gamma |\Lambda|$ (i.e. extensive) and $\lim_{\Lambda \to \infty}\omega_{t}(A_\Lambda^\dagger O)$ exists $\forall O \in \mathfrak{U}_{loc}$, then $A_\Lambda-\omega_t(A_\Lambda)$ defines a pseudolocal dynamical symmetry. A very convenient and general case is the following. 

\begin{corollary}
Assume that there exists a pseudolocal sequence $A_\Lambda$ and a clustering initial state $\omega$ in the sense defined previously and assume that they satisfy under time evolution (where $\omega_t=\omega \circ \tau_t$),
\begin{align}
&(\tau_t(A_\Lambda))_{\Lambda}=e^{-\ii \lambda t} A_\Lambda+Z_\Lambda(t), \qquad \lambda \in \RR\\
&\exists \gamma>0, \quad |\omega_t(A_\Lambda^\dagger A_\Lambda)| \le \gamma |\Lambda|, \quad \forall t,\Lambda, \label{volumegrowth} \\
& \lim_{\Lambda \to \infty} \omega_t(A^\dagger_\Lambda O) \in \CC, \qquad \forall t, O \in \mathfrak{U}_{loc}
\end{align}
This defines a (left) pseudolocal dynamical symmetry if $\omega(e^{-\ii \lambda t} A_\Lambda+Z_\Lambda(t))=0$, if not then  we may use the zero-average sequence $e^{-\ii \lambda t} A_\Lambda+Z_\Lambda(t)-\omega(e^{-\ii \lambda t} A_\Lambda+Z_\Lambda(t))$ to define a (left) pseudolocal dynamical symmetry. 
\label{pseudolocaldyn}
\end{corollary}
This follows directly from Theorem~\ref{theorem3} and the definition of pseudolocal quantities. We may analogously construct right and two-sided pseudolocal dynamical symmetries. 

This type of pseudolocal dynamical symmetry is actually more general than needed to study all topical examples of quantum many-body non-ergodicity from the literature, so we will further specialize. 

\begin{definition}[Simple pseudolocality] If a pseudolocal sequence in Cor.~\ref{pseudolocaldyn} has $Z_\Lambda(t)=0$, then we will call such a sequence \emph{simple}. \label{definition1}
\end{definition}
 Several well-known and newly introduced (in later sections) examples are given in Fig.~\ref{figureillustration2}.

\begin{figure}[ht]
   \centering
    \includegraphics[width=\columnwidth]{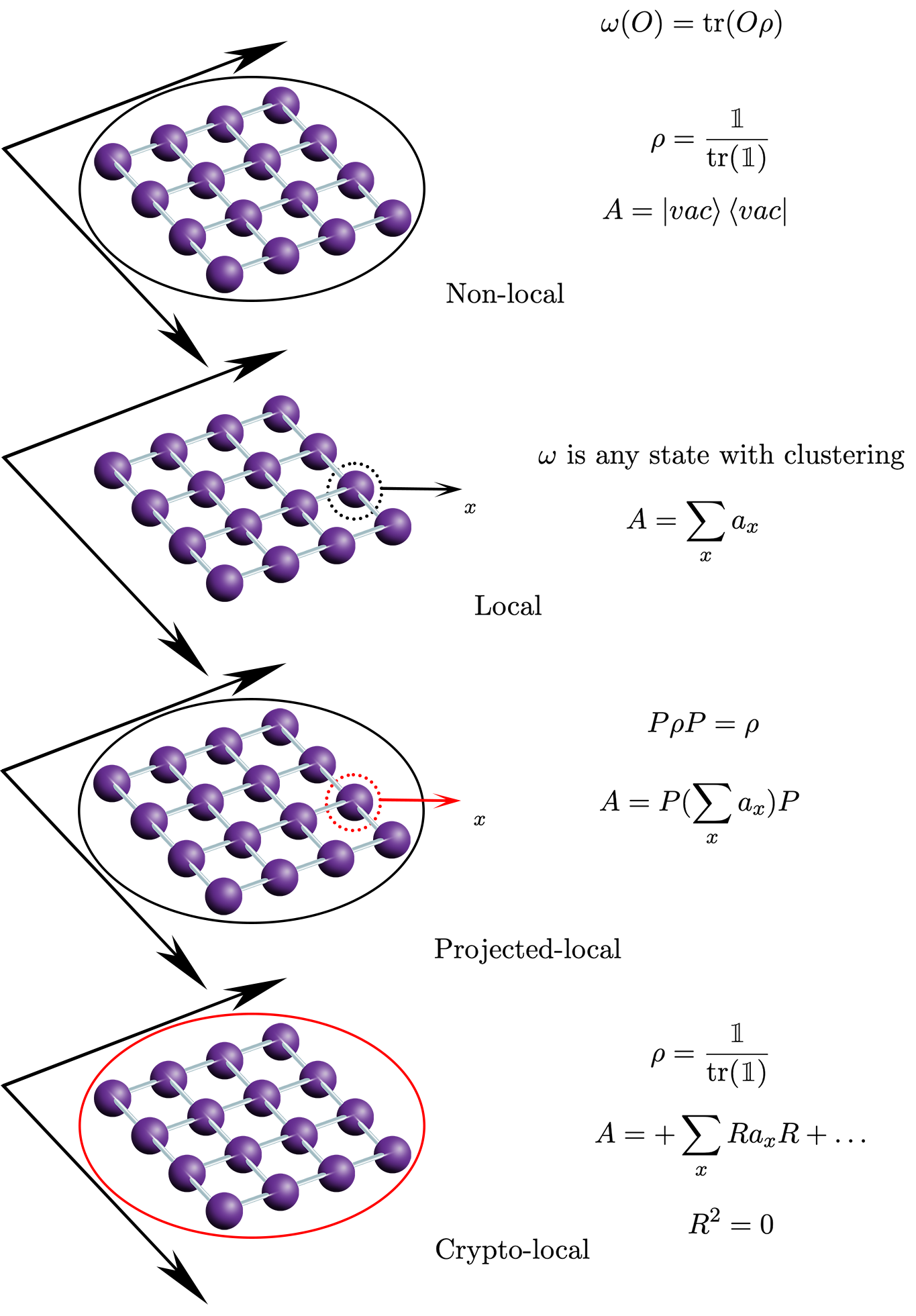}
    \caption{~\label{figureillustration2} An illustration of non-locality vs different kinds of pseudolocality. The shown subsystem is of arbitrary size $\Lambda$. First case (Non-local): $A$ has support on the entire subsystem for all sizes of the subsystem (e.g. a projector to vacuum state). It is not pseudolocal for the infinite temperature state. Second case (Local): $A$ is a simple sum of operators that act only on one site (translated to site $x$). It is pseudolocal for any state with clustering. Third case (Projected-local): A local operator (second case) is sandwiched between a projector $P$ (that acts on the entire subsystem) and the clustering state is entirely inside the subspace $P$ projects to. It is pseudolocal for such a state because the state does not "see" the projector. Fourth case (Crypto-local): $A$ contains terms (denoted $R$) that act on the entire subsystem for all size of the subsystem, but they cancel and they are not visible in the size of $A$ for e.g. the infinite temperature state.}
\end{figure}

We are now able to fully characterize the long-time dynamics in terms of the finite frequencies of the t-GGE. Let us look at \emph{local} dynamical symmetries, i.e. those whose operators sequences may be written as $A_\Lambda=\sum_{x \in \Lambda} a_x$, where $a_x$ is the translation by $x$ of the the local density $a \in \mathfrak{U}_{loc}$. 

\begin{theorem}[Local t-GGE]
\label{theorem4}
Let the initial state be pseudolocal as before. Assume that there is a countable finite set of local dynamical symmetries parameterized by $k=1,2,\ldots,M$ with sequences $\{(A_k)_\Lambda\}$ with corresponding frequencies $\lambda_l \in \RR$ (i.e. $Z_k(t)=0$) and assume that this set forms a basis for an algebra closed under commutation (i.e. any element of the algebra may written as a linear combination in this basis). Assume that the dynamics $\tau_t$ has no pseudolocal dynamical symmetries except those generated by this set. Then the long-time dynamics is given by a normal state $\omega^{tGGE}_t$ with matrix representation,
\begin{align}
\omega^{tGGE}_t(O)&=\lim_{\Lambda \to \infty} \tr_{\Lambda}\left[(\rho^{tGGE}_t)_\Lambda O\right], \qquad \forall O \in \mathfrak{U}_{\Lambda},\nonumber \\[1em]
(\rho^{tGGE}_t)_\Lambda&=\frac{\exp\left[\sum_k \mu_k e^{\ii \lambda_k t}(A_k)_\Lambda \right]}{\tr_\Lambda[\exp\left[\sum_k \mu_k e^{\ii \lambda_k t}(A_k)_\Lambda \right]]},
\end{align}
which exists, is a t-GGE in the sense above with exponential clustering, and is analytic in $\mu_k,t$ for $|\mu_k|<\mu_{*}$ and $\forall t$. In 1D $\mu_{*}=\infty$
\end{theorem}
The same holds for open quantum systems if $[L(\eta),A_k]=[L(\eta)^\dagger,A_k]=0$. Note that we view conservation laws as special cases of dynamical symmetries with $\lambda_k=0$. In particular, for any system we always have $A_0=H$, $\lambda_0=0$ and $\mu_0=\beta$. 

Taking $\lambda=0$, which is must be finite for at least some $O$, we arrive at the operator form of the (dynamical) weak ETH immediately both for open and closed systems. This also proves the ETH for open systems \cite{openETH,openETH2,Keeling}. 

\begin{corollary}[Weak eigenoperator thermalization] 
\label{ETHcorollary}
Let the system satisfy Assumptions~\ref{ass1}-\ref{ass4}. Let $\omega$ be a pseudolocal initial state and let the set $\{(A_k)_\Lambda,\mu_k,\lambda_k\}$ be as in Theorem~\ref{theorem4}. Pick a subset with $\lambda_k=0$. Let $\kappa:=\{k|\lambda_k=0\}$ We then have,
\begin{align}
&\lim_{T \to \infty} \frac{1}{T}\int_0^T dt \omega(\tau_t(O))=\nonumber\\
&\lim_{\Lambda \to \infty}\tr_\Lambda\left[\frac{\exp\left[\sum_{k \in \kappa} \mu_k (A_k)_\Lambda \right]}{\tr_\Lambda\left[\exp\left[\sum_{k \in \kappa} \mu_k (A_k)_\Lambda \right]\right]}O\right], \quad \forall O \in \mathfrak{U}_{loc},
\end{align}
where $\mu_0:=\beta$, $(A_0)_\Lambda=H_\Lambda$.  
\end{corollary}
 Intuitively, the canonical weak ETH is equivalent to the zero-frequency case of Th.~\ref{theorem3} and hence the corresponding pseudolcoal dynamical symmetries are only pseudolocal conservation laws and the corresponding state is a standard (generalized) Gibbs state.
\begin{remark}
If an additional physically reasonable assumption that $\exists \lim_{t \to \infty}||\tau_t(O)||_{\mathcal{H}_\omega}$ holds, then the $\lim_{t \to \infty} \omega\circ\tau_t=\omega^{tGGE}_t$, i.e. $\omega^{tGGE}_t$ is the asymptotic state in the strong sense in Theorems~\ref{theorem3} and \ref{theorem4}, not just at given frequencies. 
\end{remark}

In order to avoid confusion, let me now remark that in \cite{Diehl_2008} it was shown that local quantum Liouvillians may be used to engineer unique stationary dark states with long-range order (i.e. not pseudolocal states). I emphasize that there is no contradiction with the work presented here because it deals with time-averaged dynamics rather than the stationary state dynamics. As the unique stationary dark state is not clustered, using Lieb-Robinson bounds it follows that the relaxation time of such models cannot be independent of system size \cite{Poulin}. In cases where the relaxation time diverges the long-time average need not equal the exact stationary state average, cf. \cite{ZnidaricRelaxation,HoshoGap}. Moreover, \cite{Diehl_2008} shows uniqueness only for the exact 0 eigenvalue and it would be interesting to check what happens for purely imaginary eigenvalues as they also contribute to the long-time dynamics.

Let us now step away from far-from-equilibrium dynamics for the moment, and see what we can say about dynamics near equilibrium. Let us assume that the system is in a clustering equilibrium state in the sense that $\omega\circ\tau_t=\omega$. A very useful concept in that case is the Mazur bound \cite{Mazur,Ilievski_2012} for the Drude weight \cite{FiniteFreqDrude}, which provides a lower bound for the susceptibilities and ballistic transport. It is given in terms of quasilocal charges and its finite frequency version \cite{Marko2}, which is sufficient for our purposes, reads,
\begin{equation}
\lim_{T \to \infty}\frac{1}{NT}\int_0^T dt e^{-\ii \lambda t}\omega(O^\dagger O(t))\ge\sum_j \frac{|\omega(O^\dagger A_j)|^2}{N \omega(A^\dagger_jA_j)},
\end{equation}
where $A_j=\sum_x a^{(j)}_x$, $O=\sum_x o_x$, $[H,A_j]=\lambda A_j$ and we assume that $A_j$ are orthogonal in the sense $\omega(A^\dagger_j A_k)=0$ if $k\neq j$. It has long been conjectured that the Mazur bound saturates if all the $A_j$ are known \cite{QL4,karrasch2017proposal}. A partial result for finite systems called the Suzuki equality exists \cite{SUZUKI1971277,Dhar_2021}, but this conjecture remained unproven in the thermodynamic limit where it would have deep implications for e.g. superconducting transport in the linear response regime \cite{sirker2020transport}.

In the setup discussed here the Mazur bound becomes a straightforward equality. 

\begin{corollary}[Mazur equality]
\label{Mazur}
Let $\omega=\omega\circ\tau_t$ be a pseudolocal equilibrium state. We have the following identity,
\begin{equation}
\lim_{T \to \infty}\frac{1}{T}\int_0^T dt e^{-\ii \lambda t}\omega(O^\dagger\tau_t(o))=\A_{1,\lambda}(o),
\end{equation}
where $\A_{1,\lambda}$ is a pseudolocal dynamical symmetries, i.e. satisfying $\A_{1,\lambda}([H,q])=\lambda \A_{1,\lambda}(q)$, $\forall q \in \mathfrak{U}_{loc}$, which has maximal overlap with $o$ in the sense that, 
\begin{equation}
\A_{1,\lambda}(q):=\lim_{T \to \infty} \ave{\mathfrak{D}(\frac{1}{T}\int_0^T dt e^{\ii \lambda t}\tau_t(o)),q}{}^c_\omega,
\end{equation}
where $\mathfrak{D}:\mathcal{H}_{\omega} \to \mathfrak{U}_\omega$ is the bijection between the Hilbert space of local observables and the pseudolocal quantities. 
\end{corollary}
This follows directly from the definition of the inner product on $\mathcal{H}_\omega$, pseudolocal dynamical symmetries, Theorem~\ref{theorem1} and the existence of the bijection. By construction if we know all the pseudolocal dynamical symmetries we also know the ones having maximal overlap with $o$. In this setup normalization is not needed as it is present by construction. 

This is likewise consistent with similar results obtained for fully extensive operators of closed Hamiltonians at zero frequency and finite momentum \cite{Doyon_2022} and finite frequency \cite{Doyon3}.

Before we turn to studying examples let us state the following simple result in anticipation of its use for pseudolocal quantities later on. 

\begin{corollary}
Assuming that a sequence $A_\Lambda \in \mathfrak{U}_{loc}$ is a simple pseudolocal sequence defining \emph{left} pseudolocal quantities $\Aa_{s}(O)$ with respect to some state $\omega$ with flow $\omega_s$, additionally satisfying $\tau_t(A^\dagger_\Lambda A_\Lambda )=\tau_t(A^\dagger_\Lambda) \tau_t(A_\Lambda)$, i.e. a sequence satisfying,
\begin{enumerate}
    \item $\omega_s(A_\Lambda^\dagger A_\Lambda) \le \gamma |\Lambda|^D$, for some $\gamma$ and $\forall \Lambda$ and,
    \item the limit $\Aa_{s}(O):=\lim_{\Lambda \to \infty} \omega(A_\Lambda^\dagger O)$ exists for all $O \in \mathfrak{U}_{loc}$,
\end{enumerate}
then it remains a simple pseudolocal sequence wrt to $\omega \circ \tau_\tau$, $\forall \tau$. The similar holds for two-sided and right pseudolocal sequences. 
\label{corpseudo}
\end{corollary}

This immediately follows from Th.~\ref{theorem1} and Def.~\ref{definition1}. Note that the extra requirement $\tau_t(A^\dagger_\Lambda A_\Lambda )=\tau_t(A^\dagger_\Lambda) \tau_t(A_\Lambda)$ is trivially satisfied for closed systems. 

\subsection{Generalizations to time-dependent systems and systems without translational invariance}
\label{generalization}
In the previous section I specialized to time-independent systems at certain points for the sake of clarity of presentation. Likewise, I also specialized to systems that have translational invariance in some sense (e.g. at finite momentum). Let us now discuss how to generalize the main results. 

{\bf Time-dependent systems --}. The first part of Theorem~\ref{theorem3} applies to to time-dependent systems. The second part holds for periodic dynamics provided we work in the extended space representation \cite{howland1974stationary}. To sketch the idea, let $H(t)=H(t+P)$, we extended the equations of motion,
\begin{equation}
\frac{\partial}{\partial t} \psi(\theta,t)= \left(\frac{\partial}{\partial \theta}-\ii H(\theta) \right)\psi(\theta,t),
\end{equation}
then we $\phi(t)=\psi(t,t)$ solves the original time-dependent equations of motion. 
Provided that $\psi(\theta,t)$ is sufficiently well-defined (see \cite{pre-therm} and \cite{howland1974stationary} for details and examples) we may then perform a Fourier transformation, and obtain, 
\begin{equation}
\frac{\partial}{\partial t} \psi(k,t)= \sum_{k'}\left(\tilde{H}(k-k')-\frac{2 \pi k}{P} \right)\psi(k',t).
\end{equation}
The total generator $(H_\omega)_{kk'}:=\tilde{H}(k-k')-k\frac{2 \pi k}{P} $, then may be used in the rest of the theorem. 

{\bf No translation invariance --} Disordered systems are examples of systems that have no translation invariance or any generalized automorphism that may replace it. In order to treat such systems we must reformulate the framework of \cite{Doyon} from the beginning. We begin by introducing an alternative definition of the inner product. We being with the sesquilinear form, 
\begin{equation}
\ave{O,Q}^{loc}_\omega:=\omega(O^\dagger Q)-\omega(O^\dagger)\omega(Q), \qquad O,Q \in \mathfrak{U}_{loc}.
\end{equation}
Note that we do not sum over the sites in the first term. The form is degenerate, i.e. define $\mathcal{N}^{loc}_\omega:=\{\ave{O,O}^{loc}_\omega=0\}$. An example is $O=\one$. We define the quotient $H_\omega=\mathfrak{U}/\mathcal{N}^{loc}_\omega$ that we then Cauchy complete to a Hilbert space $\mathcal{H}_\omega^{loc}$. This is the standard GNS construction \cite{operator1}. 

Define pseudolocalized operator sequences that satisfy $\ave{A_\Lambda,A_\Lambda}_\omega^{loc} \le \gamma $ for some $\gamma >0$. As $A_\Lambda \in \mathfrak{U}_{loc}$ and $\mathfrak{U}_{loc}$ is dense in $\mathcal{H}_\omega^{loc}$ we may always find a pseudolocalized Cauchy sequence (wrt to $||||_{\mathcal{H}_\omega^{loc}}$) $j \mapsto A_j \in \mathfrak{U}$ that converges to $\A^{loc}_\omega(O)=\lim_{j \to \infty}\ave{A_jO}^{loc}_\omega$. The space of such pseudolocalized quantities $\mathfrak{U}^{loc}_\omega$ is in a trivial bijection with $\mathcal{H}_\omega^{loc}$ being its dual. We then may repeat the entire construction of \cite{Doyon}, as well as the present work. 

The only major differences are 1) that the \emph{localized} t-GGE (Theorem~\ref{theorem4}) with pseudolocalized quantities is defined for \emph{any} $\mu_k$ in all dimensions, 2) the Mazur equality is finite for strictly local operators without integrating over space and 3) pseudolocalized states (i.e. pseudolocal states defined via pseudolocalized quantities) are exponentially clustering and this property is preserved under the time evolution. 

Note that we can apply the above pseudolocalized construction to translationally invariant systems, but the construction is not in general useful and the one used in the previous sections is more powerful for such systems. This is because most systems will not have pseudolocalized dynamical symmetries. Two notable exceptions will be, as we shall see, lattice gauge theories and systems with disorder-free localization. 

Note that the local integrals of motion \cite{Serbyn} (or l-bits \cite{Huse}) of many-body localized models fall into the category of pseudolocalized dynamical symmetries. 

\section{Applications}
\label{applications}
The idea for exact solutions will be the following. One needs to identify the pseudolocal quantities the model poses, then one may straightforwardly construct the t-GGE giving the exact solution for the long-time dynamics. 
A practical outline of this procedure, given by the theory here in the previous section, for time-independent closed systems is the following:
\begin{enumerate}
    \item First, identify all potentially pseudolocal operators $A^V_u$ (wrt to a desired initial state $\rho_0$) that satisfy $[H_V,A^V_u]=\lambda_u A^V_u$ for a finite size Hamiltonian $H_V$ of size $V$. 
    \item Propose the solution as $\rho_{tGGE}(t)=\frac{1}{Z} \exp(\sum_u \mu_u e^{\ii \lambda_u t} A_u)$ where $Z=\tr(\exp(\sum_u \mu_u e^{\ii \lambda_u t} A_u))$ is the normalization (or time-dependent partition function) and where we have dropped the size $V$ superscript. If subsets $\{A_u\}$ have the same frequency $\lambda_u$ (i.e. in case of degeneracy), it may be necessary to orthonormalize them. The t-GGE is the correct solution to the long-time dynamics of local observables, as proven in the previous section, provided that one knows all the relevant $A_u$ and 1) that $A_u$ are pseudolocal wrt to the initial state $\rho_0$ and 2) that the state $\rho(t)$ has clustering. This needs to be separately checked as discussed below in steps \ref{pseudolocalcheck} and \ref{clustering}. 
    \item One needs to compute the chemical potentials $\mu_u$ that are fixed by the initial state. This is done by solving for $\mu_u$ the set of equations $\tr(\rho(0)A_u)=\tr(\rho_{tGGE}(0)A_u)$. This is in principle a highly complicated set of non-linear equations requiring numerical solutions, but analytical solutions can be found for wide classes of initial state $\rho_0$. For instance, those that are symmetric or antisymmetric wrt to some discrete symmetry $S$ whereas an $A_u$ is the opposite (antisymmetric or symmetric), e.g. $S\rho(0)S^\dagger=\pm \rho(0)$ and $S A_u S^\dagger=\mp A_u$, then $\mu_u=0$. Similarly, if $[\rho(0),A_u]=\kappa_u A_u$ the solution may be found with some transfer matrix approach.  We will utilize these approaches later in the examples. Likewise, one may perform a high/low $\mu_u$ expansions and then solve a reduced non-linear set.    
    \item  \label{pseudolocalcheck} Define the sequences for increasing $V$, $\tilde{A}^V_u=A^V_u-\ave{A^V_u}_t$ where $\ave{O}_t=\tr(\rho_{tGGE}(t)O)$ for the state at time $t$. We compute how the size of the operators grows with system size $\tr(\rho_{tGGE}(t)(\tilde{A}^V_u)^\dagger \tilde{A}^V_u)$. Due to Cor.~\ref{pseudolocaldyn}, it is sufficient to check only for one value of $t$ or in the initial state $\rho(0)$, which may be a simpler calculation analytically. One approach is: if the set $\{A^V_u\}$ in $\rho_{tGGE}(t)$ is in involution and is simple enough for $Z$ to admit, e.g. a transfer matrix form, then we can compute the relevant expectation values in a way that is standardly done for equilibrium partition functions $\frac{1}{Z}\frac{\partial^2}{\partial\mu_u \partial \mu_u^*} Z =\ave{(A^V_u)^\dagger A^V_u}_t$, etc. In case $\tr(\rho_{tGGE}(t)(\tilde{A}^V_u)^\dagger \tilde{A}^V_u)> C V$ ($C$ does not depend on $V$), then re-scale the sequence $\tilde{A}^V_u \to \tilde{A}^V_u/f(V)$ by some appropriate function $f(V)$ so that $\tr(\rho_{tGGE}(t)(\tilde{A}^V_u)^\dagger \tilde{A}^V_u) \le  C V$. In both cases one needs to check that $\lim_{V \to \infty}\ave{A^V_u O}$ exists for all local observables $O$ (this will likely be the case) and is non-zero for at least some local observables. Due to the re-scaling, overlap with all local observables can be zero and in that case the sequence $A^V_u$ does not correspond to any pseudolocal sequence and must be discarded from the ansatz $\rho_{tGGE}(t)$. It is sufficient to check this using the densities of $A^V_u$ as the local observables, by e.g. using the time-dependent partition function.
    \item \label{clustering} Clustering of $\rho_{tGGE}(t)$ needs to be also checked, i.e. for two local observables on sites $x$($y$) we must have $\lim_{||x-y|| \to \infty}\lim_{V \to \infty}\ave{O_x O_y}_t=\ave{O_x}_t \ave{O_y}_t$. This is again sufficient to do for the local densities of $A_u$ and may be done like in step~\ref{pseudolocalcheck} by computing $Z$. In case $\rho_{tGGE}(t)$ is \emph{not} clustered then it cannot be used as the correct ansatz. This signals formation of long-range order. In that case symmetry breaking of the $A_u$ must be considered as in equilibrium \cite{operator1,operator2}.  
    \item Finally, expectation values of local observables can be computed from the time-dependent partition function similarly to equilibrium, by e.g. adding a small field $\alpha$ corresponding to desired observable $O$, $Z \to Z(\alpha O)$ and then $\ave{O}_t=\frac{d Z(\alpha O)}{d \alpha}|_{\alpha=0}$. This may be done fully analytically provided that $O$ in some sense closes an algebra with the $A_u$ or can be again done in the low/high chemical potential expansion perturbatively in general.   
\end{enumerate}

Fortunately, in certain cases structures known from the previous literature can be used to construct the pseudolocal quantities, and in others they may be straightforwardly identified from the requirements in Cor.~\ref{pseudolocaldyn}. 
\subsection{Projected-local quantities: Quantum many-body scars and embedded Hamiltonians}

Quantum many-body scars \cite{scars,Tom1,Tom2} and embedded Hamiltonians \cite{embed,chandran2022quantum} are two different manifestations of the same underlying pseudolocal algebra as we will now see. First we need a definition. 

\begin{definition}[Projected-local quantity] A projected-local quantity is one satisfying the dynamical symmetry volume growth condition \eqref{volumegrowth} from Cor.~\ref{pseudolocaldyn} for some clustering initial states $\omega$, but not \emph{all} of them. In particular, it does \emph{not} satisfy it for the tracial state $\omega(O)=Tr(O)$, i.e. the infinite temperature state. 
\end{definition}

Specifically, let $A'_\Lambda=\sum_{x \in \Lambda} a_x$ be a pseudolocal sequence and let, 
\begin{equation}
P_\Lambda=\sum_{k,j} \ket{\psi_j}\bra{\psi_k},
\end{equation}
be a generalized projector to the eigenspaces of $H_\Lambda$, i.e.,
\begin{equation}
[H_\Lambda,P_\Lambda]=\nu P_\Lambda. 
\end{equation}
Assume that $A=P_\Lambda^\dagger A' P_\Lambda$ satisfies the condition for a pseudolocal dynamical symmetry from Cor.~\ref{pseudolocaldyn},
with a corresponding pseudolocal initial state with flow $\omega_s$ for which,
\begin{align}
&\omega_s(P_\Lambda O P_\Lambda)=\omega_s(O), \\
&\omega_s(Z(t) O)=\omega_s(O Z(t))=0,  \forall s \in[0,1], O \in \mathfrak{U}_{loc}. \label{scarsinitial}
\end{align}
The long-time dynamics is then given by a t-GGE according to Theorem~\ref{theorem3}. Moreover, the local form of the t-GGE from Theorem~\ref{theorem4} is the unique normal representation provided that,
\begin{equation}
P_\Lambda\rho_t^{tGGE}P_\Lambda=\rho_t^{tGGE}. \label{localproperty}
\end{equation}

{\bf Scarring phase transition --} for the sake of simplicity assume that $H_\Lambda$,$A_\Lambda$,$A^\dagger_\Lambda$ and $[A_\Lambda,A^\dagger_\Lambda]$ generate the $SU(2)$ algebra when acted on by $\omega_s$, i.e. $\omega_s([H_\Lambda,A_\Lambda]-\lambda A_\Lambda)=0$, etc. 
The the local t-GGE has $(A_0)_\Lambda=H_\Lambda$ ($\lambda_0=0$), $(A_1)_\Lambda=A'_\Lambda$ ($\lambda_1=\lambda$), $(A_2)_\Lambda=(A_1^\dagger)_\Lambda$ ($\lambda_2=-\lambda$), $(A_3)_\Lambda=[(A_1)_\Lambda,(A_2)_\Lambda]$ ($\lambda_3=0$). Eq.~\eqref{localproperty} is true only for certain values of $\mu_k$. For other values we cannot use the matrix representation because the state will not have clustering and hence will not be the valid pseudolocal state. 

This is intimately related to the initial state property \eqref{scarsinitial}. Indeed if we modify the flow $\omega_{s(v)}$ continuously such that for some critical $v_{*}$ value of $\omega_{s(v_{*})}$ the property \eqref{scarsinitial} does not hold, then projected-local quantities are no longer pseudolocal wrt to the initial state and are not to be included in the long-time t-GGE. 
For instance, we may vary the inverse temperature of a thermal initial state of another Hamiltonian $H'$ $\omega_{\beta(s)}$ such that for some value of $\beta(s)$ \eqref{scarsinitial} no longer holds. As this can variation of initial temperature can be done continuously and the property \eqref{scarsinitial} is discontinuous, i.e. it either does or does not hold, this demonstrates a novel kind of phase transition between ergodic behaviour and scarrred dynamics. This phase transition is to be contrasted with the standard thermodynamic phase transitions, which occur because the thermal state is no longer the valid representation above some inverse temperature $|\beta|>\beta_{*}$ so that it no longer has a certain symmetry. Here the eigenoperators themselves stop being local. This is consistent with the very recent numerical observation of the dynamical phase diagram of the PXP model \cite{BridgingScars}. Because of the property \eqref{localproperty} the scarring phase is stable to local perturbations as all the terms in the exponent of $\rho_{tGGE}$ are local operators for all values of $t$ and hence the same arguments as for thermodynamic phases can be applied \cite{operator1,operator2}.

In the rest of this subsection, for the sake of notation, we will drop the subscript $\Lambda$ and implicitly deal with the finite system case and its thermodynamic limit. 

{\bf Embedded Hamiltonians --} Let $P_x$ be a set of strictly local projectors, i.e. projectors with finite support $\Lambda_x \subset \Lambda$. Let also $[H',P_x]=0, \forall x$. Embedded Hamiltonians are defined as \cite{embed}, 
\begin{equation}
H=\sum_x P_x h^{0}_x P_x+H', 
\end{equation}
where $h^{0}_x$ is some local Hamiltonian density translated by $x$. 
Clearly, any operator for which $[H',A']=\lambda A'$ may be used to construct a simple projected-local quantity of the form $A=\sum_x (\one-P_x) a_x (\one-P_x)$, $Z(t)=0$. 

{\bf Restricted spectrum generating algebras --} Restricted spectrum generating algebras are formulation of quantum many-body scarred models (Supp of \cite{Buca},\cite{scarsdynsym2,SanjayReview}) for which,
\begin{align}
\begin{split}
&H \ket{\psi_0} = E_0 \ket{\psi_0}, \\
&[H, Q^+] \ket{\psi_k} = \lambda Q^+ \ket{\psi_k}, \forall n, (Q^+)^k \ket{\psi_0} \neq 0.
\end{split}
\end{align}
There are several equivalent formulations \cite{quasisymmetry,scarsdynsym6,scarsPnew} to this one. As discussed in Sec.~\ref{norelation} the existence of such a structure a priori does not imply anything for quantum many-body dynamics in the thermodynamic limit. However, if $Q^+$ is itself pseudolocal, then it clearly defines a restricted local quantity with $P=\sum_{k}\ket{k}\bra{k}$ and $Z(t)=0$. Numerous models studied in the literature satisfy this requirement, e.g. see \cite{SanjayReview} for a review. 

{\bf The PXP model --} Consider the original model of quantum many-body scarring \cite{scars,Olmos_2010}, the one-dimensional PXP model, 
\begin{equation}
H_{PXP} = \sum_x \frac{1}{4}(\one-\sigma^z_{x-1})\sigma^x_x(\one-\sigma^z_{x+1}),   
\end{equation}
where $\sigma^\alpha_x$ is the $\alpha=x,y,z$ Pauli matrix on site $x$. Curiously, this model does not have projected-local dynamical symmetries for all times, but it has them for finite times. To see what this means recall \cite{Tom1} (see also \cite{scarsdynsym1}) that,
\begin{align}
&[H_{PXP},S^{+}_\pi]=S^+_{\pi}+ O_{ZZZ},\\
&S^{+}_\pi=\frac{1}{2}\sum_x  (-1)^x \left[\sigma^z_x-\frac{\ii}{2}(\one-\sigma^z_{x-1})\sigma^y_x (\one-\sigma^z_{x+1})\right], \nonumber\\
&O_{ZZZ}=\sum_x(-1)^x\sigma^z_{x-1}\sigma^z_{x}\sigma^z_{x+1}. \nonumber
\end{align}
Clearly $\tau_t(S^+_\pi)=e^{\ii t}S^+_\pi+Z(t)$, but comparing with the volume growth \eqref{volumegrowth} from Cor.~\eqref{pseudolocaldyn} we get for $t>0$,
\begin{equation}
|\omega(Z^\dagger(t)(e^{\ii t}S^+_\pi+Z(t)))|\le \phi e^{|v| t}|\Lambda|, \label{PXPgrow}
\end{equation}
which we get from the Proof of Theorem~\ref{theorem1} (more specifically, the special case in \cite{Doyon}). In other words, the $\gamma$ in \eqref{volumegrowth} is time-dependent. This does not allow us to define $e^{\ii t}S^+_\pi+Z(t)$ as a pseudolocal dynamical symmetry for all $t$, but fixing some maximal time the conditions are still satisfied. Physically, this reflects the decay of oscillations of local observables \cite{scars}. Moreover, for special initial states, such as the Neel state, the growth in \eqref{PXPgrow} will be smaller than for other initial states. This confirms numerical results on weak ergodicity breaking \cite{scars}. 

Hence the PXP and the other quantum many-body scarred models come from different classes of models, but they both have the projected-local quantities as a common feature explaining physically relevant dynamics of local observables. The existence of such quantities should therefore be taken as defining quantum many-body scars. 

\subsection{Crypto-local quantities: Statistically localized integrals of motion and Hilbert space fragmentation}

In \cite{Fragmentation1,Fragmentation2,fragmentationSanjay,ZnidaricFrag,ArnabFrag2} fragmented models that do not poses (explicitly) quasi-local conservation laws, but do poses finite autocorrelation functions were introduced. This behaviour has been explained through statistically localized integrals of motion (SLIOM) in \cite{SLIOM} and, alternatively, commutant algebras in \cite{SanjayNew}. Using such algebraic structures it was possible to provide Mazur bounds on autocorrelation functions of local observables. However, it remained unclear whether these Mazur bounds saturate and the far-from-equilibrium dynamics of fragmented models remained inaccessible to analytical study. Using the theory developed in the previous section it is possible to give far-from-equilibrium dynamics in terms of the t-GGEs and show that the Mazur bound is saturated. In order to do so we must introduce new types of pseudolocal (pseudolocalized) quantities. 

\begin{definition}[Crypto-locality] Crypto-local quantities are those that satisfy the pseudolocality conditions from Cor.~\ref{pseudolocaldyn}, but cannot be written as manifestly translation invariant sums of local densities (not even with diverging quasi-local support). Likewise, crypto-localized quantities are those that meet the pseudolocalized conditions from Sec.~\ref{generalization}, but cannot be written as manifestly localized objects.
\end{definition}

We will study the prototypical one-dimensional $t-J_z$ model \cite{SanjayNew,SLIOM},
\begin{align}
&H_{t-J_z}=\sum_{x,\sigma\in\{\uparrow, \downarrow\}}{-t_{x, x+1} {\left(d_{x,\sigma} d^\dagger_{x+1,\sigma} + h.c.\right)}} \nonumber \\
&+\sum_{x,\sigma\in\{\uparrow, \downarrow\}} J^z_{x,x+1} S^z_x S^z_{x+1}\nonumber\\
&+ \sum_x{h_x S^z_x + g_x (S^z_x)^2,} \label{tJzmodel}
\end{align}
where $t_{x,x+1}$, $J^z_{x,x+1}$, $h_x$, $g_x$ are arbitrary, and
\begin{align}
&S^z_x=d^\dagger_{x, \uparrow}d_{x, \uparrow} - d^\dagger_{x,\downarrow}d_{x,\downarrow},\\
&d_{x,\sigma} = c_{x,\sigma} \left(1 - c^\dagger_{x,-\sigma} c_{x,-\sigma}\right), 
\end{align}
where $-\sigma:\uparrow(\downarrow) \to \downarrow(\uparrow)$ means taking opposite spin of $\sigma$, and $c^\dagger_{x, \sigma}$ and $c_{x, \sigma}$ are fermionic creation and annihilation operators on site $x$ with spin $\sigma$. 

Consider the "left" and "right" SLIOMs \cite{SLIOM},
\begin{equation}
\mathcal{Q}^{(l,r)}_k = \sum_{x = 1}^{L}\mathcal{P}^{(l,r)}_{k, x} (N^\uparrow_x - N^\downarrow_x), \label{eq:SLIOM}
\end{equation}
where $\mathcal{P}^{(l,r)}_{k, x}$ is the projector onto configurations where the $k$-th charge from the left (right) is on site $x$ and $N_x^\sigma=d^\dagger_{x,\sigma}d_{x,\sigma}$ (see \cite{SLIOM} for details). 

Using these we may construct cryptolocalized and cryptolocal quantities,
\begin{equation} \label{fragsequence}
A_{\vec{\alpha}}=L^\nu\sum_{k,j=l,r} \frac{\alpha^{j}_k}{\left(\sum_{k',j'=l,r}{\alpha^{j'}_{k'}}\right)} \frac{\mathcal{Q}^{(j)}_k}{\omega(\mathcal{Q}^{(j)}_k\mathcal{Q}^{(j)}_k)},
\end{equation}
where $\nu=0,1/4,1/2$. The $\nu=0$ case corresponds to cryptolocalized cases and the other two to cryptolocal. The reader may be surprised that the quantity growing as $\omega(A_{\vec{\alpha}}A_{\vec{\alpha}})\propto L^{1/2}$ is pseudolocal, but it is according to the general definition in Cor.~\ref{pseudolocaldyn}. Note also that we are free to "tune" the sequence \eqref{fragsequence} between a pseudolocal one (extensive) and pseudolocalized by changing $\nu$. If in doing so we promote a sequence that is pseudolocal to a pseudolocalized one, the corresponding pseudolocalized quantity will simply give vanishing functionals $\Aa^{loc}_\omega(O)=0$ for all $O \in \mathfrak{U}_{loc}$. Fragmented models are special because they have crypto-localized quantities wrt to the infinite temperature state, which is directly implied by the present work and by the finite values of the corresponding Mazur bounds identified previously (e.g. Eq. (11) of \cite{SLIOM}). Naturally, the reader may be concerned about the presence of infinitely long (non-local) strings $\mathcal{P}^{(l,r)}_{k, x}$ in the cryptolocal quantities and how their presence will affect the clustering of the corresponding $\rho_{tGGE}$. It will turn out, as discussed in the example below, that most of these strings will cancel and the remaining ones will be subextensive in number (thermodynamically irrelevant) for most initial states. Interestingly, they can be thermodynamically relevant for some initial states with clustering. The corresponding long-time limit will not therefore be a $\rho_{tGGE}(t)$ state containing cryptolocal charges. 

Furthermore these operator sequences may provide Mazur bounds (or equalities according to Cor.~\ref{Mazur}) and t-GGEs thus completing the picture of non-equilibrium dynamics for fragmented models. Other fragmented models (e.g. \cite{PabloNew}) may be treated analogously.

\subsection{Strictly localized quantities: Disorder-free localization and lattice gauge theories}
Now we deal with strictly localized quantities that should be contrasted from cryptolocalized cases associated with fragmentation. 

A prototypical model with disorder-free localization is one with spin-fermion coupling \cite{Dima},
\begin{equation}\label{eq: Hamiltonian}
H_{sf} = -J\sum_{x} \sigma^z_{x,x+1} c^\dag_{x} c_{x}
- h \sum_{x} \sigma^x_{x-1, x} \sigma^x_{x,x+1},
\end{equation}
where the $c_x$ ($c^\dagger_x$) are spinless fermion lowering (raising) operators acting on sites $x$ and $\sigma^\alpha_{x,x+1}$ are spin-1/2 Pauli matrices, as before, acting on the bonds between the sites. 

Very much related to disorder-free localization models are lattice gauge theories \cite{JadDis,LGT2,LGT3}, such as the simple $\mathbb{Z}_2$ lattice gauge theory \cite{LGT1},
\begin{equation}
    H_{\mathbb{Z}_2}=\sum_{x}J\left(a_x^\dagger\sigma^z_{x,x+1}a_{x+1}+h.c.\right)-h\sigma^x_{x,x+1},
\end{equation}
where $a_x$ ($a_x^\dagger$) are \emph{hard-core} bosonic annihilation (creation) operators with $n_x=a_x^\dagger a_x$ representing matter occupation on site $x$. 

Both types of models are characterized by full sets of strictly local (or pseudolocalized) symmetries, $G_x$, where $G_x \in \Lambda_x \subset \Lambda$, i.e. $G_x$ has finite support. 

For instance, the generator of the $\mathbb{Z}_2$ symmetry is
\begin{equation}
G^{\mathbb{Z}_2}_x=(-1)^{n_j}\sigma^x_{x-1,x}\sigma^x_{x,x+1}. 
\end{equation}

Understanding non-equilibrium dynamics of these models has attracted lots of interest recently. Using the theory developed here, exact solutions can be given in terms of t-GGEs given with pseudolocalized (or strictly localized) quantities generated by the corresponding symmetries of the models. More specifically, the full set consists of projectors to the eigenspaces of these generators. Similar holds for non-Abelian lattice gauge theories \cite{JadDis2} in which cases we need to be mindful that the generators close some algebra and then we may use the t-GGE solution. 

Analogous results hold for theories that have fragmentation due to strictly local quantities rather than crypto-localized ones \cite{strictlylocalfrag,LGT3,LGT1,ArnabFragmentation,LOCfrag1,LOCfrag4,LOCfrag5,LOCfrag6}, as well as pseudolocalized ones \cite{Stark1,Thivan,Henrik,Stephen}.

\subsection{Other cases}

Let us briefly discuss other cases for which ansatze similar to the t-GGE have been previously conjectured. The added benefit of the theory from the previous sections is giving the correct forms of the t-GGE and proving that these are the exact solutions. 

{\bf Discrete time crystals in closed systems --} Many-body localized models have been employed for several years for study of discrete time crystals \cite{Sacha_2015,VedikaLazarides,ElseTC}, i.e. many-body systems that display parametric down conversion in the sense of breaking the period of an external drive $T \to nT$. They have been conjectured to go into crypto-equilibrium states \cite{cryptoeq}, that maximize entropy. The present work proves this in the form of the t-GGE state. Moreover, the correct pseudolocalized dynamical symmetries are the l-bits identified in \cite{Floquetlbits1,Floquetlbits2}. 

{\bf Discrete and continuous dissipative time crystals --} Time crystals in locally interacting systems induced or stabilized by dissipation have been studied, both the discrete version (described above) \cite{Chinzei,Chinzei2,Subhajit,Jamir1,Jamir2,Sacha,LazaridesDissipation} and the continuous version for which the time-translation symmetry breaking occurs without any external time-dependent drive, in terms of dynamical symmetries \cite{Buca,BucaJaksch2019,Booker_2020,Carlos,Orazio,liang2020time,MeasurementTC,SymmetryInduced}. The present work shows that the correct form of the long-time limit is the t-GGE containing the dynamical symmetries. 

{\bf Continuous time crystals in isolated systems --} In \cite{Marko1} the t-GGE ansatz has been previously conjectured for the XXZ spin chain containing quasi-local dynamical symmetries. The present work shows that this is the exact solution to the long-time dynamics. 

{\bf Semi-local charges --} Very recently the notion of pseudolocality has been extended to include semi-local operators, i.e. operators whose densities commute with distant operators on one side only \cite{LenartSL1,Mauriziosemilocal} (see also \cite{LenartSL2,HadisehMulti,LenartFolded1,LenardFolded2}). These operators are sums of densities of the form $o^{sl}_x=\lim_{N \to \infty}\prod_{k=-N}^x \sigma^z_k o_x$, where $o_x \in \Lambda_x \subset \Lambda$ is local. Note that the operator contains a \emph{string} of e.g. Pauli $z$ operators. In \cite{LenartSL1} the algebra of quasilocal observables has been extended to include semi-local operators. Although this is perfectly correct, the work above shows that semi-local operators are indeed pseudolocal with respect to specific states. That is they fall into the projecte-local class (an initial state that does not see the Pauli string). 

\section{Examples}
\label{examples}

We will now study explicit examples from the previous section applying the procedure outlined there and go beyond existing techniques by studying general far-from-equilibrium quenches in cases where solutions where only available either from very special initial states for certain observables or near infinite temperature in the linear response regime.
We will compute the (finite frequency) time averaged expectation values of local observables,
\begin{equation}
\overline{\ave{O}}_{\lambda=\kappa}=\lim_{T \to \infty} \frac{1}{T} \int_0^T dt e^{-\ii \kappa t} \ave{O}_t,
\end{equation}
and if not otherwise written $\ave{O}$ will imply the $\lambda=0$ zero frequency case. 

\subsection{Spin-$1$ model with quantum many-body scars}

Here we will look into the scarred spin-$1$ model on a $D$-dimensional hypercubic lattice studied in \cite{Tom2},
\begin{equation}\label{H-XY}
H=\sum_{\langle xy\rangle} \left(S^x_x S^x_y+S^y_x S^y_y\right)+h\sum_x S^z_x+d\sum_x \left(S^{z}_x-2\right)^2,
\end{equation}
where ${\langle xy\rangle}$ means sum over nearest neighbors and $S_x^\alpha$ $(\alpha=x,y,z)$ are spin-$1$ operators on site $x$. The lattice number is $V=N^d$. As shown in \cite{Tom2} the model has quantum many-body scarred eigenstates,
\begin{equation}
\ket{n}= \mathcal{N} \left(J^+\right)^n \ket{-1}, \label{scarredstates}
\end{equation}
where $n=0,\dots,V$, $\ket{-1}$ is the fully polarized down state, and the normalization is $\mathcal{N}=\sqrt{\frac{(V-n)!}{n!V!}}$ and
\begin{equation}\label{eq:Jpm}
J^{\pm}= \frac{1}{2}\sum_x e^{\pm \ii \vec{x}\cdot \vec{\pi}} \left(S^{\pm }_x\right)^2,
\end{equation}
where $\vec{x}$ is the lattice site position vector and $\vec{\pi}$ is a vector of the same dimensions whose all components are $\pi$. The model also has a $U(1)$ symmetry $J^z=\sum_x S^z_x$. Define the projector to the scarred subspace $P=\sum_n \ket{n}\bra{n}$. It is not difficult to see that $[P,J^\alpha]=0$ and that $A_{\pm 1}=f(V) PJ^{\pm}$ fulfill the conditions for simple dynamical symmetries with $H$ from Def.~\ref{definition1}, i.e. $[H,A_{\pm 1}]=\pm 2 h A_{\pm 1}$ and where we anticipate that $f(V)$ is a system size dependent normalization that will be required. Thus, $\lambda_{\pm 1}=\pm 2 h$, $\lambda=2h$. We will consider a simple initial state as an example,
\begin{equation}
\rho(0)=\frac{1}{Z_0} \exp(\mu_0 J^x), \label{scarinitial}
\end{equation}
where $Z_0$ is the normalization. As shown in \cite{Tom2} exact solutions are possible in the case when $\mu_0 \to \pm \infty$ for certain observables. Here we will compute the general case. 
We begin with an ansatz for the t-GGE,
\begin{equation}
\rho_{tGGE}=\frac{1}{Z}\exp(-\beta H+\mu_z J^z+\mu e^{\ii \lambda t} A_1 +h.c.),
\end{equation}
where $Z(t)=\exp(-\beta H+\mu_z J^z+\mu e^{\ii \lambda t} A_1 +h.c.)$ is the time-dependent partition function. For the choice of initial state \eqref{scarinitial} we have $\mu_z=0$. Using the fact that $\rho(0)$ is spin-flip symmetric and $H(d=0)$ is spin-flip antisymmetric we immediately get,
\begin{equation}
\ave{H}_0=\tr(\rho(0)H)=d V \left(\frac{1}{3}-\frac{1}{2 \cosh (2 | \mu_0 | )+1}\right).
\end{equation}
We need to find for what $\mu_0$ the $A_{\pm 1}$ are pseudolocal, however, according to the procedure outlined in the beginning of Sec~\ref{applications}. In order to do so, without loss of generality, but for the purposes of easing orthonormalization we will set $d=0$.  Because $[P,J^\alpha]=0$ and $P^2=P$ we get that $\mu_0=\mu$. The partition function for the initial state can be immediately computed,
\begin{equation}
Z_0=(2 \cosh (\mu_0 )+1)^V.
\end{equation}
Likewise, as discussed in the Appendix, the partition function of the t-GGE can also be computed,
\begin{align}
&Z=\text{csch}(\mu  \cos (2 h t)) \sinh (\mu  (V+1) \cos (2 h t))-V \nonumber\\
&+\sinh (V \log (3))+\cosh (V \log (3))-1
\end{align}
As shown in the Appendix, it is straightforward to compute that in the initial state (note again that due to Cor.~\ref{pseudolocaldyn} it is sufficient to check pseudolocality for the initial state),
\begin{align}
&\ave{A_1A_{-1}}_0= \nonumber\\
&-\frac{f(V)^2}{4Z} \text{csch}^3(\mu_0 ) \left[2 \left(V^2+2 V-1\right) \sinh (\mu_0  [V+1]) \right. \nonumber\\
&\left.-(V+1) \left\{V \sinh (\mu_0  [V+3])+(V+2) \sinh (\mu_0 [V-1])\right\}\right], \\
&\ave{A_1}_0= \nonumber\\
&\frac{f(V)}{Z}\text{csch}^2(\mu_0 ) [V \sinh (\mu_0  (V+2))-(V+2) \sinh (\mu_0  V)],\\
&\ave{J^+_x}_0=\frac{2 \sinh (\mu_0 )}{2 \cosh (\mu_0 )+1}
\end{align}
Our method will be to fix $f(V)$ by demanding that $\lim_{V \to \infty}\ave{\tilde{A}_1 O_x}_0$ is not zero for at least some local $O_x$. The most convenient choice will be the density of $J^\pm$, i.e. $J^\pm_x=e^{\pm \vec{x}\cdot\vec{\pi}}(S^\pm_x)^2$ (because the $\tilde{A}_{\pm1}$ have overlap with it). We find that,
\begin{equation}
\ave{A_1 J^-_x}_0=\frac{1}{V}\ave{A_1A_{-1}}_0-\ave{A_1}_0\ave{J^-_x}_0,
\end{equation}
where we used the fact that the t-GGE and $J^\pm$ are Bloch translationally invariant with momentum $\pi$. Now we fix $f(V)$ by demanding that $\ave{A_1 J^-_x}_0=1$ (the actual value of the constant does not matter, only that it is finite) and compute $\ave{\tilde{A}_1\tilde{A}_{-1}}_0$. For $A_{\pm1}$ to be pseudolocal $\ave{\tilde{A}_1\tilde{A}_{-1}}_0  \le C V$, where $C$ does not depend on $V$. We find, by expanding in $1/V$, that for finite $\mu$ and large $V$, $\ave{\tilde{A}_1\tilde{A}_{-1}}_0 \approx e^{4 |\mu| +V (\log (\cosh (|\mu| )+1)-|\mu| )}$. For $\mu \to \pm \infty$ we find, on the other hand, that the $A_{\pm 1}$ are pseudolocal as $\ave{\tilde{A}_1\tilde{A}_{-1}}_0 =V$ for diverging $\mu$. Therefore, there is no scarring phase transition for the initial state chosen here for finite $\mu$ in the thermodynamic limit. However, as we will see there is similar behaviour to a phase transition for finite system size. Persistent oscillations (i.e. non-stationarity) in local observables is only present for diverging $\mu$. However, even though the theory given in this paper is strictly speaking for thermodynamically large systems, we may still gain insight into finite size behaviour of the models. First note that the growth of $\ave{\tilde{A}_1\tilde{A}_{-1}}_0/V$ can be almost negligible for a given $\mu$ up to some system size after which it grows quickly. To illustrate this we plot $\ave{\tilde{A}_1\tilde{A}_{-1}}_0/V$ in Fig.~\ref{fig:scarsgrowth}.

\begin{figure}[ht]
   \centering
    \includegraphics[width=\columnwidth]{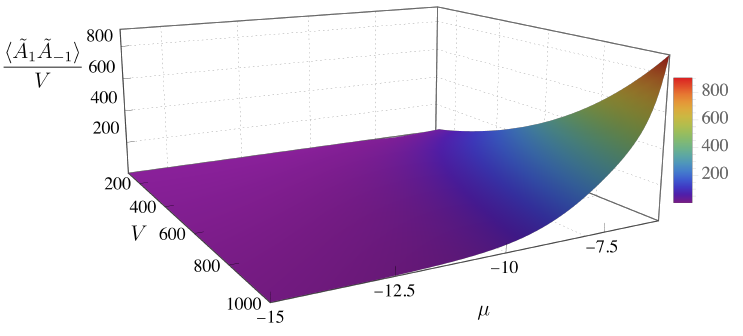}
    \caption{~\label{fig:scarsgrowth} The growth of $\ave{\tilde{A}_1\tilde{A}_{-1}}/V$ with $1/\mu$ and $V$ which can be quite slow indicating that oscillations can persist beyond the scarring phase for very large systems and times even for finite $\mu$. After some system size that depends on $\mu$, $\ave{\tilde{A}_1\tilde{A}_{-1}}/V$ grows exponentially with $V$ and the oscillations are no longer present after that system size. Note that results are the same for $\mu \to -\mu$.}
\end{figure}

It is known that translationally invariant systems with translationally invariant initial states can be expected to reach their asymptotic dynamics in times that $t_{relax}=\mathcal{O}(1)$ \cite{ETHReview}. This dynamics persists at least until finite size effects for local observables in the bulk start at times that are $t^*=\mathcal{O}(V)$. This is due to finite Lieb-Robinsons velocity (i.e. it takes at least time $V$ for quantum information to reach the end of the system and come back to the bulk before an observable there can see that the system is finite). Hence even a finite size system can be expected to be described by a t-GGE for $t_{relax} \ll t \ll t^{*}$. The growth of $\ave{\tilde{A}_1\tilde{A}_{-1}}_0/V$ essentially represents that chemical potentials in the t-GGEs must be rescaled in order for the expectation values of $A_{\pm1}$ to be equal in the initial state $\rho(0)$ and the $\rho_{tGGE}(0)$, i.e. $\mu=\frac{V}{e^{4 |\mu_0| +V (\log (\cosh (|\mu_0| )+1)-|\mu_0| )}} \mu_0$. For a given chemical potential $\mu_0$ this stays almost constant and then after reaching an almost critical system size $V$ decays abruptly and hence so does the contribution of $A_{\pm 1}$ to the expectation values of local observables at finite frequency $\lambda$. Interestingly this decay will not be at all visible for either numerical or experimental simulations up to some large system size $V^*$ and corresponding long time (cf. with the PXP discussions in the previous section). We will show this later explicitly for local observables.

 We may now calculate zero and finite frequency expectation values of local observables that have overlap with $H$ (e.g. $S^z_x$), straightforwardly, and, using the same techniques as before, those that commute with $J^x:=\frac{1}{2}(J^+J^-)$ (e.g. $J^x_x$),
\begin{align}
J^x_x(t)=\frac{1}{Z}\sum_{k=0}^{\lfloor \frac{V}{2} \rfloor} 2(V-k)\sinh[(V-k)\cos(\lambda t) \mu]
\end{align}
In Fig.~\ref{figurescars2} we illustrate finite frequency expectation values and show that we can reproduce the known exact solution for $\mu \to \pm \infty$ of \cite{Tom2}. The decay of the finite frequency amplitudes is doubly exponential with system size for finite $\mu$ (c.f. the growth of the pseudolocal dynamical symmetries in Fig.~\ref{fig:scarsgrowth}). Essentially, the system behaves as if it were in the scarred phase (with finite frequency amplitude close to 1) and then abruptly decays at some almost critical value of system size. Physically, this happens because for a given finite $\mu$ some of the initial state is not contained in the ground state of $J^x$ (which is inside the scarred subspace) and the proportion of the state that is not in the scarred subspace grows with $V$ which at some value is large enough to lead to exponential growth with system size of the (previously) pseudolocal dynamical symmetry. 
\begin{figure}[ht]
   \centering
    \includegraphics[width=0.7\columnwidth]{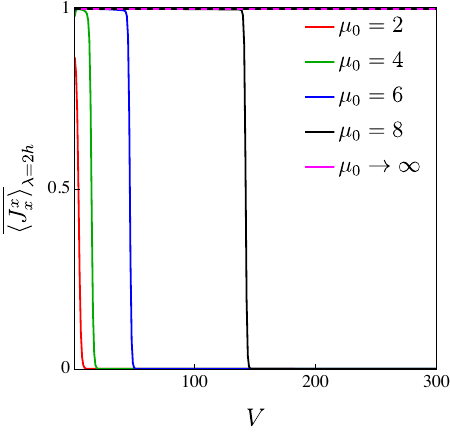}
    \caption{~\label{figurescars2} Finite frequency averages of $J^x_x$. For finite $\mu$ the results are valid at times $ 1 \ll t \ll V$. The magenta line at infinite $\mu_0$ agrees with the exact solution (dashed black line) from previous literature \cite{Tom2}. We see that the oscillation amplitudes display an almost discontinuous dependence on system size $V$ - they are constant up to some "critical" system size after which they decay abruptly to $0$. In other words, for a fixed system size, the system will be effectively in the scarred phase up to some value of the initial chemical potential after which the oscillations abruptly decay.}
\end{figure}

To compute the zero frequency values we will need the inverse temperature $\beta$ for $H$. This can be accomplished for small $\mu_0$ ($\beta$) by means of high temperature expansion (truncating to the second order) we obtain in the thermodynamic limit,
\begin{equation}
\beta=\frac{2 d \sinh ^2\left(\left| \mu _0\right| \right)}{\left(d^2+3 \left(h^2+4\right)\right) \left(2 \cosh \left(2 \left| \mu _0\right| \right)+1\right)}.
\end{equation}
We may now easily compute expectation values of observables that have overlap with $H$. This is done in Fig.~\ref{figurescars3} a) (b)) for $S^z_x$ ($S^y_xS^y_{x+1}$) as an example. The initial chemical potential is $\mu_0=0.1$. Note that the expecation values have a non-linear dependence on $h$ and $d$, which implies that there is preference towards antiferromagnetic ordering even close to infinite temperature in the system, i.e. the induced magnetic field in the system is not maximized by maximizing the external fields.  

\begin{figure}[ht]
   \centering
    \includegraphics[width=\columnwidth]{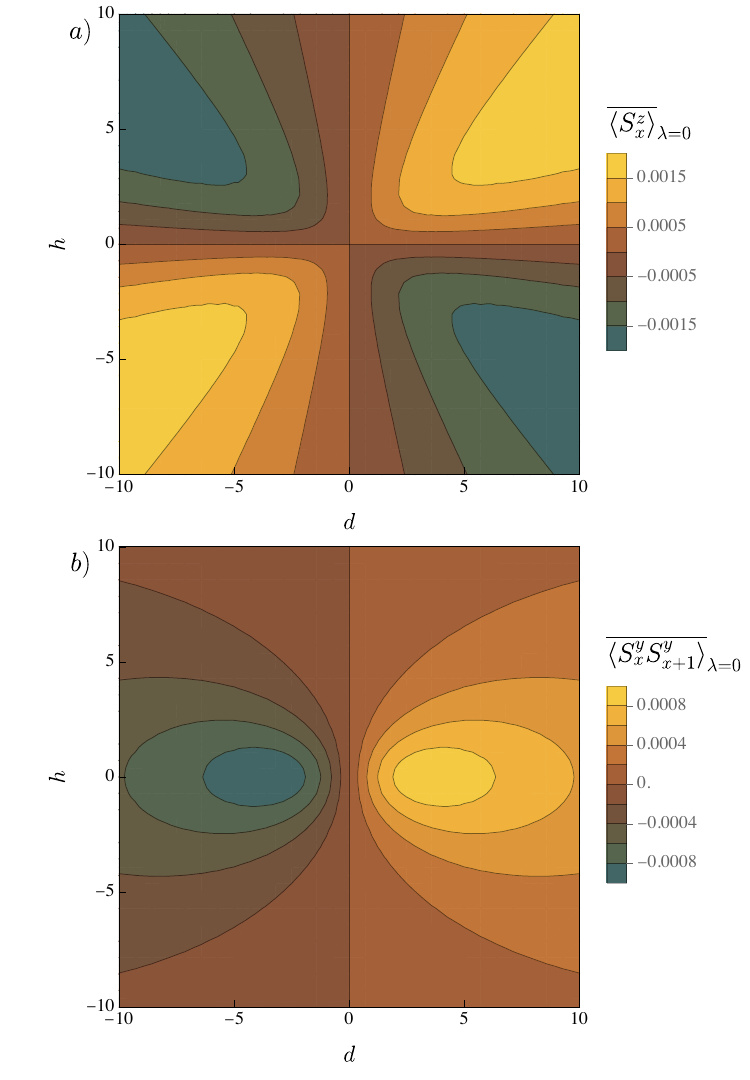}
    \caption{~\label{figurescars3} The time-averaged value of $S^z_x$ ($S^y_xS^y_{x+1}$) are given in subfigure a) (subfigure b)) as a function of the external fields $h,d$. We are close to infinite temperature because the chemical potential in the initial state is $\mu_0=0$. Hence, the values are small, but there is still a manifest non-linear dependence.}
\end{figure}

\subsection{$t-J_z$ model with fragmentation}

Previous approaches \cite{SLIOM,fragmentationSanjay} could only analytically treat the $t-J_z$ model in linear response and at infinite temperature. We will now compute a far-from-equilibrium quench case. 

As discussed in the previous Sec.~\ref{applications} the $t-J_z$ chain has cryptolocalized quantities that may be constructed from the SLIOMs. We will assume that those and the Hamiltonian are the only relevant pseudolocal quantities of the model. Remembering that the t-GGE is only the effective state governing the dynamics of local observables, we will, for the sake of simplicity, focus only on the left half of the chain and hence can only consider the left SLIOMs that we now call $A_k:=\mathcal{Q}^{(l)}_k$, for $k=1,\ldots$. Therefore, the conjectured t-GGE ansatz giving the long-time (equivalently, zero frequency) dynamics contains $A_0=H_{t-J_z}$, the total spin $A_{-1}=S^z=\sum_x S^z_x$ and the SLIOMs for $k=1,\ldots$. Naturally, $\mu_k=0$ for $\lambda_k \neq 0$ as there are no finite frequency pseudolocal dynamical symmetries. 
Let us write the basis for one-site as $\ket{\uparrow\downarrow,\uparrow,\downarrow,vac}$, where the arrows denote the spin of the fermions on that site. As an example let us take the following far-from-equilibrium initial product state,
\begin{align}
&\rho(0)= \nonumber\\
&\mathcal{N}\prod_{\otimes x=1}^{n}\ket{\alpha_{x},\beta_{x},\gamma_{x},0}\bra{\alpha_{x},\beta_{x},\beta_{x},0} \otimes \one_{N-n-1},
\end{align}
where we demand that $\gamma_x=\beta_x$ if $x$ is an even site and $\mathcal{N}$ is the normalization. The state is thus a general pure product state with singlets and doublons on sites $1\ldots n$ and an infinite temperature state (identity) on the rest of the sites. For the sake of convenience, let us set $h_x=g_x=0$ in \eqref{tJzmodel} (the external on-site field does not influence the physics significantly). The initial state $\rho(0)$ is parity antisymmetric wrt to spin-flip while $H_{t-J_z}$ is parity symmetric. Hence, $\ave{H_{t-J_z}}_0=\beta=\lambda_0=0$ and the t-GGE does not contain the Hamiltonian. Moreover, for any finite $n$ the expectation value of the extensive operator $S^z$ is finite, but the expectation value of $S^z$ will be extensive in the t-GGE for any finite $\mu_{-1}$, hence $\mu_{-1}=0$. Thus the t-GGE contains only the cryptolocalized quantities coming from the SLIOMs. The SLIOMs mutually commute and are diagonal in the particle number basis. As discussed in the Appendix, it is thus a matter of straightforward combinatorics to calcuate the partition function,
\begin{equation}
Z=2^N \left(1+\sum _{k=1}^N \binom{N}{k} \prod _{j=1}^k \cosh (\mu_j)\right). \label{tjparitionfunction}
\end{equation}
We may now show that the SLIOMs are pseudolocal by computing their norm wrt to the t-GGE. It is sufficient to check for large system size $N$,
\begin{equation}
\ave{\tilde{A_j}\tilde{A_j}}=\frac{1}{Z}\frac{d^2 }{d \mu_j^2}Z-\left(\frac{1}{Z}\frac{d }{d \mu_j}Z\right)^2,
\end{equation}
which for sufficiently large $N$ is $\ave{\tilde{A_j}\tilde{A_j}} = C_j +\mathcal{O}(1/N)$, where $0<C_j\le 1$ are constants independent of system size. 

By introducing small fields in the t-GGE, we may calculate expectation values of local observables $O_x$ in the t-GGE,
\begin{equation}
Z(\alpha_1,\alpha_2)=\tr(\sum_j \mu_j A_j +\alpha_1 O_x+\alpha_2 O_y),
\end{equation}
and hence,
\begin{align}
&\ave{O_x}=\frac{d}{d\alpha_1} Z(\alpha_1,\alpha_2)|_{\alpha_{1,2}=0} \\
&\ave{O_x O_y}=\frac{\partial^2}{\partial\alpha_1\partial\alpha_2} Z(\alpha)|_{\alpha_{1,2}=0} 
\end{align}
For example, for the diagonal and commuting observables $Z_x:=N^\uparrow_x-N^\downarrow_x$ and $N_x=N^\uparrow_x+N^\downarrow_x-2$ we have,
\begin{align}
&\ave{Z_1}=\frac{2^N}{Z} \sinh (\mu_1) \left(1+\sum _{k=2}^N \binom{N-1}{k-1} \prod _{j=2}^k \cosh (\mu_j)\right) \\
&\ave{N_x}=\frac{2^N}{Z} \left(1+\sum _{k=1}^N \binom{N-1}{k} \prod _{j=1}^k \cosh (\mu_j)\right) \\
&\ave{N_x N_y}=2 \ave{N_x}
\end{align}
In order to finish proof that the SLIOMs are pseudolocalized, we must compute their overlap with local observables and show that it is non-vanishing for at least some. It is sufficient to consider $\ave{Z_1\tilde{A_j}}$ for large enough $N$. It is not difficult to verify that this is finite $\ave{Z_1\tilde{A_j}}$ for finite $j$. Hence, the SLIOMs define crypto-localized quantities and can potentially go into the t-GGE. However, the SLIOMs have strings of operators of the type $\prod_{y=1}^x (N_y+2)Z_x$. But, as is discussed in the Appendix, these strings can contribute sub-extensively to the connected correlator and hence are thermodynamically irrelevant.

Finally, we need to show clustering of the t-GGE state itself. We find that the state is not clustered for all $\mu_k$, even though the SLIOM quantities are pseudolocalized for all $\mu_k$. This signals a fragmentation phase transition for which the pseudolocal quantities are always the same, but the state acquires long-range order and cannot be represented as a matrix exponential. This is similar to thermodynamic phase transitions \cite{operator1}, and should be contrasted to the scarring phase transition above for which the pseudolocal quantities themselves stopped being pseudolocal for certain values of the chemical potentials. As the non-local strings only have overlap with $N_x$ it is sufficient to check the long-range connected correlator for $|\ave{N_xN_y}-\ave{N_x}\ave{N_y}|$. 
We illustrate this with an example by parameterizing the chemical potentials as $\prod _{j=1}^k \cosh (\mu_j)\to x^{a k}+1$, where, for the purposes of keeping the result valid for the simpler case of $\mu_{-1}=0$, the chemical potentials $\mu_j$ should have a cutoff such that $\mu_j=0$ for some very large $j>\kappa$, but with $\kappa$ still being much smaller than $N$ when taking the thermodynamic limit. We then find that (as may be verified by means of e.g. computer algebra),
\begin{align}
&\lim_{N \to \infty}|\ave{N_xN_y}-\ave{N_x}\ave{N_y}|=\nonumber\\
&\begin{cases}
    \frac{\left(x^a-1\right) \left(3 x^a+1\right)}{4 \left(x^a+1\right)^2},& \text{if } \log \left(x^a+1\right)>\log (2)\\
    0,              & \text{otherwise}.
\end{cases}
\end{align}

Note that, curiously, unlike thermodynamic phase transition the order is either completely non-local (the same for all $x,y$) or completely absent. This proves the fragmentation phase transition in the model and the phase diagram is given in Fig.~\ref{figureFRAG} a).
\begin{figure}[ht]
   \centering
    \includegraphics[width=\columnwidth]{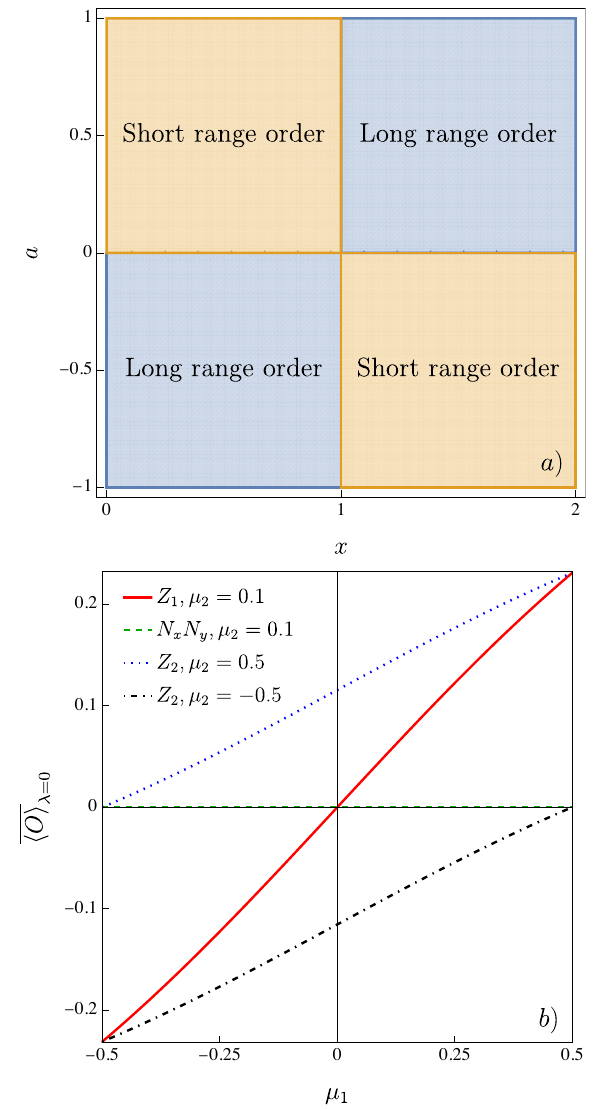}
    \caption{~\label{figureFRAG} a) The fragmentation phase diagram showing the phases where the strings are thermodynamically irrelevant (short-range order) and where they are not. b) Time average of certain local observables for various values of the initial state parameter (see \eqref{chemicalinitial} for the initial state parameters). }
\end{figure}

Finally, of course we must relate the chemical potentials to the initial state. This may be done for arbitrary choices of the initial state parameters for large system sizes numerically, but in order to give closed form expressions, we consider the case when $|\beta_x|^2-|\gamma_x|^2$ is small. In particular for, e.g. $\alpha_x=\beta_x=\gamma_x=0$ for $x>3$, we have that,
\begin{align}
&\mu_1=\frac{| \alpha_1| ^2 \left(| \beta_2| ^2-| \gamma_2| ^2\right)}{\left(| \alpha_1| ^2+2 | \beta_1| ^2\right) \left(| \alpha_2| ^2+| \beta_2| ^2+| \gamma_2| ^2\right)} \nonumber \\
&\mu_2=\frac{2 | \beta_1| ^2}{| \alpha_1| ^2} \mu_1, \label{chemicalinitial}
\end{align}
and $\mu_j=0$ for $j >2$. Thus the time-averages of local observables are quite complicated functions of the initial state parameters even for this simple product state. We illustrate some of them in Fig.~\ref{figureFRAG} b).

\section{Conclusion}
\label{conclusion}
The main goal of non-equilibrium quantum many-body theoretical physics is computing the dynamics of systems out-of-equilibrium. Locality is what crucially unifies dynamical properties of quantum many-body systems providing a framework applicable to isolated, driven and dissipative quantum many-body systems. The theory presented here allows for exact solutions of quantum many-body dynamics for all locally interacting systems with finite local degrees of freedom on hypercubic lattices of arbitrary dimensions. This constitutes a very wide class of quantum many-body systems and includes paradigmatic models such as spin models and fermionic lattice models. The theory provides the solution in terms of a time-dependent generalized Gibbs ensemble (t-GGE) and does not rely on integrability. 

The basic "recipe" is the following. If one finds evidence of non-ergodicity in a quantum many-body system (either at finite or zero frequency), the present work shows that it must be due to pseudolocal dynamical symmetries. Provided one can then identify these pseudolocal dynamical symmetries, one may find the solution to the dynamics of local observables immediately as a t-GGE. The chemical potentials in the t-GGE are set by the initial states. Conversely, proving absence of any such symmetry (apart from e.g. the Hamiltonian or some $U(1)$ charge) immediately proves ergodicity at zero frequency. The theory thus proves both weak eigenstate thermalization hypothesis in dynamical form and saturation of the Mazur bound. 

More generally, the theory is an important step towards solving the main goal of computing non-equilibrium quantum many-body dynamics because it does so for wide classes of locally interacting systems. In the future it can be applied to all such wide classes of systems and has the potential to provide analytical solutions were there were previously none. 

\subsection{Open problems}
The work presented here opens numerous possible research directions. I list only a few below.  
\begin{itemize}
\item {\bf New forms of non-ergodicity--} The complete theory presented here allows not only for study of known forms of non-ergodicity, it can help in classifying and generating models with novel types of non-ergodicity. For instance, Hilbert space fragmentation is identified with the existence of crypto-local conservation laws in such models and quantum many-body scars with projected-local dynamical symmetries. Can one have a model with quantities that are of both types, i.e. a projected crypto-local dynamical symmetry? This would imply a fragmented scar, i.e. dynamics that has local oscillating memory for certain initial states. 

\item {\bf Constructing transient dynamical symmetries --} Transient dynamical symmetries, identified here, may play an important role in the type of transport a system has, e.g. diffusive, super-diffusive, etc. To the best of my knowledge transient dynamical symmetries dictating the finite time dynamics have not been identified. However, superficially similar structures are known. For instance, in few-body bosonic models one may have quasi-normal modes (e.g. \cite{CatQUBIT}), i.e. metastable decaying eigenmodes of the Hamiltonian with complex energy (that is not self-adjoint due to being unbounded) \cite{spectraltheory}. One may attempt to find transient dynamical symmetries by adapting the procedure of \cite{Prosennumerical,Thivan,SarkarBuca} to imaginary frequency. Likewise, in \cite{Zala} and \cite{Fagotti_2014} prethermalization was studied by having a prethermal Hamiltonian as a transient conservation law. This could be a starting point for a theory of prethermalization based on transient dynamical symmetries. 

\item {\bf Scarring phase transitions --} The existence of a novel phase transition between weak ergodicity breaking (non-stationary dynamics) and ergodicity has been proven here. It is distinct from thermodynamic phase transitions because it happens because projected-local operators stop being pseudolocal when one smoothly various the chemical potential rather than being a discontinuity in the equilibrium state itself. What is the nature of this phase transition in terms of e.g. universality classes? 

 \item {\bf Long-range order in the fragmentation phase transitions --} The fragmentation phase transition introduced here and shown for the $t-J_z$ chain has only been studied in the short-range correlated phase. Studying the long-range correlated phase will require going beyond the matrix representation of the t-GGE and will entail generalizing techniques that are used for symmetry breaking in equilibrium \cite{operator1}, as it remains unclear how to treat the crypto-localized quantities responsible for the fragmentation phase transition due to their non-local strings. 

\item {\bf Generalized hydrodynamics for scars and fragmentation --} Identifying scars and fragmentation in terms of pseudolocal quantities opens the possibility of constructing a generalized hydrodynamics theory for the integrable forms of these models. For instance, recent work of \cite{Doyon2,Doyon3} using dynamical symmetries for hydrodynamics could be combined here with the integrable limit of a constrained scarred Hubbard model \cite{scarsdynsym2}. 

\item {\bf Implications for quantum information processing --} Many quantum algorithms can understood as local many body dynamics acting on a system (i.e. a collection of qubits). For instance,  the quantum Fourier transform is a local algorithm with the end result of the computation being stored locally. The present work fully classifies the long-time limit of such systems. Can it be used to identify possible quantum error correction algorithms? It could also be conceivably used to strengthen the quantum threshold theorem \cite{threshold}. Indeed, locality in the form of the Lieb-Robinson bound (crucial for the theory here) has been recently employed to study error correction and entanglement generation \cite{LocalityEnt,LocalityEnt2}. 

\item {\bf t-GGEs for deep thermalization --} As the theory here applies to dissipative systems and the long-time limit is given in terms of t-GGEs, they could offer an alternative way to study recently introduced \emph{deep thermalization} i.e. thermalization induced by projective measurements \cite{DeepTherm1,DeepTherm2,Claeys2022emergentquantum,Claeys2023universalityin}. What is the connection between t-GGEs in that case and the recently introduced deep GGE? 

\item {\bf Including unbounded densities --} The theory presented here works when the local degrees of freedom are bounded. What happens for e.g. bosons on a lattice, for which the densities can be infinite? Lieb-Robinson bounds hold for these systems, too \cite{LiebRobinsonBosons} and hence, one may conceivably upgrade the theory from this paper to account for them. 

\item {\bf Long-range interactions --} The theory here crucially relies on Lieb-Robinson bounds for local dynamics. These bounds have been extended for long-range interacting systems \cite{Matsuta_2016}. Could these bounds be used to define dynamically relevant long-range quantities instead of pseudolocal ones? Alternatively, the above discussed connections with lattice gauge theories offers another possible way to treat long-range interactions (e.g. \cite{Passarelli_2022}) based on the theory in this paper - one may introduce unphysical gauge degrees of freedom in a local model and then, assuming that they are very fast, adiabatically eliminate them \cite{paulsen_2003} reducing the local problem (which is treatable) to a long-range model (which one wants to study). This would allow for treatment of non-ergodic dynamics in long-range models (e.g. \cite{Esslinger,Esslinger2,dissipativeTCobs,Fazio,PTDisTC,Lesanovsky,Federico1,Federico2,Esslinger3,Francois1,Francois2,Francois3})

\item {\bf Pseudolocalized quantitites for proving many-body localization --} Many-body localization (MBL) in disordered systems has been proven under certain assumption on the spectrum of these systems \cite{Imbrie}. It is curious, as shown in this paper, that if one drops any notion of translational invariance (i.e. including disorder), then the many-body dynamics must be based on pseudolocalized quantities. Of course, these quantities contain precisely the l-bits of MBL \cite{PapicMBL}. Can this approach be formalized and used to prove MBL without any assumptions? 

\item {\bf Towards a non-equilibrium Landau theory --} The t-GGE introduced here is a time-dependent version of the Gibbs ensemble. Can this similarity be exploited to formulate a non-equilibrium Landau theory, and the corresponding free energies, for strongly interacting systems out-of-equilibrium? 

\item {\bf Entropy oscillations --} The present work deals purely with the dynamics of local observables.  Can the same framework be upgraded to study entanglement entropy dynamics and other quantities \cite{LOCfrag2,strictlylocalfrag,Olalla1,Olalla2,JamirDPT,Kormos_2016}? 

\item {\bf Quantum hydrodynamics --} The transient pseudolocal quantities could be a starting point for a rigorous framework of quantum hydrodynamics beyond integrable models. Once these quantities are identified the kinds of transport they imply should follow immediately. 
\end{itemize}

\begin{acknowledgments}
I thank W. De Roeck for in depth discussions and useful remarks and V. Juki\'{c} Bu\v{c}a for assistance with Fig. 1 and 2. I am grateful to T. Iadecola, H. Katsura,  H. Moriya for useful feedback on the manuscript.  This work was supported by a research grant (42085) from VILLUM FONDEN, by the EPSRC programme grant EP/P009565/1, and the EPSRC National Quantum Technology Hub in Networked Quantum Information Technology (EP/M013243/1).  
\end{acknowledgments}

\appendix
\section{Proofs}
Here I give proofs of the statements in the main text. 
\begin{proof}[Proof of Theorem~\ref{theorem1}]
We begin by noting that the Lieb-Robinson locality relation \eqref{locality} can be weakened. Namely for some $\mu>0$ and $\phi>0$,
\begin{equation}
||\tau_t(O)-(\tau_t(O))_\Lambda||\le \phi |O|||O|| \exp(-\mu \Delta+v |t|), \label{weakerlocality1}
\end{equation}
which follows directly from faster exponential growth than polynomial growth and $|O|$ is the size of the support of the operator $O$.  

Moreover, we can use the contractivity $||\tau_t(O)|| \le ||O||$ to get that
\begin{equation}
||(\tau_t(O))_\Lambda||\le (1+\phi |O| \exp(v|t|))||A||, \label{weakerlocality2}
\end{equation}
By pseudolocality of $\omega$ it is $p-$clustering \cite{Doyon}, i.e. there exist $\nu,a>0$ such that for every $\ell$, 
\begin{equation}
|(O,Q)_\omega|:=|\omega(O,Q)-\omega(O)\omega(Q)|\le\frac{\nu\ell^a||O||||Q||}{{\rm dist}(O,Q)^{p}},
\end{equation}
for some $p>D$. Likewise, the same holds for the flow $\omega_s$ for the same parameters $p,\nu,\ell,a$. Using this, \eqref{weakerlocality1} and \eqref{weakerlocality2} we can proceed along the same lines as the Proof of Theorem 6.3 of \cite{Doyon} to conclude that there exist some $a_1,\nu_1$ such that,
\begin{equation}
|(\tau_t(O)_\Lambda,Q)_\omega| \le \frac{\nu_1\ell^{a_1}||O||||Q||}{{\rm dist}(O,Q)^{q}}, 
\end{equation}
 and some some $a_2,\nu_2$
\begin{equation}
|(\tau_t(O)_\Lambda,\tau_t(Q)_{\Lambda'})_\omega| \le \frac{\nu_2\ell^{a_2}||O||||Q||}{{\rm dist}(O,Q)^{q}},
\end{equation}
$\forall q<p$. However, this by itself is not enough to show that $
\omega_t$ is $q$-clustering because the map is in general dissipative, 
\begin{equation}
\tau_t(O^\dagger Q) \neq \tau_t(O^\dagger)\tau_t(Q). \label{dissipative} 
\end{equation}
However, we may use another result based on the Lieb-Robinson bound obtained in Corollary 1 of \cite{KastoryanoEisert} again in weaker form. Namely, there exist some $C'>0$ and $\mu'>0$ such that, 
\begin{equation}
||\tau_t(O^\dagger Q)-\tau_t(O^\dagger)\tau_t(Q)||\le C' ||O||||Q||e^{v|t|-\mu'{\rm dist}(O,Q)}. \label{boundKE}
\end{equation}
We also have, 
\begin{equation}
(O,Q)_{\omega\circ\tau_t}=\omega(\tau_t(OQ)-\tau_t(O)\tau_t(Q))-(\tau_t(O),\tau_t(Q))_\omega. \label{relation}
\end{equation}
Therefore, 
\begin{equation}
|(O,Q)_{\omega\circ\tau_t}|\le ||\tau_t(OQ)-\tau_t(O)\tau_t(Q)||+|(\tau_t(O),\tau_t(Q))_\omega|.
\end{equation}
Clearly we can bound the decaying exponential in \eqref{boundKE} by some (time-dependent constant $C''$) for any finite $q>0$,
\begin{equation}
||\tau_t(OQ)-\tau_t(O)\tau_t(Q)||\le \frac{C''\ell^{a_2}||O||||Q||}{{\rm dist}(O,Q)^{q}}.
\end{equation}
From this the q-clustering clustering of $|(O,Q)_{\omega\circ\tau_t}|$ immediately follows. This proves the first point of the theorem. 
The second follows from the Proof of Theorem 6.5 of \cite{Doyon} if we observe that, 
\begin{equation}
||O||_{\mathcal{H}_{\omega \circ \tau_t}}^2 =(O^\dagger,O)_{\omega\circ\tau_t}\ge |(\tau_t(O^\dagger),\tau_t(O))_\omega|-2 ||O||^2,
\end{equation}
which we find from the dissipative property of the map and \eqref{relation}. 
\end{proof}

\begin{proof}[Proof of Theorem~\ref{theorem2}]
From the Theorem~\ref{theorem1} the state is of the form, 
\begin{equation}
\omega_{s,t}(O)=\omega_{0,t}(O)+\int_0^s du \Aa_{u,t}(O), \forall t. \label{pseudolocaltime}
\end{equation}
Hence we may write using the equations of motion \eqref{eqmotion}, 
\begin{equation}
\frac{d}{dt}\omega_{s,t}(O)=\int_0^s du \Aa_{u,t}(\mathcal{L}(O))=\int_0^s du \frac{d}{dt}\Aa_{u,t}(O), 
\end{equation}
where we used $\omega_0(\mathcal{L}(O))=0$ for the tracial state $\omega_{0,t}$. This implies, 
\begin{equation}
\int_I du \left(\Aa_{u,t}(O)-\frac{d}{dt}\Aa_{u,t}(\mathcal{L}(O))\right)=0, \label{eqmotionpseudo1}
\end{equation}
for every open interval $I \subset [0,1]$. By similar arguments as in the proof of Theorem 6.6 of \cite{Doyon} we conclude that, 
\begin{equation}
\frac{d}{dt}\Aa_{u,t}(O)=\Aa_{u,t}(\mathcal{L}(O)). \label{eqmotionpseudo}
\end{equation}
for almost all $u$. 

Consider now the map proven in Theorem~\ref{theorem1} $\tau_t:\mathfrak{U}_{loc} \to \mathcal{H}_u$ to be bounded. Using a straightforward generalization of Appendix of Doyon we find that $\lim_{t \to 0^+}||\tau_t O-\one O||_{\mathcal{H}_u}=0$, hence $\tau_t$ is strongly continuous. Moreover, as the proof Theorem~\ref{theorem1} works if we replace $t \to e^{o\ii \phi} t$ and, using Theorem 4.6 ((a) and (b)) of \cite{booksemigroup}, $\tau_t$ it is also analytic. Define the dual map $\tau_t^\oplus:\mathcal{H}^\dagger_u \to \mathcal{U}^\dagger_{loc}$. Hence $\lim_{t \to 0^+}||\tau^\oplus_t \Aa_u-\Aa_u||_{\mathfrak{U}_{loc}} = \lim_{t \to 0^+}||\tau_t O-O||_{\mathcal{H}_u}$ by a well-known result for bounded operators \cite{spectraltheory}. As $\mathfrak{U}_{loc}$ generates a (dense) subset of $\mathcal{H}_{u,t}$ we may "dilate" $\tau_t^\oplus$ to the operator $\mathfrak{T}'_t:=(\tau^\oplus_t)_{u,t}:\mathcal{H}^\dagger_{u} \to \mathcal{H}^\dagger_{u,t}$ which is also bounded and strongly continuous. By construction  $\mathfrak{T}'_t$ is a strongly continuous semigroup that solves the Cauchy problem \eqref{eqmotionpseudo} in the dual form,
\begin{equation}
\frac{d}{dt}\Aa_{u,t}=\mathfrak{L'}(\Aa_{u,t})=\Aa_{u,t} \circ \mathcal{L}. 
\end{equation}
The corresponding generator $\mathfrak{L'}$ is densely defined and closed by the Hille-Yosida theorem \cite{booksemigroup}. Moreover, by Proposition 1.4 of \cite{booksemigroup} there exists and $M$ $\mathfrak{T}_t=e^{-M t}\mathfrak{T}'_t:=e^{\mathfrak{L}t}$ is contracting and the claim about the spectral resolution also follows from the Hille-Yosida theorem. Likewise, an application of the Hille-Yosida theorem (in particular Proposition 2.2 of \cite{booksemigroup}) shows that $\Re(\sigma(\mathfrak{L}))\le0$. 

\end{proof}

\begin{proof}[Proof of Theorem~\ref{theorem3}]
From the equations of motion \eqref{eqmotion} by partial integration we immediately get,
\begin{equation}
 \int_0^T dt e^{\ii \lambda t} \omega_{s,t}(\LL(O))= e^{\ii \lambda t}\left.\omega_{s,t}(O)\right\vert_0^T+\ii\lambda\int_0^T dt e^{\ii \lambda t} \omega_{s,t}(O),
\end{equation}
and taking the $T \to \infty$ limit we immediately get the first statement of the theorem as $\omega_{s,t}(O)$ is bounded.

To proceed let us recall two useful definitions from the literature. First the integrated dissipation function \cite{Lindblad,Evans},
\begin{equation}
D_t(A,B):=\tau_t(A^\dagger B)-\tau_t(A^\dagger)\tau_t(B),
\end{equation}
which is sesquilinear and $D_t(A,A) \ge 0$ \cite{Evans}.
Second, the decoherence-free sub-algebra $\mathcal{N}$ \cite{Fagnola1,Frigerio}, which in our $C^*$-algebra case is,
\begin{equation}
\mathcal{N}:=\{O \in \mathfrak{U}_{loc} | D_t(O,O)=0\}.
\end{equation}
Our first step is to generalize a theorem by Frigerio \cite{Frigerio} for fixed points to long-time dynamics. Consider the faithful stationary state $\omega$ of $\tau_t$, i.e. $\omega \circ \tau_t =\omega$. We have (Theorem 3.1 of \cite{Frigerio}),
\begin{equation}
\lim_{s \to \infty} \omega(D_t(\tau_s(O),Q))=0, \qquad \forall O,Q \in \mathfrak{U},
\end{equation}
including $Q=\tau_s(O)$. 
By the Cauchy-Schwarz inequality and using the fact that $D_t$ is sesquilinear and $\omega$ faithful ($\omega(O^\dagger O)=0$ iff $O=0$), i.e. positive, we have,
\begin{equation}
w^*-\lim_{t \to \infty} \tau_t(O) \in \mathcal{N}. \label{limit}
\end{equation}
 Theorem 3.2 by Dhahri and Fagnola \cite{Fagnola1} that says,
\begin{equation}
\tau_t(O)=e^{\ii H t}O e^{-\ii H t}, \qquad \forall O \in \mathcal{N}, \label{Nevo}
\end{equation}
applies in our case because we assumed that $L_x(\eta), \forall x,\eta$ is bounded in the $C^*$-algebra norm and $\tau_t$ is strongly continuous \cite{NachtergaeleOpen}. 

Without loss of generality assume $\lambda \neq 0$ and take a sequence $C=\frac{2\pi n}{\lambda}$,
\begin{align*}
& \frac{1}{T} \int_0^T dt e^{\ii \lambda t} \frac{d}{dt}\omega_{s}(\tau_t(O))=\\
&\frac{1}{T} \int_0^C dt e^{\ii \lambda t} \frac{d}{dt}\omega_{s}(\tau_t(O))+\frac{1}{T} \int_0^{T-C} dt e^{\ii \lambda t} \frac{d}{dt}\omega_{s}  (\tau_{t+C}(O)).
\end{align*}
 As $\omega(\LL(O))$ is bounded, the first term on the r.h.s goes to 0 as $T \to \infty$. The second term may be estimated by \eqref{limit} and \eqref{Nevo} and these say that $\forall \varepsilon>0$, there exists a $n$ such that for every $t>0$,
\begin{equation}
\left|\frac{d}{dt}\omega_{s}(\tau_{t+C}(O))-\ii \omega_{s,t}([H,O])\right|<\varepsilon.
\end{equation}
 Hence, as $\omega_{s,t}$ are bounded, we get using Lebesgue's dominated convergence theorem that, 
 \begin{align*}
&\lim_{T \to \infty}\frac{1}{T} \int_0^T dt e^{\ii \lambda t} \frac{d}{dt}\omega_{s}(\tau_t(O))=\\
&\ii \lim_{T \to \infty}\frac{1}{T} \int_0^T dt e^{\ii \lambda t} \omega_{s,t}([H,O]).
 \end{align*}
 Using this and Theorem 3.3 of \cite{Fagnola1} that says $[L_x(\eta),O]=[L^\dagger_x(\eta),O]=0,$ $\forall O \in \mathcal{N}$ we get the claim in the second case. 
 
We have as before \eqref{eqmotionpseudo} for the pseudolocal quantities for almost all $u$ and,
\begin{equation*}
\omega_{s,\lambda}=\delta_{\lambda,0}\omega_0(A)+\int_0^s du \Aa_{u,\lambda}(O)
\end{equation*}
where we defined Fourier transform $\Aa_{u,\lambda}(O):=\lim_{T \to \infty}\frac{1}{T}\int^T_0 e^{\ii \lambda t} \Aa_{u,t}(O)$, and where we interchanged the order of integration which we can do according to Fubini's theorem because of continuity of the time evolution, i.e. $\frac{d}{dt}\omega_{s,t}(O)=\omega_{s,t}(\LL(O))$, boundedness of the linear functional $\Aa_{u,t}(O)$ and the fact that the functions are Lebesgue integrable (Definition 5.4 of \cite{Doyon}). Thus we arrive to the final statement. 
\end{proof}

\begin{proof}[Proof of Theorem~\ref{theorem4}]
Define $\beta (\mathcal{H})_\Lambda:=\sum_k \mu_k e^{\ii \lambda_k t}(A_k)_\Lambda + h.c.$. A theorem by Araki says that $\forall \beta$ in 1D the state is well-defined in the $\Lambda \to \infty$ limit and is analytic \cite{Araki}. Similar result holds for some critical value of $|\beta|>\beta_{*}$ (and by extension for $|\mu_k|>\mu_{*}$ and $\forall t$) in higher dimensions \cite{EisertLocality}. 

Let us show that $(A_k)_\Lambda:=\sum_{x\in\Lambda} a_x$ are pseudolocal dynamical symmetries with frequencies $\lambda_k$, i.e. that the corresponding $\Aa_{s,\lambda}(O)$ satisfy the conditions in Theorem~\ref{theorem3}.  Without loss of generality assume that $\omega_t(A_k)=0$. Then we can explicitly check that,
\begin{align*}
&\int_0^T dt\left\{e^{\ii \lambda_k t}\tr_{\Lambda}\left[\rho_\Lambda(t) e^{-\ii \lambda_k t}\sum_{x \in \Lambda}(a_{x,k})[H_\Lambda,O]\right]\right\}= \nonumber\\
&\int_0^T dt\left\{\tr_{\Lambda}\left[\left([H_\Lambda,\rho_\Lambda(t)] \sum_{x \in \Lambda}(a_{x,k}) \right.\right.\right.\\ 
&+\left.\left. \left.\rho_\Lambda(t) [H_\Lambda,\sum_{x \in \Lambda}(a_{x,k})]\right)O\right]\right\}.
\end{align*}
Because $\rho(t)$ is analytic we may express it as a unique uniformly converging Fourier series $\rho_\Lambda(t)=\sum_{n \in \ZZ} e^{i t n\frac{\theta}{T}}\rho_n $. Moreover, by the dual equations of motion, 
\begin{equation}
\frac{d}{dt}(\rho_\Lambda(t))_\Lambda=-\ii ( [H_\Lambda,\rho_\Lambda(t)])_\Lambda
\end{equation}

Hence,
\begin{equation*}
\int_0^T dt\left\{\tr_{\Lambda}\left[\left([H_\Lambda,\rho_\Lambda(t)] \sum_{x \in \Lambda}(a_{x,k})\right)O\right]\right\}=0,
\end{equation*}
using orthogonality of the Fourier coefficients. Now note that by assumption $\frac{d}{dt}\tau_t((A_k))_\Lambda= \LL(A_k)_\Lambda=\ii \lambda (A_k)_\Lambda$, which immediately implies the desired result upon taking the thermodynamic limit and taking into account the proof of Theorem~\ref{theorem1}. 

Parameterize the flow with a single parameter $s$, i.e. $\mu_k(s)=\mu_k s$. By analyticity, the equations for the flow of the t-GGE satisfy,
\begin{equation}
\frac{d}{dt}\omega_{s,t}(O)=\mu_k e^{\ii \lambda_k t} \Aa_{s,t=0}(O),
\end{equation}
which in general is solved by a path-ordered exponential. However, as we assumed that the pseudolocal dynamical symmetries form a closed finite algebra under commutation we may rewrite this solution in the form given in the theorem without path ordering by using the Baker–Campbell–Hausdorff formula. The rest of the claims follow directly from Theorems 6.1, 6.2 and Corollary 6.7 of \cite{Doyon} with $\beta H_\Lambda \to \beta (\mathcal{H})_\Lambda$. 
\end{proof}

\section{Details of the examples}

Here we will discuss some details of the calculations done in Sec.~\ref{examples}. 

\subsection{Spin-$1$ model with scars}
As shown in \cite{Tom2} $J^{\pm},J^z$ generate an $su(2)$ algebra and the scarred states in Eq.\eqref{scarredstates} form a representation for this algebra with $J^{\pm}$ being the roots (i.e. raising and lowering operators). The initial state \eqref{scarinitial} is an element of this algebra. The eigenvalues of $J^z$ in the scarring representation are $2n-V$, where $n=0,\ldots,V$. By rotating the algebra with $\exp{\ii \frac{\pi}{2} J^y}$ to the $x-$basis the initial state is diagonal. This allows us to easily find the trace $Z_0$ as given in the main text. For the t-GGE using the scar representation we may directly find that,
\begin{equation}
Z=\frac{e^{-V \mu(t))} \left(e^{\mu(t)  V+\mu(t)  (2 (V+1)-V)}-1\right)}{e^{2 \mu(t) }-1}+3^V-V-1,
\end{equation}
by counting the dimension of the kernel $P$ for the second part in the sum and we abbreviated $\mu(t)=\cos(2ht)\mu$. Likewise, we find for the other results from the main text using the representation in the scarred subspace and rotating $x$ to $z$,
\begin{align}
&\ave{A_1A_{-1}}_0=4 \frac{f(V)^2}{Z_0}\sum_{m=-V/2}^{V/2} e^{2 \mu_0 V} \left[\frac{m^2}{2}+\frac{V}{4}\left(\frac{V}{2}+1\right) \right]\\
&\ave{A_1}_0=\frac{2f(V)}{Z_0}\sum_{m=-V/2}^{V/2} e^{2 \mu_0 V} m
\end{align}
All of these can be simplified into the forms given in the main text. The initial average energy $\ave{H}_0$ is straightforwardly found using the spin flip symmetry discussed in the main text.
Similarly, to find the $\beta$ at small $\mu_0$ we simply expand $\exp(-\beta H)\approx \one + \beta H$ the equation,
\begin{equation}
\ave{H}_0=\ave{H}=\tr(-\beta H^2),
\end{equation}
is easily solved to be,
\begin{equation}
\beta=\frac{2 d N^D \sinh ^2\left(\left| \mu _0\right| \right)}{\left(2 \cosh \left(2 \left| \mu _0\right| \right)+1\right) \left(\left(d^2+3 h^2\right) N^D+12 (N-1)^D\right)}.
\end{equation}
Taking the thermodynamic limit gives the result in the main text.

\subsection{$t-J_z model$}
As shown in \cite{fragmentationSanjay} the SLIOMs can be conveniently written as,
\begin{equation}
A_k =\sum_{j_1 < \ldots < j_k} S_{1,j_1-1} \left(\prod_{m = 1}^{k-1}{P_{j_m} S_{j_m+1,j_{m+1}-1}}\right) Z_{j_k} ,
\end{equation}
where $P_x=N^\uparrow_x+N^\downarrow_x$, $S_{x,y}=\prod_{j=x}^y (\one-P_j)$ and the sum is defined as $\sum_{j_1 < \ldots < j_k}(x):=\sum_{j_1=1}^N \sum_{j_2=j_1+1}^N\ldots \sum_{j_k=j_{k-1}+1}^N (x)$. Useful relations are also, $ Z_x^2 = P_x, P_x Z_x = Z_x$ and it is useful to see that all the operators in the SLIOMs commute and that $P_x$, ($\one-P_x$), $S_{x,y}$ are projectors. The $A_k$ are diagonal in the $Z_x$ basis. Hence, in order to compute the partition function of the t-GGE \eqref{tjparitionfunction} we note that $P_x$, ($\one-P_x$) have eigenvalues $\{1,0\}$ (resp. $\{0,1\}$) on opposite subspaces and $Z_x$ has eigenvalues $\{1,-1,0\}$ ($1,-1$ correspond to the $P_x$ eigenvalue 1 subspace). Hence, it becomes a matter of simple combinatorics and splitting of the expressions $\exp(\mu_k A_k)$ to evaluate,
\begin{equation}
Z=\sum_{j=1}^{4^N} \exp\left(\prod_{m=1}^N \lambda_{j,m}\right),
\end{equation}
where $\lambda_{j,m}$ is the $j-$th diagonal value of the $m$-th operator from the left. Similarly, we can find $Z(\alpha_1,\alpha_2)$ for the examples given in the main text. 

The SLIOMs containt non-local strings $S_{j,m}$. However, since $P_x S_{j,m}=0$ for $j \le x \le m$ and recalling that $ Z_x^2 = P_x, P_x Z_x = Z_x$ and that $P_x$, ($\one-P_x$), $S_{x,y}$ are projectors, it is easy to see that these strings can identically cancel in the expressions $exp(\mu_k A_k)$ when expanded, apart from the leading order in $\mu_k$ the contribution of which in the thermodynamic limit is identically small. However products of the different $A_k$ in the full $\rho_{tGGE}$ can, for different $\mu_k$ render the contribution for the strings thermodynamically finite and this is the origin of the fragmentation phase transition discussed in the main text.

\bibliography{main}% Produces the bibliography via BibTeX.

%apsrev4-2.bst 2019-01-14 (MD) hand-edited version of apsrev4-1.bst
%Control: key (0)
%Control: author (8) initials jnrlst
%Control: editor formatted (1) identically to author
%Control: production of article title (0) allowed
%Control: page (0) single
%Control: year (1) truncated
%Control: production of eprint (0) enabled
\begin{thebibliography}{194}%
\makeatletter
\providecommand \@ifxundefined [1]{%
 \@ifx{#1\undefined}
}%
\providecommand \@ifnum [1]{%
 \ifnum #1\expandafter \@firstoftwo
 \else \expandafter \@secondoftwo
 \fi
}%
\providecommand \@ifx [1]{%
 \ifx #1\expandafter \@firstoftwo
 \else \expandafter \@secondoftwo
 \fi
}%
\providecommand \natexlab [1]{#1}%
\providecommand \enquote  [1]{``#1''}%
\providecommand \bibnamefont  [1]{#1}%
\providecommand \bibfnamefont [1]{#1}%
\providecommand \citenamefont [1]{#1}%
\providecommand \href@noop [0]{\@secondoftwo}%
\providecommand \href [0]{\begingroup \@sanitize@url \@href}%
\providecommand \@href[1]{\@@startlink{#1}\@@href}%
\providecommand \@@href[1]{\endgroup#1\@@endlink}%
\providecommand \@sanitize@url [0]{\catcode `\\12\catcode `\$12\catcode
  `\&12\catcode `\#12\catcode `\^12\catcode `\_12\catcode `\%12\relax}%
\providecommand \@@startlink[1]{}%
\providecommand \@@endlink[0]{}%
\providecommand \url  [0]{\begingroup\@sanitize@url \@url }%
\providecommand \@url [1]{\endgroup\@href {#1}{\urlprefix }}%
\providecommand \urlprefix  [0]{URL }%
\providecommand \Eprint [0]{\href }%
\providecommand \doibase [0]{https://doi.org/}%
\providecommand \selectlanguage [0]{\@gobble}%
\providecommand \bibinfo  [0]{\@secondoftwo}%
\providecommand \bibfield  [0]{\@secondoftwo}%
\providecommand \translation [1]{[#1]}%
\providecommand \BibitemOpen [0]{}%
\providecommand \bibitemStop [0]{}%
\providecommand \bibitemNoStop [0]{.\EOS\space}%
\providecommand \EOS [0]{\spacefactor3000\relax}%
\providecommand \BibitemShut  [1]{\csname bibitem#1\endcsname}%
\let\auto@bib@innerbib\@empty
%</preamble>
\bibitem [{\citenamefont {D'Alessio}\ \emph {et~al.}(2016)\citenamefont
  {D'Alessio}, \citenamefont {Kafri}, \citenamefont {Polkovnikov},\ and\
  \citenamefont {Rigol}}]{ETHReview}%
  \BibitemOpen
  \bibfield  {author} {\bibinfo {author} {\bibfnamefont {L.}~\bibnamefont
  {D'Alessio}}, \bibinfo {author} {\bibfnamefont {Y.}~\bibnamefont {Kafri}},
  \bibinfo {author} {\bibfnamefont {A.}~\bibnamefont {Polkovnikov}},\ and\
  \bibinfo {author} {\bibfnamefont {M.}~\bibnamefont {Rigol}},\ }\bibfield
  {title} {\bibinfo {title} {From quantum chaos and eigenstate thermalization
  to statistical mechanics and thermodynamics},\ }\href
  {https://doi.org/10.1080/00018732.2016.1198134} {\bibfield  {journal}
  {\bibinfo  {journal} {Advances in Physics}\ }\textbf {\bibinfo {volume}
  {65}},\ \bibinfo {pages} {239} (\bibinfo {year} {2016})},\ \Eprint
  {https://arxiv.org/abs/https://doi.org/10.1080/00018732.2016.1198134}
  {https://doi.org/10.1080/00018732.2016.1198134} \BibitemShut {NoStop}%
\bibitem [{\citenamefont {Essler}\ and\ \citenamefont
  {Fagotti}(2016)}]{thermoreview}%
  \BibitemOpen
  \bibfield  {author} {\bibinfo {author} {\bibfnamefont {F.~H.~L.}\
  \bibnamefont {Essler}}\ and\ \bibinfo {author} {\bibfnamefont
  {M.}~\bibnamefont {Fagotti}},\ }\bibfield  {title} {\bibinfo {title} {Quench
  dynamics and relaxation in isolated integrable quantum spin chains},\ }\href
  {https://doi.org/10.1088/1742-5468/2016/06/064002} {\bibfield  {journal}
  {\bibinfo  {journal} {Journal of Statistical Mechanics: Theory and
  Experiment}\ }\textbf {\bibinfo {volume} {2016}},\ \bibinfo {pages} {064002}
  (\bibinfo {year} {2016})}\BibitemShut {NoStop}%
\bibitem [{\citenamefont {Castro-Alvaredo}\ \emph {et~al.}(2016)\citenamefont
  {Castro-Alvaredo}, \citenamefont {Doyon},\ and\ \citenamefont
  {Yoshimura}}]{castroalvaredo2016emergent}%
  \BibitemOpen
  \bibfield  {author} {\bibinfo {author} {\bibfnamefont {O.~A.}\ \bibnamefont
  {Castro-Alvaredo}}, \bibinfo {author} {\bibfnamefont {B.}~\bibnamefont
  {Doyon}},\ and\ \bibinfo {author} {\bibfnamefont {T.}~\bibnamefont
  {Yoshimura}},\ }\bibfield  {title} {\bibinfo {title} {Emergent hydrodynamics
  in integrable quantum systems out of equilibrium},\ }\href
  {https://doi.org/10.1103/PhysRevX.6.041065} {\bibfield  {journal} {\bibinfo
  {journal} {Phys. Rev. X}\ }\textbf {\bibinfo {volume} {6}},\ \bibinfo {pages}
  {041065} (\bibinfo {year} {2016})}\BibitemShut {NoStop}%
\bibitem [{\citenamefont {Bertini}\ \emph {et~al.}(2016)\citenamefont
  {Bertini}, \citenamefont {Collura}, \citenamefont {De~Nardis},\ and\
  \citenamefont {Fagotti}}]{bertini2016transport}%
  \BibitemOpen
  \bibfield  {author} {\bibinfo {author} {\bibfnamefont {B.}~\bibnamefont
  {Bertini}}, \bibinfo {author} {\bibfnamefont {M.}~\bibnamefont {Collura}},
  \bibinfo {author} {\bibfnamefont {J.}~\bibnamefont {De~Nardis}},\ and\
  \bibinfo {author} {\bibfnamefont {M.}~\bibnamefont {Fagotti}},\ }\bibfield
  {title} {\bibinfo {title} {Transport in out-of-equilibrium $xxz$ chains:
  Exact profiles of charges and currents},\ }\href
  {https://doi.org/10.1103/PhysRevLett.117.207201} {\bibfield  {journal}
  {\bibinfo  {journal} {Phys. Rev. Lett.}\ }\textbf {\bibinfo {volume} {117}},\
  \bibinfo {pages} {207201} (\bibinfo {year} {2016})}\BibitemShut {NoStop}%
\bibitem [{\citenamefont {Doyon}(2021)}]{Doyon1}%
  \BibitemOpen
  \bibfield  {author} {\bibinfo {author} {\bibfnamefont {B.}~\bibnamefont
  {Doyon}},\ }\href@noop {} {\bibinfo {title} {Hydrodynamic projections and the
  emergence of linearised euler equations in one-dimensional isolated systems}}
  (\bibinfo {year} {2021}),\ \Eprint {https://arxiv.org/abs/2011.00611}
  {arXiv:2011.00611 [math-ph]} \BibitemShut {NoStop}%
\bibitem [{\citenamefont {Turner}\ \emph {et~al.}(2018)\citenamefont {Turner},
  \citenamefont {Michailidis}, \citenamefont {Abanin}, \citenamefont {Serbyn},\
  and\ \citenamefont {Papi{\'c}}}]{scars}%
  \BibitemOpen
  \bibfield  {author} {\bibinfo {author} {\bibfnamefont {C.~J.}\ \bibnamefont
  {Turner}}, \bibinfo {author} {\bibfnamefont {A.~A.}\ \bibnamefont
  {Michailidis}}, \bibinfo {author} {\bibfnamefont {D.~A.}\ \bibnamefont
  {Abanin}}, \bibinfo {author} {\bibfnamefont {M.}~\bibnamefont {Serbyn}},\
  and\ \bibinfo {author} {\bibfnamefont {Z.}~\bibnamefont {Papi{\'c}}},\
  }\bibfield  {title} {\bibinfo {title} {Weak ergodicity breaking from quantum
  many-body scars},\ }\href {https://doi.org/10.1038/s41567-018-0137-5}
  {\bibfield  {journal} {\bibinfo  {journal} {Nat. Phys.}\ }\textbf {\bibinfo
  {volume} {14}},\ \bibinfo {pages} {745} (\bibinfo {year} {2018})}\BibitemShut
  {NoStop}%
\bibitem [{\citenamefont {Moudgalya}\ \emph {et~al.}(2018)\citenamefont
  {Moudgalya}, \citenamefont {Rachel}, \citenamefont {Bernevig},\ and\
  \citenamefont {Regnault}}]{scars2}%
  \BibitemOpen
  \bibfield  {author} {\bibinfo {author} {\bibfnamefont {S.}~\bibnamefont
  {Moudgalya}}, \bibinfo {author} {\bibfnamefont {S.}~\bibnamefont {Rachel}},
  \bibinfo {author} {\bibfnamefont {B.~A.}\ \bibnamefont {Bernevig}},\ and\
  \bibinfo {author} {\bibfnamefont {N.}~\bibnamefont {Regnault}},\ }\bibfield
  {title} {\bibinfo {title} {Exact excited states of nonintegrable models},\
  }\href {http://dx.doi.org/10.1103/PhysRevB.98.235155} {\bibfield  {journal}
  {\bibinfo  {journal} {Phys. Rev. B}\ }\textbf {\bibinfo {volume} {98}}
  (\bibinfo {year} {2018})}\BibitemShut {NoStop}%
\bibitem [{\citenamefont {Serbyn}\ \emph {et~al.}(2021)\citenamefont {Serbyn},
  \citenamefont {Abanin},\ and\ \citenamefont {Papić}}]{SerbynScars}%
  \BibitemOpen
  \bibfield  {author} {\bibinfo {author} {\bibfnamefont {M.}~\bibnamefont
  {Serbyn}}, \bibinfo {author} {\bibfnamefont {D.~A.}\ \bibnamefont {Abanin}},\
  and\ \bibinfo {author} {\bibfnamefont {Z.}~\bibnamefont {Papić}},\
  }\bibfield  {title} {\bibinfo {title} {Quantum many-body scars and weak
  breaking of ergodicity},\ }\href {https://doi.org/10.1038/s41567-021-01230-2}
  {\bibfield  {journal} {\bibinfo  {journal} {Nature Physics}\ }\textbf
  {\bibinfo {volume} {17}},\ \bibinfo {pages} {675–685} (\bibinfo {year}
  {2021})}\BibitemShut {NoStop}%
\bibitem [{\citenamefont {Tamura}\ and\ \citenamefont
  {Katsura}(2022)}]{HoshoScars}%
  \BibitemOpen
  \bibfield  {author} {\bibinfo {author} {\bibfnamefont {K.}~\bibnamefont
  {Tamura}}\ and\ \bibinfo {author} {\bibfnamefont {H.}~\bibnamefont
  {Katsura}},\ }\bibfield  {title} {\bibinfo {title} {Quantum many-body scars
  of spinless fermions with density-assisted hopping in higher dimensions},\
  }\href {https://doi.org/10.1103/PhysRevB.106.144306} {\bibfield  {journal}
  {\bibinfo  {journal} {Phys. Rev. B}\ }\textbf {\bibinfo {volume} {106}},\
  \bibinfo {pages} {144306} (\bibinfo {year} {2022})}\BibitemShut {NoStop}%
\bibitem [{\citenamefont {Valencia-Tortora}\ \emph {et~al.}(2022)\citenamefont
  {Valencia-Tortora}, \citenamefont {Pancotti},\ and\ \citenamefont
  {Marino}}]{JamirScars}%
  \BibitemOpen
  \bibfield  {author} {\bibinfo {author} {\bibfnamefont {R.~J.}\ \bibnamefont
  {Valencia-Tortora}}, \bibinfo {author} {\bibfnamefont {N.}~\bibnamefont
  {Pancotti}},\ and\ \bibinfo {author} {\bibfnamefont {J.}~\bibnamefont
  {Marino}},\ }\bibfield  {title} {\bibinfo {title} {Kinetically constrained
  quantum dynamics in superconducting circuits},\ }\href
  {https://doi.org/10.1103/PRXQuantum.3.020346} {\bibfield  {journal} {\bibinfo
   {journal} {PRX Quantum}\ }\textbf {\bibinfo {volume} {3}},\ \bibinfo {pages}
  {020346} (\bibinfo {year} {2022})}\BibitemShut {NoStop}%
\bibitem [{\citenamefont {Chen}\ \emph {et~al.}(2022)\citenamefont {Chen},
  \citenamefont {Chen},\ and\ \citenamefont {Zhu}}]{ZhengScars1}%
  \BibitemOpen
  \bibfield  {author} {\bibinfo {author} {\bibfnamefont {Q.}~\bibnamefont
  {Chen}}, \bibinfo {author} {\bibfnamefont {S.~A.}\ \bibnamefont {Chen}},\
  and\ \bibinfo {author} {\bibfnamefont {Z.}~\bibnamefont {Zhu}},\ }\href
  {https://doi.org/10.48550/ARXIV.2202.08638} {\bibinfo {title} {Weak
  ergodicity breaking in non-hermitian many-body systems}} (\bibinfo {year}
  {2022})\BibitemShut {NoStop}%
\bibitem [{\citenamefont {Chen}\ and\ \citenamefont {Zhu}(2023)}]{ZhengScars2}%
  \BibitemOpen
  \bibfield  {author} {\bibinfo {author} {\bibfnamefont {Q.}~\bibnamefont
  {Chen}}\ and\ \bibinfo {author} {\bibfnamefont {Z.}~\bibnamefont {Zhu}},\
  }\href {https://doi.org/10.48550/ARXIV.2301.03405} {\bibinfo {title}
  {Inverting multiple quantum many-body scars via disorder}} (\bibinfo {year}
  {2023})\BibitemShut {NoStop}%
\bibitem [{\citenamefont {McClarty}\ \emph {et~al.}(2020)\citenamefont
  {McClarty}, \citenamefont {Haque}, \citenamefont {Sen},\ and\ \citenamefont
  {Richter}}]{ArnabScars1}%
  \BibitemOpen
  \bibfield  {author} {\bibinfo {author} {\bibfnamefont {P.~A.}\ \bibnamefont
  {McClarty}}, \bibinfo {author} {\bibfnamefont {M.}~\bibnamefont {Haque}},
  \bibinfo {author} {\bibfnamefont {A.}~\bibnamefont {Sen}},\ and\ \bibinfo
  {author} {\bibfnamefont {J.}~\bibnamefont {Richter}},\ }\bibfield  {title}
  {\bibinfo {title} {Disorder-free localization and many-body quantum scars
  from magnetic frustration},\ }\href
  {https://doi.org/10.1103/PhysRevB.102.224303} {\bibfield  {journal} {\bibinfo
   {journal} {Phys. Rev. B}\ }\textbf {\bibinfo {volume} {102}},\ \bibinfo
  {pages} {224303} (\bibinfo {year} {2020})}\BibitemShut {NoStop}%
\bibitem [{\citenamefont {Desaules}\ \emph {et~al.}(2022)\citenamefont
  {Desaules}, \citenamefont {Hudomal}, \citenamefont {Banerjee}, \citenamefont
  {Sen}, \citenamefont {Papić},\ and\ \citenamefont {Halimeh}}]{ArnabScars2}%
  \BibitemOpen
  \bibfield  {author} {\bibinfo {author} {\bibfnamefont {J.-Y.}\ \bibnamefont
  {Desaules}}, \bibinfo {author} {\bibfnamefont {A.}~\bibnamefont {Hudomal}},
  \bibinfo {author} {\bibfnamefont {D.}~\bibnamefont {Banerjee}}, \bibinfo
  {author} {\bibfnamefont {A.}~\bibnamefont {Sen}}, \bibinfo {author}
  {\bibfnamefont {Z.}~\bibnamefont {Papić}},\ and\ \bibinfo {author}
  {\bibfnamefont {J.~C.}\ \bibnamefont {Halimeh}},\ }\href
  {https://doi.org/10.48550/ARXIV.2204.01745} {\bibinfo {title} {Prominent
  quantum many-body scars in a truncated schwinger model}} (\bibinfo {year}
  {2022})\BibitemShut {NoStop}%
\bibitem [{\citenamefont {Sala}\ \emph
  {et~al.}(2020{\natexlab{a}})\citenamefont {Sala}, \citenamefont {Rakovszky},
  \citenamefont {Verresen}, \citenamefont {Knap},\ and\ \citenamefont
  {Pollmann}}]{Fragmentation1}%
  \BibitemOpen
  \bibfield  {author} {\bibinfo {author} {\bibfnamefont {P.}~\bibnamefont
  {Sala}}, \bibinfo {author} {\bibfnamefont {T.}~\bibnamefont {Rakovszky}},
  \bibinfo {author} {\bibfnamefont {R.}~\bibnamefont {Verresen}}, \bibinfo
  {author} {\bibfnamefont {M.}~\bibnamefont {Knap}},\ and\ \bibinfo {author}
  {\bibfnamefont {F.}~\bibnamefont {Pollmann}},\ }\bibfield  {title} {\bibinfo
  {title} {Ergodicity breaking arising from hilbert space fragmentation in
  dipole-conserving hamiltonians},\ }\href
  {https://doi.org/10.1103/PhysRevX.10.011047} {\bibfield  {journal} {\bibinfo
  {journal} {Phys. Rev. X}\ }\textbf {\bibinfo {volume} {10}},\ \bibinfo
  {pages} {011047} (\bibinfo {year} {2020}{\natexlab{a}})}\BibitemShut
  {NoStop}%
\bibitem [{\citenamefont {Khemani}\ \emph {et~al.}(2020)\citenamefont
  {Khemani}, \citenamefont {Hermele},\ and\ \citenamefont
  {Nandkishore}}]{Fragmentation2}%
  \BibitemOpen
  \bibfield  {author} {\bibinfo {author} {\bibfnamefont {V.}~\bibnamefont
  {Khemani}}, \bibinfo {author} {\bibfnamefont {M.}~\bibnamefont {Hermele}},\
  and\ \bibinfo {author} {\bibfnamefont {R.}~\bibnamefont {Nandkishore}},\
  }\bibfield  {title} {\bibinfo {title} {Localization from hilbert space
  shattering: From theory to physical realizations},\ }\href
  {https://doi.org/10.1103/PhysRevB.101.174204} {\bibfield  {journal} {\bibinfo
   {journal} {Phys. Rev. B}\ }\textbf {\bibinfo {volume} {101}},\ \bibinfo
  {pages} {174204} (\bibinfo {year} {2020})}\BibitemShut {NoStop}%
\bibitem [{\citenamefont {Moudgalya}\ \emph {et~al.}(2019)\citenamefont
  {Moudgalya}, \citenamefont {Prem}, \citenamefont {Nandkishore}, \citenamefont
  {Regnault},\ and\ \citenamefont {Bernevig}}]{fragmentationSanjay}%
  \BibitemOpen
  \bibfield  {author} {\bibinfo {author} {\bibfnamefont {S.}~\bibnamefont
  {Moudgalya}}, \bibinfo {author} {\bibfnamefont {A.}~\bibnamefont {Prem}},
  \bibinfo {author} {\bibfnamefont {R.}~\bibnamefont {Nandkishore}}, \bibinfo
  {author} {\bibfnamefont {N.}~\bibnamefont {Regnault}},\ and\ \bibinfo
  {author} {\bibfnamefont {B.~A.}\ \bibnamefont {Bernevig}},\ }\href@noop {}
  {\bibinfo {title} {Thermalization and its absence within krylov subspaces of
  a constrained hamiltonian}} (\bibinfo {year} {2019}),\ \Eprint
  {https://arxiv.org/abs/1910.14048} {arXiv:1910.14048 [cond-mat.str-el]}
  \BibitemShut {NoStop}%
\bibitem [{\citenamefont {Wilczek}(2012)}]{Wilczek2012}%
  \BibitemOpen
  \bibfield  {author} {\bibinfo {author} {\bibfnamefont {F.}~\bibnamefont
  {Wilczek}},\ }\bibfield  {title} {\bibinfo {title} {Quantum time crystals},\
  }\href {https://doi.org/10.1103/PhysRevLett.109.160401} {\bibfield  {journal}
  {\bibinfo  {journal} {Phys. Rev. Lett.}\ }\textbf {\bibinfo {volume} {109}},\
  \bibinfo {pages} {160401} (\bibinfo {year} {2012})}\BibitemShut {NoStop}%
\bibitem [{\citenamefont {Khemani}\ \emph {et~al.}(2016)\citenamefont
  {Khemani}, \citenamefont {Lazarides}, \citenamefont {Moessner},\ and\
  \citenamefont {Sondhi}}]{VedikaLazarides}%
  \BibitemOpen
  \bibfield  {author} {\bibinfo {author} {\bibfnamefont {V.}~\bibnamefont
  {Khemani}}, \bibinfo {author} {\bibfnamefont {A.}~\bibnamefont {Lazarides}},
  \bibinfo {author} {\bibfnamefont {R.}~\bibnamefont {Moessner}},\ and\
  \bibinfo {author} {\bibfnamefont {S.~L.}\ \bibnamefont {Sondhi}},\ }\bibfield
   {title} {\bibinfo {title} {{Phase Structure of Driven Quantum Systems}},\
  }\href {https://doi.org/10.1103/PhysRevLett.116.250401} {\bibfield  {journal}
  {\bibinfo  {journal} {Phys. Rev. Lett.}\ }\textbf {\bibinfo {volume} {116}},\
  \bibinfo {pages} {250401} (\bibinfo {year} {2016})}\BibitemShut {NoStop}%
\bibitem [{\citenamefont {Liang}\ \emph {et~al.}(2020)\citenamefont {Liang},
  \citenamefont {Fazio},\ and\ \citenamefont {Chesi}}]{liang2020time}%
  \BibitemOpen
  \bibfield  {author} {\bibinfo {author} {\bibfnamefont {P.}~\bibnamefont
  {Liang}}, \bibinfo {author} {\bibfnamefont {R.}~\bibnamefont {Fazio}},\ and\
  \bibinfo {author} {\bibfnamefont {S.}~\bibnamefont {Chesi}},\ }\bibfield
  {title} {\bibinfo {title} {Time crystals in the driven transverse field ising
  model under quasiperiodic modulation},\ }\href@noop {} {\bibfield  {journal}
  {\bibinfo  {journal} {New Journal of Physics}\ }\textbf {\bibinfo {volume}
  {22}},\ \bibinfo {pages} {125001} (\bibinfo {year} {2020})}\BibitemShut
  {NoStop}%
\bibitem [{\citenamefont {Zhou}\ \emph {et~al.}(2022)\citenamefont {Zhou},
  \citenamefont {Su}, \citenamefont {Halimeh}, \citenamefont {Ott},
  \citenamefont {Sun}, \citenamefont {Hauke}, \citenamefont {Yang},
  \citenamefont {Yuan}, \citenamefont {Berges},\ and\ \citenamefont
  {Pan}}]{Zhou_2022}%
  \BibitemOpen
  \bibfield  {author} {\bibinfo {author} {\bibfnamefont {Z.-Y.}\ \bibnamefont
  {Zhou}}, \bibinfo {author} {\bibfnamefont {G.-X.}\ \bibnamefont {Su}},
  \bibinfo {author} {\bibfnamefont {J.~C.}\ \bibnamefont {Halimeh}}, \bibinfo
  {author} {\bibfnamefont {R.}~\bibnamefont {Ott}}, \bibinfo {author}
  {\bibfnamefont {H.}~\bibnamefont {Sun}}, \bibinfo {author} {\bibfnamefont
  {P.}~\bibnamefont {Hauke}}, \bibinfo {author} {\bibfnamefont
  {B.}~\bibnamefont {Yang}}, \bibinfo {author} {\bibfnamefont {Z.-S.}\
  \bibnamefont {Yuan}}, \bibinfo {author} {\bibfnamefont {J.}~\bibnamefont
  {Berges}},\ and\ \bibinfo {author} {\bibfnamefont {J.-W.}\ \bibnamefont
  {Pan}},\ }\bibfield  {title} {\bibinfo {title} {Thermalization dynamics of a
  gauge theory on a quantum simulator},\ }\href
  {https://doi.org/10.1126/science.abl6277} {\bibfield  {journal} {\bibinfo
  {journal} {Science}\ }\textbf {\bibinfo {volume} {377}},\ \bibinfo {pages}
  {311} (\bibinfo {year} {2022})}\BibitemShut {NoStop}%
\bibitem [{\citenamefont {Halimeh}\ \emph
  {et~al.}(2022{\natexlab{a}})\citenamefont {Halimeh}, \citenamefont {Hauke},
  \citenamefont {Knolle},\ and\ \citenamefont {Grusdt}}]{JadDis}%
  \BibitemOpen
  \bibfield  {author} {\bibinfo {author} {\bibfnamefont {J.~C.}\ \bibnamefont
  {Halimeh}}, \bibinfo {author} {\bibfnamefont {P.}~\bibnamefont {Hauke}},
  \bibinfo {author} {\bibfnamefont {J.}~\bibnamefont {Knolle}},\ and\ \bibinfo
  {author} {\bibfnamefont {F.}~\bibnamefont {Grusdt}},\ }\href
  {https://doi.org/10.48550/ARXIV.2206.11273} {\bibinfo {title}
  {Temperature-induced disorder-free localization}} (\bibinfo {year}
  {2022}{\natexlab{a}})\BibitemShut {NoStop}%
\bibitem [{\citenamefont {Smith}\ \emph {et~al.}(2017)\citenamefont {Smith},
  \citenamefont {Knolle}, \citenamefont {Kovrizhin},\ and\ \citenamefont
  {Moessner}}]{Dima}%
  \BibitemOpen
  \bibfield  {author} {\bibinfo {author} {\bibfnamefont {A.}~\bibnamefont
  {Smith}}, \bibinfo {author} {\bibfnamefont {J.}~\bibnamefont {Knolle}},
  \bibinfo {author} {\bibfnamefont {D.~L.}\ \bibnamefont {Kovrizhin}},\ and\
  \bibinfo {author} {\bibfnamefont {R.}~\bibnamefont {Moessner}},\ }\bibfield
  {title} {\bibinfo {title} {Disorder-free localization},\ }\href
  {https://doi.org/10.1103/PhysRevLett.118.266601} {\bibfield  {journal}
  {\bibinfo  {journal} {Phys. Rev. Lett.}\ }\textbf {\bibinfo {volume} {118}},\
  \bibinfo {pages} {266601} (\bibinfo {year} {2017})}\BibitemShut {NoStop}%
\bibitem [{\citenamefont {Bratteli}\ and\ \citenamefont
  {Robinson}(1981)}]{operator2}%
  \BibitemOpen
  \bibfield  {author} {\bibinfo {author} {\bibfnamefont {O.}~\bibnamefont
  {Bratteli}}\ and\ \bibinfo {author} {\bibfnamefont {D.}~\bibnamefont
  {Robinson}},\ }\href@noop {} {\emph {\bibinfo {title} {Operator algebras and
  quantum statistical mechanics II}}}\ (\bibinfo  {publisher} {Springer},\
  \bibinfo {year} {1981})\BibitemShut {NoStop}%
\bibitem [{\citenamefont {Cheverry}\ and\ \citenamefont
  {Raymond}(2021)}]{spectraltheory}%
  \BibitemOpen
  \bibfield  {author} {\bibinfo {author} {\bibfnamefont {C.}~\bibnamefont
  {Cheverry}}\ and\ \bibinfo {author} {\bibfnamefont {N.}~\bibnamefont
  {Raymond}},\ }\href@noop {} {\emph {\bibinfo {title} {A guide to spectral
  theory: applications and exercises}}}\ (\bibinfo  {publisher} {Springer
  Nature},\ \bibinfo {year} {2021})\BibitemShut {NoStop}%
\bibitem [{\citenamefont {Bratteli}\ and\ \citenamefont
  {Robinson}(1987)}]{operator1}%
  \BibitemOpen
  \bibfield  {author} {\bibinfo {author} {\bibfnamefont {O.}~\bibnamefont
  {Bratteli}}\ and\ \bibinfo {author} {\bibfnamefont {D.~W.}\ \bibnamefont
  {Robinson}},\ }\href@noop {} {\emph {\bibinfo {title} {Operator algebras and
  quantum statistical mechanics: Volume 1: C*-and W*-Algebras. Symmetry Groups.
  Decomposition of States}}}\ (\bibinfo  {publisher} {Springer Science \&
  Business Media},\ \bibinfo {year} {1987})\BibitemShut {NoStop}%
\bibitem [{\citenamefont {Naaijkens}(2013)}]{spinsystsemsbook}%
  \BibitemOpen
  \bibfield  {author} {\bibinfo {author} {\bibfnamefont {P.}~\bibnamefont
  {Naaijkens}},\ }\href@noop {} {\emph {\bibinfo {title} {Quantum spin systems
  on infinite lattices}}}\ (\bibinfo  {publisher} {Springer},\ \bibinfo {year}
  {2013})\BibitemShut {NoStop}%
\bibitem [{\citenamefont {Doyon}(2017)}]{Doyon}%
  \BibitemOpen
  \bibfield  {author} {\bibinfo {author} {\bibfnamefont {B.}~\bibnamefont
  {Doyon}},\ }\bibfield  {title} {\bibinfo {title} {Thermalization and
  pseudolocality in extended quantum systems},\ }\href
  {https://doi.org/10.1007/s00220-017-2836-7} {\bibfield  {journal} {\bibinfo
  {journal} {Communications in Mathematical Physics}\ }\textbf {\bibinfo
  {volume} {351}},\ \bibinfo {pages} {155} (\bibinfo {year}
  {2017})}\BibitemShut {NoStop}%
\bibitem [{\citenamefont {Paulsen}(2003)}]{paulsen_2003}%
  \BibitemOpen
  \bibfield  {author} {\bibinfo {author} {\bibfnamefont {V.}~\bibnamefont
  {Paulsen}},\ }\href {https://doi.org/10.1017/CBO9780511546631} {\emph
  {\bibinfo {title} {Completely Bounded Maps and Operator Algebras}}},\
  Cambridge Studies in Advanced Mathematics\ (\bibinfo  {publisher} {Cambridge
  University Press},\ \bibinfo {year} {2003})\BibitemShut {NoStop}%
\bibitem [{\citenamefont {Neumann}(1929)}]{vonNeumann}%
  \BibitemOpen
  \bibfield  {author} {\bibinfo {author} {\bibfnamefont {J.~v.}\ \bibnamefont
  {Neumann}},\ }\bibfield  {title} {\bibinfo {title} {Beweis des ergodensatzes
  und desh-theorems in der neuen mechanik},\ }\href
  {https://doi.org/10.1007/BF01339852} {\bibfield  {journal} {\bibinfo
  {journal} {Zeitschrift f{\"u}r Physik}\ }\textbf {\bibinfo {volume} {57}},\
  \bibinfo {pages} {30} (\bibinfo {year} {1929})}\BibitemShut {NoStop}%
\bibitem [{\citenamefont {Emch}(1978)}]{Emch1}%
  \BibitemOpen
  \bibfield  {author} {\bibinfo {author} {\bibfnamefont {G.~G.}\ \bibnamefont
  {Emch}},\ }\bibfield  {title} {\bibinfo {title} {Phase transitions, approach
  to equilibrium, and structural stability},\ }in\ \href@noop {} {\emph
  {\bibinfo {booktitle} {Group Theoretical Methods in Physics}}},\ \bibinfo
  {editor} {edited by\ \bibinfo {editor} {\bibfnamefont {P.}~\bibnamefont
  {Kramer}}\ and\ \bibinfo {editor} {\bibfnamefont {A.}~\bibnamefont
  {Rieckers}}}\ (\bibinfo  {publisher} {Springer Berlin Heidelberg},\ \bibinfo
  {address} {Berlin, Heidelberg},\ \bibinfo {year} {1978})\ pp.\ \bibinfo
  {pages} {223--246}\BibitemShut {NoStop}%
\bibitem [{\citenamefont {Hume}\ and\ \citenamefont
  {Robinson}(1986)}]{Robinson1}%
  \BibitemOpen
  \bibfield  {author} {\bibinfo {author} {\bibfnamefont {L.}~\bibnamefont
  {Hume}}\ and\ \bibinfo {author} {\bibfnamefont {D.~W.}\ \bibnamefont
  {Robinson}},\ }\bibfield  {title} {\bibinfo {title} {Return to equilibrium in
  thexy model},\ }\href {https://doi.org/10.1007/BF01011909} {\bibfield
  {journal} {\bibinfo  {journal} {Journal of Statistical Physics}\ }\textbf
  {\bibinfo {volume} {44}},\ \bibinfo {pages} {829} (\bibinfo {year}
  {1986})}\BibitemShut {NoStop}%
\bibitem [{\citenamefont {Robinson}(1973)}]{Robinson2}%
  \BibitemOpen
  \bibfield  {author} {\bibinfo {author} {\bibfnamefont {D.~W.}\ \bibnamefont
  {Robinson}},\ }\bibfield  {title} {\bibinfo {title} {Return to equilibrium},\
  }\href {https://doi.org/10.1007/BF01646264} {\bibfield  {journal} {\bibinfo
  {journal} {Communications in Mathematical Physics}\ }\textbf {\bibinfo
  {volume} {31}},\ \bibinfo {pages} {171} (\bibinfo {year} {1973})}\BibitemShut
  {NoStop}%
\bibitem [{\citenamefont {Narnhofer}\ and\ \citenamefont
  {Robinson}(1975)}]{Robinson3}%
  \BibitemOpen
  \bibfield  {author} {\bibinfo {author} {\bibfnamefont {H.}~\bibnamefont
  {Narnhofer}}\ and\ \bibinfo {author} {\bibfnamefont {D.}~\bibnamefont
  {Robinson}},\ }\bibfield  {title} {\bibinfo {title} {Stability of pure
  thermodynamic phases in quantum statistics},\ }\href@noop {} {\bibfield
  {journal} {\bibinfo  {journal} {Commun. Math. Phys}\ }\textbf {\bibinfo
  {volume} {41}},\ \bibinfo {pages} {89} (\bibinfo {year} {1975})}\BibitemShut
  {NoStop}%
\bibitem [{\citenamefont {Prosen}(1998)}]{Tomaz_Prosen_1998}%
  \BibitemOpen
  \bibfield  {author} {\bibinfo {author} {\bibfnamefont {T.}~\bibnamefont
  {Prosen}},\ }\bibfield  {title} {\bibinfo {title} {Quantum invariants of
  motion in a generic many-body system},\ }\href
  {https://doi.org/10.1088/0305-4470/31/37/004} {\bibfield  {journal} {\bibinfo
   {journal} {Journal of Physics A: Mathematical and General}\ }\textbf
  {\bibinfo {volume} {31}},\ \bibinfo {pages} {L645} (\bibinfo {year}
  {1998})}\BibitemShut {NoStop}%
\bibitem [{\citenamefont {Prosen}(1999)}]{ProsenQL2}%
  \BibitemOpen
  \bibfield  {author} {\bibinfo {author} {\bibfnamefont {T.~c.~v.}\
  \bibnamefont {Prosen}},\ }\bibfield  {title} {\bibinfo {title} {Ergodic
  properties of a generic nonintegrable quantum many-body system in the
  thermodynamic limit},\ }\href {https://doi.org/10.1103/PhysRevE.60.3949}
  {\bibfield  {journal} {\bibinfo  {journal} {Phys. Rev. E}\ }\textbf {\bibinfo
  {volume} {60}},\ \bibinfo {pages} {3949} (\bibinfo {year}
  {1999})}\BibitemShut {NoStop}%
\bibitem [{\citenamefont {Prosen}\ and\ \citenamefont {Ilievski}(2013)}]{QL4}%
  \BibitemOpen
  \bibfield  {author} {\bibinfo {author} {\bibfnamefont {T.~c.~v.}\
  \bibnamefont {Prosen}}\ and\ \bibinfo {author} {\bibfnamefont
  {E.}~\bibnamefont {Ilievski}},\ }\bibfield  {title} {\bibinfo {title}
  {Families of quasilocal conservation laws and quantum spin transport},\
  }\href {https://doi.org/10.1103/PhysRevLett.111.057203} {\bibfield  {journal}
  {\bibinfo  {journal} {Phys. Rev. Lett.}\ }\textbf {\bibinfo {volume} {111}},\
  \bibinfo {pages} {057203} (\bibinfo {year} {2013})}\BibitemShut {NoStop}%
\bibitem [{foo()}]{footnote1}%
  \BibitemOpen
  \href@noop {} {}\bibinfo {howpublished} {We call them pseudolocal quantities
  instead of charges because they will not necessarily be conserved in our
  framework and to avoid giving the wrong impression.}\BibitemShut {Stop}%
\bibitem [{\citenamefont {Bu\ifmmode~\check{c}\else \v{c}\fi{}a}\ and\
  \citenamefont {Jaksch}(2019)}]{BucaJaksch2019}%
  \BibitemOpen
  \bibfield  {author} {\bibinfo {author} {\bibfnamefont {B.}~\bibnamefont
  {Bu\ifmmode~\check{c}\else \v{c}\fi{}a}}\ and\ \bibinfo {author}
  {\bibfnamefont {D.}~\bibnamefont {Jaksch}},\ }\bibfield  {title} {\bibinfo
  {title} {Dissipation induced nonstationarity in a quantum gas},\ }\href
  {https://doi.org/10.1103/PhysRevLett.123.260401} {\bibfield  {journal}
  {\bibinfo  {journal} {Phys. Rev. Lett.}\ }\textbf {\bibinfo {volume} {123}},\
  \bibinfo {pages} {260401} (\bibinfo {year} {2019})}\BibitemShut {NoStop}%
\bibitem [{\citenamefont {Medenjak}\ \emph
  {et~al.}(2020{\natexlab{a}})\citenamefont {Medenjak}, \citenamefont
  {Bu\ifmmode~\check{c}\else \v{c}\fi{}a},\ and\ \citenamefont
  {Jaksch}}]{Marko1}%
  \BibitemOpen
  \bibfield  {author} {\bibinfo {author} {\bibfnamefont {M.}~\bibnamefont
  {Medenjak}}, \bibinfo {author} {\bibfnamefont {B.}~\bibnamefont
  {Bu\ifmmode~\check{c}\else \v{c}\fi{}a}},\ and\ \bibinfo {author}
  {\bibfnamefont {D.}~\bibnamefont {Jaksch}},\ }\bibfield  {title} {\bibinfo
  {title} {{Isolated Heisenberg magnet as a quantum time crystal}},\ }\href
  {https://doi.org/10.1103/PhysRevB.102.041117} {\bibfield  {journal} {\bibinfo
   {journal} {Phys. Rev. B}\ }\textbf {\bibinfo {volume} {102}},\ \bibinfo
  {pages} {041117} (\bibinfo {year} {2020}{\natexlab{a}})}\BibitemShut
  {NoStop}%
\bibitem [{\citenamefont {Müller}\ \emph {et~al.}(2015)\citenamefont
  {Müller}, \citenamefont {Adlam}, \citenamefont {Masanes},\ and\
  \citenamefont {Wiebe}}]{Muller_2015}%
  \BibitemOpen
  \bibfield  {author} {\bibinfo {author} {\bibfnamefont {M.~P.}\ \bibnamefont
  {Müller}}, \bibinfo {author} {\bibfnamefont {E.}~\bibnamefont {Adlam}},
  \bibinfo {author} {\bibfnamefont {L.}~\bibnamefont {Masanes}},\ and\ \bibinfo
  {author} {\bibfnamefont {N.}~\bibnamefont {Wiebe}},\ }\bibfield  {title}
  {\bibinfo {title} {Thermalization and canonical typicality in
  translation-invariant quantum lattice systems},\ }\href
  {https://doi.org/10.1007/s00220-015-2473-y} {\bibfield  {journal} {\bibinfo
  {journal} {Communications in Mathematical Physics}\ }\textbf {\bibinfo
  {volume} {340}},\ \bibinfo {pages} {499} (\bibinfo {year}
  {2015})}\BibitemShut {NoStop}%
\bibitem [{\citenamefont {Biroli}\ \emph {et~al.}(2010)\citenamefont {Biroli},
  \citenamefont {Kollath},\ and\ \citenamefont {L\"auchli}}]{ETH1}%
  \BibitemOpen
  \bibfield  {author} {\bibinfo {author} {\bibfnamefont {G.}~\bibnamefont
  {Biroli}}, \bibinfo {author} {\bibfnamefont {C.}~\bibnamefont {Kollath}},\
  and\ \bibinfo {author} {\bibfnamefont {A.~M.}\ \bibnamefont {L\"auchli}},\
  }\bibfield  {title} {\bibinfo {title} {Effect of rare fluctuations on the
  thermalization of isolated quantum systems},\ }\href
  {https://doi.org/10.1103/PhysRevLett.105.250401} {\bibfield  {journal}
  {\bibinfo  {journal} {Phys. Rev. Lett.}\ }\textbf {\bibinfo {volume} {105}},\
  \bibinfo {pages} {250401} (\bibinfo {year} {2010})}\BibitemShut {NoStop}%
\bibitem [{\citenamefont {Iyoda}\ \emph {et~al.}(2017)\citenamefont {Iyoda},
  \citenamefont {Kaneko},\ and\ \citenamefont {Sagawa}}]{ETH2}%
  \BibitemOpen
  \bibfield  {author} {\bibinfo {author} {\bibfnamefont {E.}~\bibnamefont
  {Iyoda}}, \bibinfo {author} {\bibfnamefont {K.}~\bibnamefont {Kaneko}},\ and\
  \bibinfo {author} {\bibfnamefont {T.}~\bibnamefont {Sagawa}},\ }\bibfield
  {title} {\bibinfo {title} {Fluctuation theorem for many-body pure quantum
  states},\ }\href {https://doi.org/10.1103/PhysRevLett.119.100601} {\bibfield
  {journal} {\bibinfo  {journal} {Phys. Rev. Lett.}\ }\textbf {\bibinfo
  {volume} {119}},\ \bibinfo {pages} {100601} (\bibinfo {year}
  {2017})}\BibitemShut {NoStop}%
\bibitem [{\citenamefont {Kuwahara}\ and\ \citenamefont {Saito}(2020)}]{ETH3}%
  \BibitemOpen
  \bibfield  {author} {\bibinfo {author} {\bibfnamefont {T.}~\bibnamefont
  {Kuwahara}}\ and\ \bibinfo {author} {\bibfnamefont {K.}~\bibnamefont
  {Saito}},\ }\bibfield  {title} {\bibinfo {title} {Eigenstate thermalization
  from the clustering property of correlation},\ }\bibfield  {journal}
  {\bibinfo  {journal} {Physical Review Letters}\ }\textbf {\bibinfo {volume}
  {124}},\ \href {https://doi.org/10.1103/physrevlett.124.200604}
  {10.1103/physrevlett.124.200604} (\bibinfo {year} {2020})\BibitemShut
  {NoStop}%
\bibitem [{\citenamefont {Frigerio}(1978)}]{Frigerio}%
  \BibitemOpen
  \bibfield  {author} {\bibinfo {author} {\bibfnamefont {A.}~\bibnamefont
  {Frigerio}},\ }\bibfield  {title} {\bibinfo {title} {Stationary states of
  quantum dynamical semigroups},\ }\href {https://doi.org/10.1007/BF01196936}
  {\bibfield  {journal} {\bibinfo  {journal} {Communications in Mathematical
  Physics}\ }\textbf {\bibinfo {volume} {63}},\ \bibinfo {pages} {269}
  (\bibinfo {year} {1978})}\BibitemShut {NoStop}%
\bibitem [{\citenamefont {Evans}(1977)}]{Evans}%
  \BibitemOpen
  \bibfield  {author} {\bibinfo {author} {\bibfnamefont {D.~E.}\ \bibnamefont
  {Evans}},\ }\bibfield  {title} {\bibinfo {title} {Irreducible quantum
  dynamical semigroups},\ }\href@noop {} {\bibfield  {journal} {\bibinfo
  {journal} {Communications in Mathematical Physics}\ }\textbf {\bibinfo
  {volume} {54}},\ \bibinfo {pages} {293} (\bibinfo {year} {1977})}\BibitemShut
  {NoStop}%
\bibitem [{\citenamefont {Frigerio}\ \emph {et~al.}(1982)\citenamefont
  {Frigerio}, \citenamefont {Verri} \emph {et~al.}}]{Frigerio2}%
  \BibitemOpen
  \bibfield  {author} {\bibinfo {author} {\bibfnamefont {A.}~\bibnamefont
  {Frigerio}}, \bibinfo {author} {\bibfnamefont {M.}~\bibnamefont {Verri}},
  \emph {et~al.},\ }\bibfield  {title} {\bibinfo {title} {Long-time asymptotic
  properties of dynamical semigroups on w*-algebras},\ }\href@noop {}
  {\bibfield  {journal} {\bibinfo  {journal} {Math. Z}\ }\textbf {\bibinfo
  {volume} {180}},\ \bibinfo {pages} {275} (\bibinfo {year}
  {1982})}\BibitemShut {NoStop}%
\bibitem [{\citenamefont {Fagnola}\ and\ \citenamefont
  {Rebolledo}(2008)}]{Fagnola2}%
  \BibitemOpen
  \bibfield  {author} {\bibinfo {author} {\bibfnamefont {F.}~\bibnamefont
  {Fagnola}}\ and\ \bibinfo {author} {\bibfnamefont {R.}~\bibnamefont
  {Rebolledo}},\ }\bibfield  {title} {\bibinfo {title} {Algebraic conditions
  for convergence of a quantum markov semigroup to a steady state},\ }\href
  {https://doi.org/10.1142/S0219025708003142} {\bibfield  {journal} {\bibinfo
  {journal} {Infinite Dimensional Analysis, Quantum Probability and Related
  Topics}\ }\textbf {\bibinfo {volume} {11}},\ \bibinfo {pages} {467} (\bibinfo
  {year} {2008})},\ \Eprint
  {https://arxiv.org/abs/https://doi.org/10.1142/S0219025708003142}
  {https://doi.org/10.1142/S0219025708003142} \BibitemShut {NoStop}%
\bibitem [{\citenamefont {Dhahri}\ \emph {et~al.}(2010)\citenamefont {Dhahri},
  \citenamefont {Fagnola},\ and\ \citenamefont {Rebolledo}}]{Fagnola1}%
  \BibitemOpen
  \bibfield  {author} {\bibinfo {author} {\bibfnamefont {A.}~\bibnamefont
  {Dhahri}}, \bibinfo {author} {\bibfnamefont {F.}~\bibnamefont {Fagnola}},\
  and\ \bibinfo {author} {\bibfnamefont {R.}~\bibnamefont {Rebolledo}},\
  }\bibfield  {title} {\bibinfo {title} {The decoherence-free subalgebra of a
  quantum markov semigroup with unbounded generator},\ }\href
  {https://doi.org/10.1142/S0219025710004176} {\bibfield  {journal} {\bibinfo
  {journal} {Infinite Dimensional Analysis, Quantum Probability and Related
  Topics}\ }\textbf {\bibinfo {volume} {13}},\ \bibinfo {pages} {413} (\bibinfo
  {year} {2010})},\ \Eprint
  {https://arxiv.org/abs/https://doi.org/10.1142/S0219025710004176}
  {https://doi.org/10.1142/S0219025710004176} \BibitemShut {NoStop}%
\bibitem [{\citenamefont {Barut}\ \emph {et~al.}(1988)\citenamefont {Barut},
  \citenamefont {Bohm},\ and\ \citenamefont {Ne'eman}}]{SGA}%
  \BibitemOpen
  \bibfield  {author} {\bibinfo {author} {\bibfnamefont {A.}~\bibnamefont
  {Barut}}, \bibinfo {author} {\bibfnamefont {A.}~\bibnamefont {Bohm}},\ and\
  \bibinfo {author} {\bibfnamefont {Y.}~\bibnamefont {Ne'eman}},\ }\href
  {https://doi.org/10.1142/0299} {\emph {\bibinfo {title} {Dynamical Groups and
  Spectrum Generating Algebras}}}\ (\bibinfo  {publisher} {World Scientific
  Publishing Company},\ \bibinfo {year} {1988})\ \Eprint
  {https://arxiv.org/abs/https://www.worldscientific.com/doi/pdf/10.1142/0299}
  {https://www.worldscientific.com/doi/pdf/10.1142/0299} \BibitemShut {NoStop}%
\bibitem [{\citenamefont {Moudgalya}\ \emph {et~al.}(2022)\citenamefont
  {Moudgalya}, \citenamefont {Bernevig},\ and\ \citenamefont
  {Regnault}}]{SanjayReview}%
  \BibitemOpen
  \bibfield  {author} {\bibinfo {author} {\bibfnamefont {S.}~\bibnamefont
  {Moudgalya}}, \bibinfo {author} {\bibfnamefont {B.~A.}\ \bibnamefont
  {Bernevig}},\ and\ \bibinfo {author} {\bibfnamefont {N.}~\bibnamefont
  {Regnault}},\ }\bibfield  {title} {\bibinfo {title} {Quantum many-body scars
  and hilbert space fragmentation: a review of exact results},\ }\href
  {https://doi.org/10.1088/1361-6633/ac73a0} {\bibfield  {journal} {\bibinfo
  {journal} {Reports on Progress in Physics}\ }\textbf {\bibinfo {volume}
  {85}},\ \bibinfo {pages} {086501} (\bibinfo {year} {2022})}\BibitemShut
  {NoStop}%
\bibitem [{\citenamefont {Moudgalya}\ and\ \citenamefont
  {Motrunich}(2022{\natexlab{a}})}]{SanjayNew}%
  \BibitemOpen
  \bibfield  {author} {\bibinfo {author} {\bibfnamefont {S.}~\bibnamefont
  {Moudgalya}}\ and\ \bibinfo {author} {\bibfnamefont {O.~I.}\ \bibnamefont
  {Motrunich}},\ }\bibfield  {title} {\bibinfo {title} {Hilbert space
  fragmentation and commutant algebras},\ }\bibfield  {journal} {\bibinfo
  {journal} {Physical Review X}\ }\textbf {\bibinfo {volume} {12}},\ \href
  {https://doi.org/10.1103/physrevx.12.011050} {10.1103/physrevx.12.011050}
  (\bibinfo {year} {2022}{\natexlab{a}})\BibitemShut {NoStop}%
\bibitem [{\citenamefont {Moudgalya}\ and\ \citenamefont
  {Motrunich}(2022{\natexlab{b}})}]{SanjaySymCom}%
  \BibitemOpen
  \bibfield  {author} {\bibinfo {author} {\bibfnamefont {S.}~\bibnamefont
  {Moudgalya}}\ and\ \bibinfo {author} {\bibfnamefont {O.~I.}\ \bibnamefont
  {Motrunich}},\ }\href {https://doi.org/10.48550/ARXIV.2209.03370} {\bibinfo
  {title} {From symmetries to commutant algebras in standard hamiltonians}}
  (\bibinfo {year} {2022}{\natexlab{b}})\BibitemShut {NoStop}%
\bibitem [{\citenamefont {Buca}\ \emph {et~al.}(2020)\citenamefont {Buca},
  \citenamefont {Purkayastha}, \citenamefont {Guarnieri}, \citenamefont
  {Mitchison}, \citenamefont {Jaksch},\ and\ \citenamefont {Goold}}]{Buca}%
  \BibitemOpen
  \bibfield  {author} {\bibinfo {author} {\bibfnamefont {B.}~\bibnamefont
  {Buca}}, \bibinfo {author} {\bibfnamefont {A.}~\bibnamefont {Purkayastha}},
  \bibinfo {author} {\bibfnamefont {G.}~\bibnamefont {Guarnieri}}, \bibinfo
  {author} {\bibfnamefont {M.~T.}\ \bibnamefont {Mitchison}}, \bibinfo {author}
  {\bibfnamefont {D.}~\bibnamefont {Jaksch}},\ and\ \bibinfo {author}
  {\bibfnamefont {J.}~\bibnamefont {Goold}},\ }\href@noop {} {\bibinfo {title}
  {Quantum many-body attractors}} (\bibinfo {year} {2020}),\ \Eprint
  {https://arxiv.org/abs/2008.11166} {arXiv:2008.11166 [quant-ph]} \BibitemShut
  {NoStop}%
\bibitem [{\citenamefont {Sala}\ \emph {et~al.}(2021)\citenamefont {Sala},
  \citenamefont {Lehmann}, \citenamefont {Rakovszky},\ and\ \citenamefont
  {Pollmann}}]{Pollmann}%
  \BibitemOpen
  \bibfield  {author} {\bibinfo {author} {\bibfnamefont {P.}~\bibnamefont
  {Sala}}, \bibinfo {author} {\bibfnamefont {J.}~\bibnamefont {Lehmann}},
  \bibinfo {author} {\bibfnamefont {T.}~\bibnamefont {Rakovszky}},\ and\
  \bibinfo {author} {\bibfnamefont {F.}~\bibnamefont {Pollmann}},\ }\href@noop
  {} {\bibinfo {title} {Dynamics in systems with modulated symmetries}}
  (\bibinfo {year} {2021}),\ \Eprint {https://arxiv.org/abs/2110.08302}
  {arXiv:2110.08302 [cond-mat.stat-mech]} \BibitemShut {NoStop}%
\bibitem [{\citenamefont {Haake}(1991)}]{haake1991quantum}%
  \BibitemOpen
  \bibfield  {author} {\bibinfo {author} {\bibfnamefont {F.}~\bibnamefont
  {Haake}},\ }\bibfield  {title} {\bibinfo {title} {Quantum signatures of
  chaos},\ }in\ \href@noop {} {\emph {\bibinfo {booktitle} {Quantum Coherence
  in Mesoscopic Systems}}}\ (\bibinfo  {publisher} {Springer},\ \bibinfo {year}
  {1991})\ pp.\ \bibinfo {pages} {583--595}\BibitemShut {NoStop}%
\bibitem [{\citenamefont {Barthel}\ and\ \citenamefont
  {Schollw\"ock}(2008)}]{eigenstatedephasing}%
  \BibitemOpen
  \bibfield  {author} {\bibinfo {author} {\bibfnamefont {T.}~\bibnamefont
  {Barthel}}\ and\ \bibinfo {author} {\bibfnamefont {U.}~\bibnamefont
  {Schollw\"ock}},\ }\bibfield  {title} {\bibinfo {title} {Dephasing and the
  steady state in quantum many-particle systems},\ }\href
  {https://doi.org/10.1103/PhysRevLett.100.100601} {\bibfield  {journal}
  {\bibinfo  {journal} {Phys. Rev. Lett.}\ }\textbf {\bibinfo {volume} {100}},\
  \bibinfo {pages} {100601} (\bibinfo {year} {2008})}\BibitemShut {NoStop}%
\bibitem [{\citenamefont {Mazur}(1969)}]{Mazur}%
  \BibitemOpen
  \bibfield  {author} {\bibinfo {author} {\bibfnamefont {P.}~\bibnamefont
  {Mazur}},\ }\bibfield  {title} {\bibinfo {title} {Non-ergodicity of phase
  functions in certain systems},\ }\href
  {https://doi.org/https://doi.org/10.1016/0031-8914(69)90185-2} {\bibfield
  {journal} {\bibinfo  {journal} {Physica}\ }\textbf {\bibinfo {volume} {43}},\
  \bibinfo {pages} {533} (\bibinfo {year} {1969})}\BibitemShut {NoStop}%
\bibitem [{\citenamefont {Abanin}\ and\ \citenamefont
  {Papi\'{c}}(2017)}]{Abanin_MBL}%
  \BibitemOpen
  \bibfield  {author} {\bibinfo {author} {\bibfnamefont {D.~A.}\ \bibnamefont
  {Abanin}}\ and\ \bibinfo {author} {\bibfnamefont {Z.}~\bibnamefont
  {Papi\'{c}}},\ }\bibfield  {title} {\bibinfo {title} {Recent progress in
  many-body localization},\ }\href
  {https://doi.org/https://doi.org/10.1002/andp.201700169} {\bibfield
  {journal} {\bibinfo  {journal} {Annalen der Physik}\ }\textbf {\bibinfo
  {volume} {529}},\ \bibinfo {pages} {1700169} (\bibinfo {year}
  {2017})}\BibitemShut {NoStop}%
\bibitem [{\citenamefont {Howland}(1974)}]{howland1974stationary}%
  \BibitemOpen
  \bibfield  {author} {\bibinfo {author} {\bibfnamefont {J.~S.}\ \bibnamefont
  {Howland}},\ }\bibfield  {title} {\bibinfo {title} {Stationary scattering
  theory for time-dependent hamiltonians},\ }\href@noop {} {\bibfield
  {journal} {\bibinfo  {journal} {Mathematische Annalen}\ }\textbf {\bibinfo
  {volume} {207}},\ \bibinfo {pages} {315} (\bibinfo {year}
  {1974})}\BibitemShut {NoStop}%
\bibitem [{\citenamefont {Abanin}\ \emph {et~al.}(2017)\citenamefont {Abanin},
  \citenamefont {De~Roeck}, \citenamefont {Ho},\ and\ \citenamefont
  {Huveneers}}]{pre-therm}%
  \BibitemOpen
  \bibfield  {author} {\bibinfo {author} {\bibfnamefont {D.}~\bibnamefont
  {Abanin}}, \bibinfo {author} {\bibfnamefont {W.}~\bibnamefont {De~Roeck}},
  \bibinfo {author} {\bibfnamefont {W.~W.}\ \bibnamefont {Ho}},\ and\ \bibinfo
  {author} {\bibfnamefont {F.}~\bibnamefont {Huveneers}},\ }\bibfield  {title}
  {\bibinfo {title} {A rigorous theory of many-body prethermalization for
  periodically driven and closed quantum systems},\ }\href
  {https://doi.org/10.1007/s00220-017-2930-x} {\bibfield  {journal} {\bibinfo
  {journal} {Communications in Mathematical Physics}\ }\textbf {\bibinfo
  {volume} {354}},\ \bibinfo {pages} {809–827} (\bibinfo {year}
  {2017})}\BibitemShut {NoStop}%
\bibitem [{\citenamefont {Iadecola}\ and\ \citenamefont
  {Schecter}(2020)}]{scars5}%
  \BibitemOpen
  \bibfield  {author} {\bibinfo {author} {\bibfnamefont {T.}~\bibnamefont
  {Iadecola}}\ and\ \bibinfo {author} {\bibfnamefont {M.}~\bibnamefont
  {Schecter}},\ }\bibfield  {title} {\bibinfo {title} {Quantum many-body scar
  states with emergent kinetic constraints and finite-entanglement revivals},\
  }\href {http://dx.doi.org/10.1103/PhysRevB.101.024306} {\bibfield  {journal}
  {\bibinfo  {journal} {Phys. Rev. B}\ }\textbf {\bibinfo {volume} {101}}
  (\bibinfo {year} {2020})}\BibitemShut {NoStop}%
\bibitem [{\citenamefont {Moudgalya}\ \emph {et~al.}(2020)\citenamefont
  {Moudgalya}, \citenamefont {Regnault},\ and\ \citenamefont
  {Bernevig}}]{scarsdynsym2}%
  \BibitemOpen
  \bibfield  {author} {\bibinfo {author} {\bibfnamefont {S.}~\bibnamefont
  {Moudgalya}}, \bibinfo {author} {\bibfnamefont {N.}~\bibnamefont
  {Regnault}},\ and\ \bibinfo {author} {\bibfnamefont {B.~A.}\ \bibnamefont
  {Bernevig}},\ }\bibfield  {title} {\bibinfo {title}
  {{$\ensuremath{\eta}$-pairing in Hubbard models: From spectrum generating
  algebras to quantum many-body scars}},\ }\href
  {https://doi.org/10.1103/PhysRevB.102.085140} {\bibfield  {journal} {\bibinfo
   {journal} {Phys. Rev. B}\ }\textbf {\bibinfo {volume} {102}},\ \bibinfo
  {pages} {085140} (\bibinfo {year} {2020})}\BibitemShut {NoStop}%
\bibitem [{\citenamefont {Pakrouski}\ \emph {et~al.}(2020)\citenamefont
  {Pakrouski}, \citenamefont {Pallegar}, \citenamefont {Popov},\ and\
  \citenamefont {Klebanov}}]{scarsdynsym6}%
  \BibitemOpen
  \bibfield  {author} {\bibinfo {author} {\bibfnamefont {K.}~\bibnamefont
  {Pakrouski}}, \bibinfo {author} {\bibfnamefont {P.~N.}\ \bibnamefont
  {Pallegar}}, \bibinfo {author} {\bibfnamefont {F.~K.}\ \bibnamefont
  {Popov}},\ and\ \bibinfo {author} {\bibfnamefont {I.~R.}\ \bibnamefont
  {Klebanov}},\ }\bibfield  {title} {\bibinfo {title} {{Many-Body Scars as a
  Group Invariant Sector of Hilbert Space}},\ }\href
  {https://doi.org/10.1103/PhysRevLett.125.230602} {\bibfield  {journal}
  {\bibinfo  {journal} {Phys. Rev. Lett.}\ }\textbf {\bibinfo {volume} {125}},\
  \bibinfo {pages} {230602} (\bibinfo {year} {2020})}\BibitemShut {NoStop}%
\bibitem [{\citenamefont {Nakagawa}\ \emph {et~al.}(2022)\citenamefont
  {Nakagawa}, \citenamefont {Katsura},\ and\ \citenamefont
  {Ueda}}]{HoshoScars3}%
  \BibitemOpen
  \bibfield  {author} {\bibinfo {author} {\bibfnamefont {M.}~\bibnamefont
  {Nakagawa}}, \bibinfo {author} {\bibfnamefont {H.}~\bibnamefont {Katsura}},\
  and\ \bibinfo {author} {\bibfnamefont {M.}~\bibnamefont {Ueda}},\ }\href
  {https://doi.org/10.48550/ARXIV.2205.07235} {\bibinfo {title} {Exact
  eigenstates of multicomponent hubbard models: Su($n$) magnetic $\eta$
  pairing, weak ergodicity breaking, and partial integrability}} (\bibinfo
  {year} {2022})\BibitemShut {NoStop}%
\bibitem [{\citenamefont {Sala}\ \emph
  {et~al.}(2020{\natexlab{b}})\citenamefont {Sala}, \citenamefont {Rakovszky},
  \citenamefont {Verresen}, \citenamefont {Knap},\ and\ \citenamefont
  {Pollmann}}]{fragmentation3}%
  \BibitemOpen
  \bibfield  {author} {\bibinfo {author} {\bibfnamefont {P.}~\bibnamefont
  {Sala}}, \bibinfo {author} {\bibfnamefont {T.}~\bibnamefont {Rakovszky}},
  \bibinfo {author} {\bibfnamefont {R.}~\bibnamefont {Verresen}}, \bibinfo
  {author} {\bibfnamefont {M.}~\bibnamefont {Knap}},\ and\ \bibinfo {author}
  {\bibfnamefont {F.}~\bibnamefont {Pollmann}},\ }\bibfield  {title} {\bibinfo
  {title} {Ergodicity breaking arising from hilbert space fragmentation in
  dipole-conserving hamiltonians},\ }\href
  {https://doi.org/10.1103/PhysRevX.10.011047} {\bibfield  {journal} {\bibinfo
  {journal} {Phys. Rev. X}\ }\textbf {\bibinfo {volume} {10}},\ \bibinfo
  {pages} {011047} (\bibinfo {year} {2020}{\natexlab{b}})}\BibitemShut
  {NoStop}%
\bibitem [{\citenamefont {Bu\ifmmode~\check{c}\else
  \v{c}\fi{}a}(2022)}]{strictlylocalfrag}%
  \BibitemOpen
  \bibfield  {author} {\bibinfo {author} {\bibfnamefont {B.}~\bibnamefont
  {Bu\ifmmode~\check{c}\else \v{c}\fi{}a}},\ }\bibfield  {title} {\bibinfo
  {title} {Out-of-time-ordered crystals and fragmentation},\ }\href
  {https://doi.org/10.1103/PhysRevLett.128.100601} {\bibfield  {journal}
  {\bibinfo  {journal} {Phys. Rev. Lett.}\ }\textbf {\bibinfo {volume} {128}},\
  \bibinfo {pages} {100601} (\bibinfo {year} {2022})}\BibitemShut {NoStop}%
\bibitem [{\citenamefont {Mukherjee}\ \emph {et~al.}(2021)\citenamefont
  {Mukherjee}, \citenamefont {Banerjee}, \citenamefont {Sengupta},\ and\
  \citenamefont {Sen}}]{ArnabFragmentation}%
  \BibitemOpen
  \bibfield  {author} {\bibinfo {author} {\bibfnamefont {B.}~\bibnamefont
  {Mukherjee}}, \bibinfo {author} {\bibfnamefont {D.}~\bibnamefont {Banerjee}},
  \bibinfo {author} {\bibfnamefont {K.}~\bibnamefont {Sengupta}},\ and\
  \bibinfo {author} {\bibfnamefont {A.}~\bibnamefont {Sen}},\ }\bibfield
  {title} {\bibinfo {title} {Minimal model for hilbert space fragmentation with
  local constraints},\ }\bibfield  {journal} {\bibinfo  {journal} {Physical
  Review B}\ }\textbf {\bibinfo {volume} {104}},\ \href
  {https://doi.org/10.1103/physrevb.104.155117} {10.1103/physrevb.104.155117}
  (\bibinfo {year} {2021})\BibitemShut {NoStop}%
\bibitem [{\citenamefont {Andreadakis}\ and\ \citenamefont
  {Zanardi}(2023)}]{ZanardiFRAG}%
  \BibitemOpen
  \bibfield  {author} {\bibinfo {author} {\bibfnamefont {F.}~\bibnamefont
  {Andreadakis}}\ and\ \bibinfo {author} {\bibfnamefont {P.}~\bibnamefont
  {Zanardi}},\ }\bibfield  {title} {\bibinfo {title} {Coherence generation,
  symmetry algebras, and hilbert space fragmentation},\ }\href
  {https://doi.org/10.1103/PhysRevA.107.062402} {\bibfield  {journal} {\bibinfo
   {journal} {Phys. Rev. A}\ }\textbf {\bibinfo {volume} {107}},\ \bibinfo
  {pages} {062402} (\bibinfo {year} {2023})}\BibitemShut {NoStop}%
\bibitem [{\citenamefont {Zohar}\ \emph {et~al.}(2017)\citenamefont {Zohar},
  \citenamefont {Farace}, \citenamefont {Reznik},\ and\ \citenamefont
  {Cirac}}]{LGT1}%
  \BibitemOpen
  \bibfield  {author} {\bibinfo {author} {\bibfnamefont {E.}~\bibnamefont
  {Zohar}}, \bibinfo {author} {\bibfnamefont {A.}~\bibnamefont {Farace}},
  \bibinfo {author} {\bibfnamefont {B.}~\bibnamefont {Reznik}},\ and\ \bibinfo
  {author} {\bibfnamefont {J.~I.}\ \bibnamefont {Cirac}},\ }\bibfield  {title}
  {\bibinfo {title} {Digital quantum simulation of ${\mathbb{z}}_{2}$ lattice
  gauge theories with dynamical fermionic matter},\ }\href
  {https://doi.org/10.1103/PhysRevLett.118.070501} {\bibfield  {journal}
  {\bibinfo  {journal} {Phys. Rev. Lett.}\ }\textbf {\bibinfo {volume} {118}},\
  \bibinfo {pages} {070501} (\bibinfo {year} {2017})}\BibitemShut {NoStop}%
\bibitem [{\citenamefont {Borla}\ \emph {et~al.}(2020)\citenamefont {Borla},
  \citenamefont {Verresen}, \citenamefont {Grusdt},\ and\ \citenamefont
  {Moroz}}]{LGT2}%
  \BibitemOpen
  \bibfield  {author} {\bibinfo {author} {\bibfnamefont {U.}~\bibnamefont
  {Borla}}, \bibinfo {author} {\bibfnamefont {R.}~\bibnamefont {Verresen}},
  \bibinfo {author} {\bibfnamefont {F.}~\bibnamefont {Grusdt}},\ and\ \bibinfo
  {author} {\bibfnamefont {S.}~\bibnamefont {Moroz}},\ }\bibfield  {title}
  {\bibinfo {title} {Confined phases of one-dimensional spinless fermions
  coupled to ${Z}_{2}$ gauge theory},\ }\href
  {https://doi.org/10.1103/PhysRevLett.124.120503} {\bibfield  {journal}
  {\bibinfo  {journal} {Phys. Rev. Lett.}\ }\textbf {\bibinfo {volume} {124}},\
  \bibinfo {pages} {120503} (\bibinfo {year} {2020})}\BibitemShut {NoStop}%
\bibitem [{\citenamefont {Yang}\ \emph
  {et~al.}(2020{\natexlab{a}})\citenamefont {Yang}, \citenamefont {Liu},
  \citenamefont {Gorshkov},\ and\ \citenamefont {Iadecola}}]{LGT3}%
  \BibitemOpen
  \bibfield  {author} {\bibinfo {author} {\bibfnamefont {Z.-C.}\ \bibnamefont
  {Yang}}, \bibinfo {author} {\bibfnamefont {F.}~\bibnamefont {Liu}}, \bibinfo
  {author} {\bibfnamefont {A.~V.}\ \bibnamefont {Gorshkov}},\ and\ \bibinfo
  {author} {\bibfnamefont {T.}~\bibnamefont {Iadecola}},\ }\bibfield  {title}
  {\bibinfo {title} {Hilbert-space fragmentation from strict confinement},\
  }\href {https://doi.org/10.1103/PhysRevLett.124.207602} {\bibfield  {journal}
  {\bibinfo  {journal} {Phys. Rev. Lett.}\ }\textbf {\bibinfo {volume} {124}},\
  \bibinfo {pages} {207602} (\bibinfo {year} {2020}{\natexlab{a}})}\BibitemShut
  {NoStop}%
\bibitem [{\citenamefont {Halimeh}\ \emph
  {et~al.}(2022{\natexlab{b}})\citenamefont {Halimeh}, \citenamefont {Homeier},
  \citenamefont {Zhao}, \citenamefont {Bohrdt}, \citenamefont {Grusdt},
  \citenamefont {Hauke},\ and\ \citenamefont {Knolle}}]{Jad4}%
  \BibitemOpen
  \bibfield  {author} {\bibinfo {author} {\bibfnamefont {J.~C.}\ \bibnamefont
  {Halimeh}}, \bibinfo {author} {\bibfnamefont {L.}~\bibnamefont {Homeier}},
  \bibinfo {author} {\bibfnamefont {H.}~\bibnamefont {Zhao}}, \bibinfo {author}
  {\bibfnamefont {A.}~\bibnamefont {Bohrdt}}, \bibinfo {author} {\bibfnamefont
  {F.}~\bibnamefont {Grusdt}}, \bibinfo {author} {\bibfnamefont
  {P.}~\bibnamefont {Hauke}},\ and\ \bibinfo {author} {\bibfnamefont
  {J.}~\bibnamefont {Knolle}},\ }\bibfield  {title} {\bibinfo {title}
  {Enhancing disorder-free localization through dynamically emergent local
  symmetries},\ }\href {https://doi.org/10.1103/PRXQuantum.3.020345} {\bibfield
   {journal} {\bibinfo  {journal} {PRX Quantum}\ }\textbf {\bibinfo {volume}
  {3}},\ \bibinfo {pages} {020345} (\bibinfo {year}
  {2022}{\natexlab{b}})}\BibitemShut {NoStop}%
\bibitem [{\citenamefont {Yang}\ \emph
  {et~al.}(2020{\natexlab{b}})\citenamefont {Yang}, \citenamefont {Sun},
  \citenamefont {Ott}, \citenamefont {Wang}, \citenamefont {Zache},
  \citenamefont {Halimeh}, \citenamefont {Yuan}, \citenamefont {Hauke},\ and\
  \citenamefont {Pan}}]{Jad5}%
  \BibitemOpen
  \bibfield  {author} {\bibinfo {author} {\bibfnamefont {B.}~\bibnamefont
  {Yang}}, \bibinfo {author} {\bibfnamefont {H.}~\bibnamefont {Sun}}, \bibinfo
  {author} {\bibfnamefont {R.}~\bibnamefont {Ott}}, \bibinfo {author}
  {\bibfnamefont {H.-Y.}\ \bibnamefont {Wang}}, \bibinfo {author}
  {\bibfnamefont {T.~V.}\ \bibnamefont {Zache}}, \bibinfo {author}
  {\bibfnamefont {J.~C.}\ \bibnamefont {Halimeh}}, \bibinfo {author}
  {\bibfnamefont {Z.-S.}\ \bibnamefont {Yuan}}, \bibinfo {author}
  {\bibfnamefont {P.}~\bibnamefont {Hauke}},\ and\ \bibinfo {author}
  {\bibfnamefont {J.-W.}\ \bibnamefont {Pan}},\ }\bibfield  {title} {\bibinfo
  {title} {Observation of gauge invariance in a 71-site bose--hubbard quantum
  simulator},\ }\href {https://doi.org/10.1038/s41586-020-2910-8} {\bibfield
  {journal} {\bibinfo  {journal} {Nature}\ }\textbf {\bibinfo {volume} {587}},\
  \bibinfo {pages} {392} (\bibinfo {year} {2020}{\natexlab{b}})}\BibitemShut
  {NoStop}%
\bibitem [{\citenamefont {Magnifico}\ \emph {et~al.}(2020)\citenamefont
  {Magnifico}, \citenamefont {Dalmonte}, \citenamefont {Facchi}, \citenamefont
  {Pascazio}, \citenamefont {Pepe},\ and\ \citenamefont
  {Ercolessi}}]{Magnifico2020realtimedynamics}%
  \BibitemOpen
  \bibfield  {author} {\bibinfo {author} {\bibfnamefont {G.}~\bibnamefont
  {Magnifico}}, \bibinfo {author} {\bibfnamefont {M.}~\bibnamefont {Dalmonte}},
  \bibinfo {author} {\bibfnamefont {P.}~\bibnamefont {Facchi}}, \bibinfo
  {author} {\bibfnamefont {S.}~\bibnamefont {Pascazio}}, \bibinfo {author}
  {\bibfnamefont {F.~V.}\ \bibnamefont {Pepe}},\ and\ \bibinfo {author}
  {\bibfnamefont {E.}~\bibnamefont {Ercolessi}},\ }\bibfield  {title} {\bibinfo
  {title} {Real {T}ime {D}ynamics and {C}onfinement in the {$\mathbb{Z}_{n}$}
  {S}chwinger-{W}eyl lattice model for 1+1 {QED}},\ }\href
  {https://doi.org/10.22331/q-2020-06-15-281} {\bibfield  {journal} {\bibinfo
  {journal} {{Quantum}}\ }\textbf {\bibinfo {volume} {4}},\ \bibinfo {pages}
  {281} (\bibinfo {year} {2020})}\BibitemShut {NoStop}%
\bibitem [{\citenamefont {Nyhegn}\ \emph {et~al.}(2021)\citenamefont {Nyhegn},
  \citenamefont {Chung},\ and\ \citenamefont {Burrello}}]{Michele1}%
  \BibitemOpen
  \bibfield  {author} {\bibinfo {author} {\bibfnamefont {J.}~\bibnamefont
  {Nyhegn}}, \bibinfo {author} {\bibfnamefont {C.-M.}\ \bibnamefont {Chung}},\
  and\ \bibinfo {author} {\bibfnamefont {M.}~\bibnamefont {Burrello}},\
  }\bibfield  {title} {\bibinfo {title} {${\mathbb{z}}_{N}$ lattice gauge
  theory in a ladder geometry},\ }\href
  {https://doi.org/10.1103/PhysRevResearch.3.013133} {\bibfield  {journal}
  {\bibinfo  {journal} {Phys. Rev. Res.}\ }\textbf {\bibinfo {volume} {3}},\
  \bibinfo {pages} {013133} (\bibinfo {year} {2021})}\BibitemShut {NoStop}%
\bibitem [{\citenamefont {Lumia}\ \emph {et~al.}(2022)\citenamefont {Lumia},
  \citenamefont {Torta}, \citenamefont {Mbeng}, \citenamefont {Santoro},
  \citenamefont {Ercolessi}, \citenamefont {Burrello},\ and\ \citenamefont
  {Wauters}}]{Michele2}%
  \BibitemOpen
  \bibfield  {author} {\bibinfo {author} {\bibfnamefont {L.}~\bibnamefont
  {Lumia}}, \bibinfo {author} {\bibfnamefont {P.}~\bibnamefont {Torta}},
  \bibinfo {author} {\bibfnamefont {G.~B.}\ \bibnamefont {Mbeng}}, \bibinfo
  {author} {\bibfnamefont {G.~E.}\ \bibnamefont {Santoro}}, \bibinfo {author}
  {\bibfnamefont {E.}~\bibnamefont {Ercolessi}}, \bibinfo {author}
  {\bibfnamefont {M.}~\bibnamefont {Burrello}},\ and\ \bibinfo {author}
  {\bibfnamefont {M.~M.}\ \bibnamefont {Wauters}},\ }\bibfield  {title}
  {\bibinfo {title} {Two-dimensional ${\mathbb{z}}_{2}$ lattice gauge theory on
  a near-term quantum simulator: Variational quantum optimization, confinement,
  and topological order},\ }\href {https://doi.org/10.1103/PRXQuantum.3.020320}
  {\bibfield  {journal} {\bibinfo  {journal} {PRX Quantum}\ }\textbf {\bibinfo
  {volume} {3}},\ \bibinfo {pages} {020320} (\bibinfo {year}
  {2022})}\BibitemShut {NoStop}%
\bibitem [{\citenamefont {Lazarides}\ \emph {et~al.}(2014)\citenamefont
  {Lazarides}, \citenamefont {Das},\ and\ \citenamefont
  {Moessner}}]{timecrystal2}%
  \BibitemOpen
  \bibfield  {author} {\bibinfo {author} {\bibfnamefont {A.}~\bibnamefont
  {Lazarides}}, \bibinfo {author} {\bibfnamefont {A.}~\bibnamefont {Das}},\
  and\ \bibinfo {author} {\bibfnamefont {R.}~\bibnamefont {Moessner}},\
  }\bibfield  {title} {\bibinfo {title} {Equilibrium states of generic quantum
  systems subject to periodic driving},\ }\href
  {https://doi.org/10.1103/PhysRevE.90.012110} {\bibfield  {journal} {\bibinfo
  {journal} {Phys. Rev. E}\ }\textbf {\bibinfo {volume} {90}},\ \bibinfo
  {pages} {012110} (\bibinfo {year} {2014})}\BibitemShut {NoStop}%
\bibitem [{\citenamefont {Seibold}\ \emph {et~al.}(2020)\citenamefont
  {Seibold}, \citenamefont {Rota},\ and\ \citenamefont {Savona}}]{Seibold}%
  \BibitemOpen
  \bibfield  {author} {\bibinfo {author} {\bibfnamefont {K.}~\bibnamefont
  {Seibold}}, \bibinfo {author} {\bibfnamefont {R.}~\bibnamefont {Rota}},\ and\
  \bibinfo {author} {\bibfnamefont {V.}~\bibnamefont {Savona}},\ }\bibfield
  {title} {\bibinfo {title} {Dissipative time crystal in an asymmetric
  nonlinear photonic dimer},\ }\href@noop {} {\bibfield  {journal} {\bibinfo
  {journal} {Physical Review A}\ }\textbf {\bibinfo {volume} {101}},\ \bibinfo
  {pages} {033839} (\bibinfo {year} {2020})}\BibitemShut {NoStop}%
\bibitem [{\citenamefont {Ke\ss{}ler}\ \emph {et~al.}(2021)\citenamefont
  {Ke\ss{}ler}, \citenamefont {Kongkhambut}, \citenamefont {Georges},
  \citenamefont {Mathey}, \citenamefont {Cosme},\ and\ \citenamefont
  {Hemmerich}}]{dissipativeTCobs}%
  \BibitemOpen
  \bibfield  {author} {\bibinfo {author} {\bibfnamefont {H.}~\bibnamefont
  {Ke\ss{}ler}}, \bibinfo {author} {\bibfnamefont {P.}~\bibnamefont
  {Kongkhambut}}, \bibinfo {author} {\bibfnamefont {C.}~\bibnamefont
  {Georges}}, \bibinfo {author} {\bibfnamefont {L.}~\bibnamefont {Mathey}},
  \bibinfo {author} {\bibfnamefont {J.~G.}\ \bibnamefont {Cosme}},\ and\
  \bibinfo {author} {\bibfnamefont {A.}~\bibnamefont {Hemmerich}},\ }\bibfield
  {title} {\bibinfo {title} {{Observation of a Dissipative Time Crystal}},\
  }\href {https://doi.org/10.1103/PhysRevLett.127.043602} {\bibfield  {journal}
  {\bibinfo  {journal} {Phys. Rev. Lett.}\ }\textbf {\bibinfo {volume} {127}},\
  \bibinfo {pages} {043602} (\bibinfo {year} {2021})}\BibitemShut {NoStop}%
\bibitem [{\citenamefont {Gong}\ \emph {et~al.}(2018)\citenamefont {Gong},
  \citenamefont {Hamazaki},\ and\ \citenamefont {Ueda}}]{UedaTC}%
  \BibitemOpen
  \bibfield  {author} {\bibinfo {author} {\bibfnamefont {Z.}~\bibnamefont
  {Gong}}, \bibinfo {author} {\bibfnamefont {R.}~\bibnamefont {Hamazaki}},\
  and\ \bibinfo {author} {\bibfnamefont {M.}~\bibnamefont {Ueda}},\ }\bibfield
  {title} {\bibinfo {title} {Discrete time-crystalline order in cavity and
  circuit qed systems},\ }\bibfield  {journal} {\bibinfo  {journal} {Physical
  Review Letters}\ }\textbf {\bibinfo {volume} {120}},\ \href
  {https://doi.org/10.1103/physrevlett.120.040404}
  {10.1103/physrevlett.120.040404} (\bibinfo {year} {2018})\BibitemShut
  {NoStop}%
\bibitem [{\citenamefont {Buca}\ \emph {et~al.}(2021)\citenamefont {Buca},
  \citenamefont {Booker},\ and\ \citenamefont {Jaksch}}]{buca2021algebraic}%
  \BibitemOpen
  \bibfield  {author} {\bibinfo {author} {\bibfnamefont {B.}~\bibnamefont
  {Buca}}, \bibinfo {author} {\bibfnamefont {C.}~\bibnamefont {Booker}},\ and\
  \bibinfo {author} {\bibfnamefont {D.}~\bibnamefont {Jaksch}},\ }\href@noop {}
  {\bibinfo {title} {Algebraic theory of quantum synchronization and limit
  cycles under dissipation}} (\bibinfo {year} {2021}),\ \Eprint
  {https://arxiv.org/abs/2103.01808} {arXiv:2103.01808 [quant-ph]} \BibitemShut
  {NoStop}%
\bibitem [{\citenamefont {Bu\v{c}a}\ \emph {et~al.}(2019)\citenamefont
  {Bu\v{c}a}, \citenamefont {Tindall},\ and\ \citenamefont {Jaksch}}]{Buca2}%
  \BibitemOpen
  \bibfield  {author} {\bibinfo {author} {\bibfnamefont {B.}~\bibnamefont
  {Bu\v{c}a}}, \bibinfo {author} {\bibfnamefont {J.}~\bibnamefont {Tindall}},\
  and\ \bibinfo {author} {\bibfnamefont {D.}~\bibnamefont {Jaksch}},\
  }\bibfield  {title} {\bibinfo {title} {Non-stationary coherent quantum
  many-body dynamics through dissipation},\ }\bibfield  {journal} {\bibinfo
  {journal} {Nat Commun}\ }\textbf {\bibinfo {volume} {10}},\ \href
  {https://doi.org/10.1038/s41467-019-09757-y} {10.1038/s41467-019-09757-y}
  (\bibinfo {year} {2019})\BibitemShut {NoStop}%
\bibitem [{\citenamefont {Wang}\ and\ \citenamefont {Wang}(2023)}]{WangWang}%
  \BibitemOpen
  \bibfield  {author} {\bibinfo {author} {\bibfnamefont {H.}~\bibnamefont
  {Wang}}\ and\ \bibinfo {author} {\bibfnamefont {J.}~\bibnamefont {Wang}},\
  }\bibfield  {title} {\bibinfo {title} {Time-crystal phase emerging from a
  qubit network under unitary random operations},\ }\href
  {https://doi.org/10.1103/PhysRevA.108.012209} {\bibfield  {journal} {\bibinfo
   {journal} {Phys. Rev. A}\ }\textbf {\bibinfo {volume} {108}},\ \bibinfo
  {pages} {012209} (\bibinfo {year} {2023})}\BibitemShut {NoStop}%
\bibitem [{\citenamefont {Rakovszky}\ \emph {et~al.}(2020)\citenamefont
  {Rakovszky}, \citenamefont {Sala}, \citenamefont {Verresen}, \citenamefont
  {Knap},\ and\ \citenamefont {Pollmann}}]{SLIOM}%
  \BibitemOpen
  \bibfield  {author} {\bibinfo {author} {\bibfnamefont {T.}~\bibnamefont
  {Rakovszky}}, \bibinfo {author} {\bibfnamefont {P.}~\bibnamefont {Sala}},
  \bibinfo {author} {\bibfnamefont {R.}~\bibnamefont {Verresen}}, \bibinfo
  {author} {\bibfnamefont {M.}~\bibnamefont {Knap}},\ and\ \bibinfo {author}
  {\bibfnamefont {F.}~\bibnamefont {Pollmann}},\ }\bibfield  {title} {\bibinfo
  {title} {Statistical localization: From strong fragmentation to strong edge
  modes},\ }\href {https://doi.org/10.1103/PhysRevB.101.125126} {\bibfield
  {journal} {\bibinfo  {journal} {Phys. Rev. B}\ }\textbf {\bibinfo {volume}
  {101}},\ \bibinfo {pages} {125126} (\bibinfo {year} {2020})}\BibitemShut
  {NoStop}%
\bibitem [{\citenamefont {Morningstar}\ \emph {et~al.}(2020)\citenamefont
  {Morningstar}, \citenamefont {Khemani},\ and\ \citenamefont
  {Huse}}]{Morningstar}%
  \BibitemOpen
  \bibfield  {author} {\bibinfo {author} {\bibfnamefont {A.}~\bibnamefont
  {Morningstar}}, \bibinfo {author} {\bibfnamefont {V.}~\bibnamefont
  {Khemani}},\ and\ \bibinfo {author} {\bibfnamefont {D.~A.}\ \bibnamefont
  {Huse}},\ }\bibfield  {title} {\bibinfo {title} {Kinetically constrained
  freezing transition in a dipole-conserving system},\ }\href
  {https://doi.org/10.1103/PhysRevB.101.214205} {\bibfield  {journal} {\bibinfo
   {journal} {Phys. Rev. B}\ }\textbf {\bibinfo {volume} {101}},\ \bibinfo
  {pages} {214205} (\bibinfo {year} {2020})}\BibitemShut {NoStop}%
\bibitem [{\citenamefont {Pozderac}\ \emph {et~al.}(2023)\citenamefont
  {Pozderac}, \citenamefont {Speck}, \citenamefont {Feng}, \citenamefont
  {Huse},\ and\ \citenamefont {Skinner}}]{Pozderac}%
  \BibitemOpen
  \bibfield  {author} {\bibinfo {author} {\bibfnamefont {C.}~\bibnamefont
  {Pozderac}}, \bibinfo {author} {\bibfnamefont {S.}~\bibnamefont {Speck}},
  \bibinfo {author} {\bibfnamefont {X.}~\bibnamefont {Feng}}, \bibinfo {author}
  {\bibfnamefont {D.~A.}\ \bibnamefont {Huse}},\ and\ \bibinfo {author}
  {\bibfnamefont {B.}~\bibnamefont {Skinner}},\ }\bibfield  {title} {\bibinfo
  {title} {Exact solution for the filling-induced thermalization transition in
  a one-dimensional fracton system},\ }\href
  {https://doi.org/10.1103/PhysRevB.107.045137} {\bibfield  {journal} {\bibinfo
   {journal} {Phys. Rev. B}\ }\textbf {\bibinfo {volume} {107}},\ \bibinfo
  {pages} {045137} (\bibinfo {year} {2023})}\BibitemShut {NoStop}%
\bibitem [{\citenamefont {Lindblad}(1976)}]{Lindblad}%
  \BibitemOpen
  \bibfield  {author} {\bibinfo {author} {\bibfnamefont {G.}~\bibnamefont
  {Lindblad}},\ }\bibfield  {title} {\bibinfo {title} {On the generators of
  quantum dynamical semigroups},\ }\href@noop {} {\bibfield  {journal}
  {\bibinfo  {journal} {Communications in Mathematical Physics}\ }\textbf
  {\bibinfo {volume} {48}},\ \bibinfo {pages} {119} (\bibinfo {year}
  {1976})}\BibitemShut {NoStop}%
\bibitem [{\citenamefont {Kastoryano}\ and\ \citenamefont
  {Eisert}(2013)}]{KastoryanoEisert}%
  \BibitemOpen
  \bibfield  {author} {\bibinfo {author} {\bibfnamefont {M.~J.}\ \bibnamefont
  {Kastoryano}}\ and\ \bibinfo {author} {\bibfnamefont {J.}~\bibnamefont
  {Eisert}},\ }\bibfield  {title} {\bibinfo {title} {Rapid mixing implies
  exponential decay of correlations},\ }\href
  {https://doi.org/10.1063/1.4822481} {\bibfield  {journal} {\bibinfo
  {journal} {Journal of Mathematical Physics}\ }\textbf {\bibinfo {volume}
  {54}},\ \bibinfo {pages} {102201} (\bibinfo {year} {2013})},\ \Eprint
  {https://arxiv.org/abs/https://doi.org/10.1063/1.4822481}
  {https://doi.org/10.1063/1.4822481} \BibitemShut {NoStop}%
\bibitem [{\citenamefont {Kliesch}\ \emph
  {et~al.}(2014{\natexlab{a}})\citenamefont {Kliesch}, \citenamefont
  {Gogolin},\ and\ \citenamefont {Eisert}}]{Kliesch_2014}%
  \BibitemOpen
  \bibfield  {author} {\bibinfo {author} {\bibfnamefont {M.}~\bibnamefont
  {Kliesch}}, \bibinfo {author} {\bibfnamefont {C.}~\bibnamefont {Gogolin}},\
  and\ \bibinfo {author} {\bibfnamefont {J.}~\bibnamefont {Eisert}},\
  }\bibfield  {title} {\bibinfo {title} {Lieb-robinson bounds and the
  simulation of time-evolution of local observables in lattice systems},\ }in\
  \href {https://doi.org/10.1007/978-3-319-06379-9_17} {\emph {\bibinfo
  {booktitle} {Many-Electron Approaches in Physics, Chemistry and
  Mathematics}}}\ (\bibinfo  {publisher} {Springer International Publishing},\
  \bibinfo {year} {2014})\ pp.\ \bibinfo {pages} {301--318}\BibitemShut
  {NoStop}%
\bibitem [{\citenamefont {Nachtergaele}\ \emph {et~al.}(2011)\citenamefont
  {Nachtergaele}, \citenamefont {Vershynina},\ and\ \citenamefont
  {Zagrebnov}}]{NachtergaeleOpen}%
  \BibitemOpen
  \bibfield  {author} {\bibinfo {author} {\bibfnamefont {B.}~\bibnamefont
  {Nachtergaele}}, \bibinfo {author} {\bibfnamefont {A.}~\bibnamefont
  {Vershynina}},\ and\ \bibinfo {author} {\bibfnamefont {V.~A.}\ \bibnamefont
  {Zagrebnov}},\ }\href {https://doi.org/10.48550/ARXIV.1103.1122} {\bibinfo
  {title} {Lieb-robinson bounds and existence of the thermodynamic limit for a
  class of irreversible quantum dynamics}} (\bibinfo {year} {2011})\BibitemShut
  {NoStop}%
\bibitem [{\citenamefont {Lieb}\ and\ \citenamefont
  {Robinson}(1972)}]{LiebRobinson}%
  \BibitemOpen
  \bibfield  {author} {\bibinfo {author} {\bibfnamefont {E.}~\bibnamefont
  {Lieb}}\ and\ \bibinfo {author} {\bibfnamefont {D.}~\bibnamefont
  {Robinson}},\ }\bibfield  {title} {\bibinfo {title} {The finite group
  velocity of quantum spin systems},\ }\href
  {https://doi.org/10.1007/BF01645779} {\bibfield  {journal} {\bibinfo
  {journal} {Communications in Mathematical Physics}\ }\textbf {\bibinfo
  {volume} {28}} (\bibinfo {year} {1972})}\BibitemShut {NoStop}%
\bibitem [{\citenamefont {Poulin}(2010)}]{Poulin}%
  \BibitemOpen
  \bibfield  {author} {\bibinfo {author} {\bibfnamefont {D.}~\bibnamefont
  {Poulin}},\ }\bibfield  {title} {\bibinfo {title} {Lieb-robinson bound and
  locality for general markovian quantum dynamics},\ }\href
  {https://doi.org/10.1103/PhysRevLett.104.190401} {\bibfield  {journal}
  {\bibinfo  {journal} {Phys. Rev. Lett.}\ }\textbf {\bibinfo {volume} {104}},\
  \bibinfo {pages} {190401} (\bibinfo {year} {2010})}\BibitemShut {NoStop}%
\bibitem [{\citenamefont {Barthel}\ and\ \citenamefont
  {Kliesch}(2012)}]{BarthelLR}%
  \BibitemOpen
  \bibfield  {author} {\bibinfo {author} {\bibfnamefont {T.}~\bibnamefont
  {Barthel}}\ and\ \bibinfo {author} {\bibfnamefont {M.}~\bibnamefont
  {Kliesch}},\ }\bibfield  {title} {\bibinfo {title} {Quasilocality and
  efficient simulation of markovian quantum dynamics},\ }\href
  {https://doi.org/10.1103/PhysRevLett.108.230504} {\bibfield  {journal}
  {\bibinfo  {journal} {Phys. Rev. Lett.}\ }\textbf {\bibinfo {volume} {108}},\
  \bibinfo {pages} {230504} (\bibinfo {year} {2012})}\BibitemShut {NoStop}%
\bibitem [{\citenamefont {Bravyi}\ \emph {et~al.}(2006)\citenamefont {Bravyi},
  \citenamefont {Hastings},\ and\ \citenamefont {Verstraete}}]{BravyiLR}%
  \BibitemOpen
  \bibfield  {author} {\bibinfo {author} {\bibfnamefont {S.}~\bibnamefont
  {Bravyi}}, \bibinfo {author} {\bibfnamefont {M.~B.}\ \bibnamefont
  {Hastings}},\ and\ \bibinfo {author} {\bibfnamefont {F.}~\bibnamefont
  {Verstraete}},\ }\bibfield  {title} {\bibinfo {title} {Lieb-robinson bounds
  and the generation of correlations and topological quantum order},\ }\href
  {https://doi.org/10.1103/PhysRevLett.97.050401} {\bibfield  {journal}
  {\bibinfo  {journal} {Phys. Rev. Lett.}\ }\textbf {\bibinfo {volume} {97}},\
  \bibinfo {pages} {050401} (\bibinfo {year} {2006})}\BibitemShut {NoStop}%
\bibitem [{\citenamefont {Tsubota}\ \emph {et~al.}(2013)\citenamefont
  {Tsubota}, \citenamefont {Kobayashi},\ and\ \citenamefont
  {Takeuchi}}]{quantumhydro}%
  \BibitemOpen
  \bibfield  {author} {\bibinfo {author} {\bibfnamefont {M.}~\bibnamefont
  {Tsubota}}, \bibinfo {author} {\bibfnamefont {M.}~\bibnamefont {Kobayashi}},\
  and\ \bibinfo {author} {\bibfnamefont {H.}~\bibnamefont {Takeuchi}},\
  }\bibfield  {title} {\bibinfo {title} {Quantum hydrodynamics},\ }\href@noop
  {} {\bibfield  {journal} {\bibinfo  {journal} {Physics Reports}\ }\textbf
  {\bibinfo {volume} {522}},\ \bibinfo {pages} {191} (\bibinfo {year}
  {2013})}\BibitemShut {NoStop}%
\bibitem [{\citenamefont {Shirai}\ and\ \citenamefont {Mori}(2020)}]{openETH}%
  \BibitemOpen
  \bibfield  {author} {\bibinfo {author} {\bibfnamefont {T.}~\bibnamefont
  {Shirai}}\ and\ \bibinfo {author} {\bibfnamefont {T.}~\bibnamefont {Mori}},\
  }\bibfield  {title} {\bibinfo {title} {Thermalization in open many-body
  systems based on eigenstate thermalization hypothesis},\ }\bibfield
  {journal} {\bibinfo  {journal} {Physical Review E}\ }\textbf {\bibinfo
  {volume} {101}},\ \href {https://doi.org/10.1103/physreve.101.042116}
  {10.1103/physreve.101.042116} (\bibinfo {year} {2020})\BibitemShut {NoStop}%
\bibitem [{\citenamefont {Ashida}\ \emph {et~al.}(2018)\citenamefont {Ashida},
  \citenamefont {Saito},\ and\ \citenamefont {Ueda}}]{openETH2}%
  \BibitemOpen
  \bibfield  {author} {\bibinfo {author} {\bibfnamefont {Y.}~\bibnamefont
  {Ashida}}, \bibinfo {author} {\bibfnamefont {K.}~\bibnamefont {Saito}},\ and\
  \bibinfo {author} {\bibfnamefont {M.}~\bibnamefont {Ueda}},\ }\bibfield
  {title} {\bibinfo {title} {Thermalization and heating dynamics in open
  generic many-body systems},\ }\href
  {https://doi.org/10.1103/PhysRevLett.121.170402} {\bibfield  {journal}
  {\bibinfo  {journal} {Phys. Rev. Lett.}\ }\textbf {\bibinfo {volume} {121}},\
  \bibinfo {pages} {170402} (\bibinfo {year} {2018})}\BibitemShut {NoStop}%
\bibitem [{\citenamefont {Fux}\ \emph {et~al.}(2022)\citenamefont {Fux},
  \citenamefont {Kilda}, \citenamefont {Lovett},\ and\ \citenamefont
  {Keeling}}]{Keeling}%
  \BibitemOpen
  \bibfield  {author} {\bibinfo {author} {\bibfnamefont {G.~E.}\ \bibnamefont
  {Fux}}, \bibinfo {author} {\bibfnamefont {D.}~\bibnamefont {Kilda}}, \bibinfo
  {author} {\bibfnamefont {B.~W.}\ \bibnamefont {Lovett}},\ and\ \bibinfo
  {author} {\bibfnamefont {J.}~\bibnamefont {Keeling}},\ }\href
  {https://doi.org/10.48550/ARXIV.2201.05529} {\bibinfo {title} {Thermalization
  of a spin chain strongly coupled to its environment}} (\bibinfo {year}
  {2022})\BibitemShut {NoStop}%
\bibitem [{\citenamefont {Diehl}\ \emph {et~al.}(2008)\citenamefont {Diehl},
  \citenamefont {Micheli}, \citenamefont {Kantian}, \citenamefont {Kraus},
  \citenamefont {Büchler},\ and\ \citenamefont {Zoller}}]{Diehl_2008}%
  \BibitemOpen
  \bibfield  {author} {\bibinfo {author} {\bibfnamefont {S.}~\bibnamefont
  {Diehl}}, \bibinfo {author} {\bibfnamefont {A.}~\bibnamefont {Micheli}},
  \bibinfo {author} {\bibfnamefont {A.}~\bibnamefont {Kantian}}, \bibinfo
  {author} {\bibfnamefont {B.}~\bibnamefont {Kraus}}, \bibinfo {author}
  {\bibfnamefont {H.~P.}\ \bibnamefont {Büchler}},\ and\ \bibinfo {author}
  {\bibfnamefont {P.}~\bibnamefont {Zoller}},\ }\bibfield  {title} {\bibinfo
  {title} {Quantum states and phases in driven open quantum systems with cold
  atoms},\ }\href {https://doi.org/10.1038/nphys1073} {\bibfield  {journal}
  {\bibinfo  {journal} {Nature Physics}\ }\textbf {\bibinfo {volume} {4}},\
  \bibinfo {pages} {878} (\bibinfo {year} {2008})}\BibitemShut {NoStop}%
\bibitem [{\citenamefont {\v{Z}nidari\v{c}}(2015)}]{ZnidaricRelaxation}%
  \BibitemOpen
  \bibfield  {author} {\bibinfo {author} {\bibfnamefont {M.}~\bibnamefont
  {\v{Z}nidari\v{c}}},\ }\bibfield  {title} {\bibinfo {title} {Relaxation times
  of dissipative many-body quantum systems},\ }\bibfield  {journal} {\bibinfo
  {journal} {Physical Review E}\ }\textbf {\bibinfo {volume} {92}},\ \href
  {https://doi.org/10.1103/physreve.92.042143} {10.1103/physreve.92.042143}
  (\bibinfo {year} {2015})\BibitemShut {NoStop}%
\bibitem [{\citenamefont {Yoshida}\ and\ \citenamefont
  {Katsura}(2022)}]{HoshoGap}%
  \BibitemOpen
  \bibfield  {author} {\bibinfo {author} {\bibfnamefont {H.}~\bibnamefont
  {Yoshida}}\ and\ \bibinfo {author} {\bibfnamefont {H.}~\bibnamefont
  {Katsura}},\ }\href {https://doi.org/10.48550/ARXIV.2209.03743} {\bibinfo
  {title} {Exact analysis of the liouvillian gap and dynamics in the
  dissipative su($n$) fermi-hubbard model}} (\bibinfo {year}
  {2022})\BibitemShut {NoStop}%
\bibitem [{\citenamefont {Ilievski}\ and\ \citenamefont
  {Prosen}(2012)}]{Ilievski_2012}%
  \BibitemOpen
  \bibfield  {author} {\bibinfo {author} {\bibfnamefont {E.}~\bibnamefont
  {Ilievski}}\ and\ \bibinfo {author} {\bibfnamefont {T.}~\bibnamefont
  {Prosen}},\ }\bibfield  {title} {\bibinfo {title} {Thermodyamic bounds on
  drude weights in terms of almost-conserved quantities},\ }\href
  {https://doi.org/10.1007/s00220-012-1599-4} {\bibfield  {journal} {\bibinfo
  {journal} {Communications in Mathematical Physics}\ }\textbf {\bibinfo
  {volume} {318}},\ \bibinfo {pages} {809} (\bibinfo {year}
  {2012})}\BibitemShut {NoStop}%
\bibitem [{\citenamefont {Singh}\ \emph {et~al.}(2023)\citenamefont {Singh},
  \citenamefont {Vasseur},\ and\ \citenamefont
  {Gopalakrishnan}}]{FiniteFreqDrude}%
  \BibitemOpen
  \bibfield  {author} {\bibinfo {author} {\bibfnamefont {H.}~\bibnamefont
  {Singh}}, \bibinfo {author} {\bibfnamefont {R.}~\bibnamefont {Vasseur}},\
  and\ \bibinfo {author} {\bibfnamefont {S.}~\bibnamefont {Gopalakrishnan}},\
  }\bibfield  {title} {\bibinfo {title} {Fredkin staircase: An integrable
  system with a finite-frequency drude peak},\ }\href
  {https://doi.org/10.1103/PhysRevLett.130.046001} {\bibfield  {journal}
  {\bibinfo  {journal} {Phys. Rev. Lett.}\ }\textbf {\bibinfo {volume} {130}},\
  \bibinfo {pages} {046001} (\bibinfo {year} {2023})}\BibitemShut {NoStop}%
\bibitem [{\citenamefont {Medenjak}\ \emph
  {et~al.}(2020{\natexlab{b}})\citenamefont {Medenjak}, \citenamefont
  {Prosen},\ and\ \citenamefont {Zadnik}}]{Marko2}%
  \BibitemOpen
  \bibfield  {author} {\bibinfo {author} {\bibfnamefont {M.}~\bibnamefont
  {Medenjak}}, \bibinfo {author} {\bibfnamefont {T.}~\bibnamefont {Prosen}},\
  and\ \bibinfo {author} {\bibfnamefont {L.}~\bibnamefont {Zadnik}},\
  }\bibfield  {title} {\bibinfo {title} {Rigorous bounds on dynamical response
  functions and time-translation symmetry breaking},\ }\href
  {http://dx.doi.org/10.21468/SciPostPhys.9.1.003} {\bibfield  {journal}
  {\bibinfo  {journal} {SciPost Physics}\ }\textbf {\bibinfo {volume} {9}}
  (\bibinfo {year} {2020}{\natexlab{b}})}\BibitemShut {NoStop}%
\bibitem [{\citenamefont {Karrasch}\ \emph {et~al.}(2017)\citenamefont
  {Karrasch}, \citenamefont {Prosen},\ and\ \citenamefont
  {Heidrich-Meisner}}]{karrasch2017proposal}%
  \BibitemOpen
  \bibfield  {author} {\bibinfo {author} {\bibfnamefont {C.}~\bibnamefont
  {Karrasch}}, \bibinfo {author} {\bibfnamefont {T.}~\bibnamefont {Prosen}},\
  and\ \bibinfo {author} {\bibfnamefont {F.}~\bibnamefont {Heidrich-Meisner}},\
  }\bibfield  {title} {\bibinfo {title} {Proposal for measuring the
  finite-temperature drude weight of integrable systems},\ }\href@noop {}
  {\bibfield  {journal} {\bibinfo  {journal} {Physical Review B}\ }\textbf
  {\bibinfo {volume} {95}},\ \bibinfo {pages} {060406} (\bibinfo {year}
  {2017})}\BibitemShut {NoStop}%
\bibitem [{\citenamefont {Suzuki}(1971)}]{SUZUKI1971277}%
  \BibitemOpen
  \bibfield  {author} {\bibinfo {author} {\bibfnamefont {M.}~\bibnamefont
  {Suzuki}},\ }\bibfield  {title} {\bibinfo {title} {Ergodicity, constants of
  motion, and bounds for susceptibilities},\ }\href
  {https://doi.org/https://doi.org/10.1016/0031-8914(71)90226-6} {\bibfield
  {journal} {\bibinfo  {journal} {Physica}\ }\textbf {\bibinfo {volume} {51}},\
  \bibinfo {pages} {277} (\bibinfo {year} {1971})}\BibitemShut {NoStop}%
\bibitem [{\citenamefont {Dhar}\ \emph {et~al.}(2021)\citenamefont {Dhar},
  \citenamefont {Kundu},\ and\ \citenamefont {Saito}}]{Dhar_2021}%
  \BibitemOpen
  \bibfield  {author} {\bibinfo {author} {\bibfnamefont {A.}~\bibnamefont
  {Dhar}}, \bibinfo {author} {\bibfnamefont {A.}~\bibnamefont {Kundu}},\ and\
  \bibinfo {author} {\bibfnamefont {K.}~\bibnamefont {Saito}},\ }\bibfield
  {title} {\bibinfo {title} {Revisiting the mazur bound and the suzuki
  equality},\ }\href {https://doi.org/10.1016/j.chaos.2020.110618} {\bibfield
  {journal} {\bibinfo  {journal} {Chaos, Solitons and Fractals}\ }\textbf
  {\bibinfo {volume} {144}},\ \bibinfo {pages} {110618} (\bibinfo {year}
  {2021})}\BibitemShut {NoStop}%
\bibitem [{\citenamefont {Sirker}(2020)}]{sirker2020transport}%
  \BibitemOpen
  \bibfield  {author} {\bibinfo {author} {\bibfnamefont {J.}~\bibnamefont
  {Sirker}},\ }\bibfield  {title} {\bibinfo {title} {Transport in
  one-dimensional integrable quantum systems},\ }\href@noop {} {\bibfield
  {journal} {\bibinfo  {journal} {SciPost Physics Lecture Notes}\ ,\ \bibinfo
  {pages} {017}} (\bibinfo {year} {2020})}\BibitemShut {NoStop}%
\bibitem [{\citenamefont {Doyon}(2022)}]{Doyon_2022}%
  \BibitemOpen
  \bibfield  {author} {\bibinfo {author} {\bibfnamefont {B.}~\bibnamefont
  {Doyon}},\ }\bibfield  {title} {\bibinfo {title} {Hydrodynamic projections
  and the emergence of linearised euler equations in one-dimensional isolated
  systems},\ }\href {https://doi.org/10.1007/s00220-022-04310-3} {\bibfield
  {journal} {\bibinfo  {journal} {Communications in Mathematical Physics}\
  }\textbf {\bibinfo {volume} {391}},\ \bibinfo {pages} {293} (\bibinfo {year}
  {2022})}\BibitemShut {NoStop}%
\bibitem [{\citenamefont {Ampelogiannis}\ and\ \citenamefont
  {Doyon}(2021{\natexlab{a}})}]{Doyon3}%
  \BibitemOpen
  \bibfield  {author} {\bibinfo {author} {\bibfnamefont {D.}~\bibnamefont
  {Ampelogiannis}}\ and\ \bibinfo {author} {\bibfnamefont {B.}~\bibnamefont
  {Doyon}},\ }\href {https://doi.org/10.48550/ARXIV.2112.12747} {\bibinfo
  {title} {Long-time dynamics in quantum spin lattices: ergodicity and
  hydrodynamic projections at all frequencies and wavelengths}} (\bibinfo
  {year} {2021}{\natexlab{a}})\BibitemShut {NoStop}%
\bibitem [{\citenamefont {Serbyn}\ \emph {et~al.}(2013)\citenamefont {Serbyn},
  \citenamefont {Papi\ifmmode~\acute{c}\else \'{c}\fi{}},\ and\ \citenamefont
  {Abanin}}]{Serbyn}%
  \BibitemOpen
  \bibfield  {author} {\bibinfo {author} {\bibfnamefont {M.}~\bibnamefont
  {Serbyn}}, \bibinfo {author} {\bibfnamefont {Z.}~\bibnamefont
  {Papi\ifmmode~\acute{c}\else \'{c}\fi{}}},\ and\ \bibinfo {author}
  {\bibfnamefont {D.~A.}\ \bibnamefont {Abanin}},\ }\bibfield  {title}
  {\bibinfo {title} {Local conservation laws and the structure of the many-body
  localized states},\ }\href {https://doi.org/10.1103/PhysRevLett.111.127201}
  {\bibfield  {journal} {\bibinfo  {journal} {Phys. Rev. Lett.}\ }\textbf
  {\bibinfo {volume} {111}},\ \bibinfo {pages} {127201} (\bibinfo {year}
  {2013})}\BibitemShut {NoStop}%
\bibitem [{\citenamefont {Huse}\ \emph {et~al.}(2014)\citenamefont {Huse},
  \citenamefont {Nandkishore},\ and\ \citenamefont {Oganesyan}}]{Huse}%
  \BibitemOpen
  \bibfield  {author} {\bibinfo {author} {\bibfnamefont {D.~A.}\ \bibnamefont
  {Huse}}, \bibinfo {author} {\bibfnamefont {R.}~\bibnamefont {Nandkishore}},\
  and\ \bibinfo {author} {\bibfnamefont {V.}~\bibnamefont {Oganesyan}},\
  }\bibfield  {title} {\bibinfo {title} {Phenomenology of fully
  many-body-localized systems},\ }\href
  {https://doi.org/10.1103/PhysRevB.90.174202} {\bibfield  {journal} {\bibinfo
  {journal} {Phys. Rev. B}\ }\textbf {\bibinfo {volume} {90}},\ \bibinfo
  {pages} {174202} (\bibinfo {year} {2014})}\BibitemShut {NoStop}%
\bibitem [{\citenamefont {Iadecola}\ \emph {et~al.}(2019)\citenamefont
  {Iadecola}, \citenamefont {Schecter},\ and\ \citenamefont {Xu}}]{Tom1}%
  \BibitemOpen
  \bibfield  {author} {\bibinfo {author} {\bibfnamefont {T.}~\bibnamefont
  {Iadecola}}, \bibinfo {author} {\bibfnamefont {M.}~\bibnamefont {Schecter}},\
  and\ \bibinfo {author} {\bibfnamefont {S.}~\bibnamefont {Xu}},\ }\bibfield
  {title} {\bibinfo {title} {Quantum many-body scars from magnon
  condensation},\ }\href {https://doi.org/10.1103/PhysRevB.100.184312}
  {\bibfield  {journal} {\bibinfo  {journal} {Phys. Rev. B}\ }\textbf {\bibinfo
  {volume} {100}},\ \bibinfo {pages} {184312} (\bibinfo {year}
  {2019})}\BibitemShut {NoStop}%
\bibitem [{\citenamefont {Schecter}\ and\ \citenamefont
  {Iadecola}(2019)}]{Tom2}%
  \BibitemOpen
  \bibfield  {author} {\bibinfo {author} {\bibfnamefont {M.}~\bibnamefont
  {Schecter}}\ and\ \bibinfo {author} {\bibfnamefont {T.}~\bibnamefont
  {Iadecola}},\ }\bibfield  {title} {\bibinfo {title} {Weak ergodicity breaking
  and quantum many-body scars in spin-1 $xy$ magnets},\ }\href
  {https://doi.org/10.1103/PhysRevLett.123.147201} {\bibfield  {journal}
  {\bibinfo  {journal} {Phys. Rev. Lett.}\ }\textbf {\bibinfo {volume} {123}},\
  \bibinfo {pages} {147201} (\bibinfo {year} {2019})}\BibitemShut {NoStop}%
\bibitem [{\citenamefont {Shiraishi}\ and\ \citenamefont {Mori}(2017)}]{embed}%
  \BibitemOpen
  \bibfield  {author} {\bibinfo {author} {\bibfnamefont {N.}~\bibnamefont
  {Shiraishi}}\ and\ \bibinfo {author} {\bibfnamefont {T.}~\bibnamefont
  {Mori}},\ }\bibfield  {title} {\bibinfo {title} {Systematic construction of
  counterexamples to the eigenstate thermalization hypothesis},\ }\href
  {https://doi.org/10.1103/PhysRevLett.119.030601} {\bibfield  {journal}
  {\bibinfo  {journal} {Phys. Rev. Lett.}\ }\textbf {\bibinfo {volume} {119}},\
  \bibinfo {pages} {030601} (\bibinfo {year} {2017})}\BibitemShut {NoStop}%
\bibitem [{\citenamefont {Chandran}\ \emph {et~al.}(2022)\citenamefont
  {Chandran}, \citenamefont {Iadecola}, \citenamefont {Khemani},\ and\
  \citenamefont {Moessner}}]{chandran2022quantum}%
  \BibitemOpen
  \bibfield  {author} {\bibinfo {author} {\bibfnamefont {A.}~\bibnamefont
  {Chandran}}, \bibinfo {author} {\bibfnamefont {T.}~\bibnamefont {Iadecola}},
  \bibinfo {author} {\bibfnamefont {V.}~\bibnamefont {Khemani}},\ and\ \bibinfo
  {author} {\bibfnamefont {R.}~\bibnamefont {Moessner}},\ }\bibfield  {title}
  {\bibinfo {title} {Quantum many-body scars: A quasiparticle perspective},\
  }\href@noop {} {\bibfield  {journal} {\bibinfo  {journal} {Annual Review of
  Condensed Matter Physics}\ }\textbf {\bibinfo {volume} {14}} (\bibinfo {year}
  {2022})}\BibitemShut {NoStop}%
\bibitem [{\citenamefont {Daniel}\ \emph {et~al.}(2023)\citenamefont {Daniel},
  \citenamefont {Hallam}, \citenamefont {Desaules}, \citenamefont {Hudomal},
  \citenamefont {{Guo-Xian}}, \citenamefont {Halimeh},\ and\ \citenamefont
  {Papić}}]{BridgingScars}%
  \BibitemOpen
  \bibfield  {author} {\bibinfo {author} {\bibfnamefont {A.}~\bibnamefont
  {Daniel}}, \bibinfo {author} {\bibfnamefont {A.}~\bibnamefont {Hallam}},
  \bibinfo {author} {\bibfnamefont {J.-Y.}\ \bibnamefont {Desaules}}, \bibinfo
  {author} {\bibfnamefont {A.}~\bibnamefont {Hudomal}}, \bibinfo {author}
  {\bibnamefont {{Guo-Xian}}}, \bibinfo {author} {\bibfnamefont {J.~C.}\
  \bibnamefont {Halimeh}},\ and\ \bibinfo {author} {\bibfnamefont
  {Z.}~\bibnamefont {Papić}},\ }\href
  {https://doi.org/10.48550/ARXIV.2301.03631} {\bibinfo {title} {Bridging
  quantum criticality via many-body scarring}} (\bibinfo {year}
  {2023})\BibitemShut {NoStop}%
\bibitem [{\citenamefont {Ren}\ \emph {et~al.}(2021)\citenamefont {Ren},
  \citenamefont {Liang},\ and\ \citenamefont {Fang}}]{quasisymmetry}%
  \BibitemOpen
  \bibfield  {author} {\bibinfo {author} {\bibfnamefont {J.}~\bibnamefont
  {Ren}}, \bibinfo {author} {\bibfnamefont {C.}~\bibnamefont {Liang}},\ and\
  \bibinfo {author} {\bibfnamefont {C.}~\bibnamefont {Fang}},\ }\bibfield
  {title} {\bibinfo {title} {Quasisymmetry groups and many-body scar
  dynamics},\ }\href {https://doi.org/10.1103/PhysRevLett.126.120604}
  {\bibfield  {journal} {\bibinfo  {journal} {Phys. Rev. Lett.}\ }\textbf
  {\bibinfo {volume} {126}},\ \bibinfo {pages} {120604} (\bibinfo {year}
  {2021})}\BibitemShut {NoStop}%
\bibitem [{\citenamefont {Pakrouski}\ \emph {et~al.}(2021)\citenamefont
  {Pakrouski}, \citenamefont {Pallegar}, \citenamefont {Popov},\ and\
  \citenamefont {Klebanov}}]{scarsPnew}%
  \BibitemOpen
  \bibfield  {author} {\bibinfo {author} {\bibfnamefont {K.}~\bibnamefont
  {Pakrouski}}, \bibinfo {author} {\bibfnamefont {P.~N.}\ \bibnamefont
  {Pallegar}}, \bibinfo {author} {\bibfnamefont {F.~K.}\ \bibnamefont
  {Popov}},\ and\ \bibinfo {author} {\bibfnamefont {I.~R.}\ \bibnamefont
  {Klebanov}},\ }\bibfield  {title} {\bibinfo {title} {Group theoretic approach
  to many-body scar states in fermionic lattice models},\ }\href
  {https://doi.org/10.1103/PhysRevResearch.3.043156} {\bibfield  {journal}
  {\bibinfo  {journal} {Phys. Rev. Res.}\ }\textbf {\bibinfo {volume} {3}},\
  \bibinfo {pages} {043156} (\bibinfo {year} {2021})}\BibitemShut {NoStop}%
\bibitem [{\citenamefont {Olmos}\ \emph {et~al.}(2010)\citenamefont {Olmos},
  \citenamefont {Müller},\ and\ \citenamefont {Lesanovsky}}]{Olmos_2010}%
  \BibitemOpen
  \bibfield  {author} {\bibinfo {author} {\bibfnamefont {B.}~\bibnamefont
  {Olmos}}, \bibinfo {author} {\bibfnamefont {M.}~\bibnamefont {Müller}},\
  and\ \bibinfo {author} {\bibfnamefont {I.}~\bibnamefont {Lesanovsky}},\
  }\bibfield  {title} {\bibinfo {title} {Thermalization of a strongly
  interacting 1d rydberg lattice gas},\ }\href
  {https://doi.org/10.1088/1367-2630/12/1/013024} {\bibfield  {journal}
  {\bibinfo  {journal} {New Journal of Physics}\ }\textbf {\bibinfo {volume}
  {12}},\ \bibinfo {pages} {013024} (\bibinfo {year} {2010})}\BibitemShut
  {NoStop}%
\bibitem [{\citenamefont {Bull}\ \emph {et~al.}(2020)\citenamefont {Bull},
  \citenamefont {Desaules},\ and\ \citenamefont {Papi\ifmmode~\acute{c}\else
  \'{c}\fi{}}}]{scarsdynsym1}%
  \BibitemOpen
  \bibfield  {author} {\bibinfo {author} {\bibfnamefont {K.}~\bibnamefont
  {Bull}}, \bibinfo {author} {\bibfnamefont {J.-Y.}\ \bibnamefont {Desaules}},\
  and\ \bibinfo {author} {\bibfnamefont {Z.}~\bibnamefont
  {Papi\ifmmode~\acute{c}\else \'{c}\fi{}}},\ }\bibfield  {title} {\bibinfo
  {title} {Quantum scars as embeddings of weakly broken lie algebra
  representations},\ }\href {https://doi.org/10.1103/PhysRevB.101.165139}
  {\bibfield  {journal} {\bibinfo  {journal} {Phys. Rev. B}\ }\textbf {\bibinfo
  {volume} {101}},\ \bibinfo {pages} {165139} (\bibinfo {year}
  {2020})}\BibitemShut {NoStop}%
\bibitem [{\citenamefont {Iadecola}\ and\ \citenamefont {\ifmmode
  \check{Z}\else \v{Z}\fi{}nidari\ifmmode~\check{c}\else
  \v{c}\fi{}}(2019)}]{ZnidaricFrag}%
  \BibitemOpen
  \bibfield  {author} {\bibinfo {author} {\bibfnamefont {T.}~\bibnamefont
  {Iadecola}}\ and\ \bibinfo {author} {\bibfnamefont {M.}~\bibnamefont
  {\ifmmode \check{Z}\else \v{Z}\fi{}nidari\ifmmode~\check{c}\else
  \v{c}\fi{}}},\ }\bibfield  {title} {\bibinfo {title} {Exact localized and
  ballistic eigenstates in disordered chaotic spin ladders and the
  fermi-hubbard model},\ }\href
  {https://doi.org/10.1103/PhysRevLett.123.036403} {\bibfield  {journal}
  {\bibinfo  {journal} {Phys. Rev. Lett.}\ }\textbf {\bibinfo {volume} {123}},\
  \bibinfo {pages} {036403} (\bibinfo {year} {2019})}\BibitemShut {NoStop}%
\bibitem [{\citenamefont {Chattopadhyay}\ \emph {et~al.}(2022)\citenamefont
  {Chattopadhyay}, \citenamefont {Mukherjee}, \citenamefont {Sengupta},\ and\
  \citenamefont {Sen}}]{ArnabFrag2}%
  \BibitemOpen
  \bibfield  {author} {\bibinfo {author} {\bibfnamefont {A.}~\bibnamefont
  {Chattopadhyay}}, \bibinfo {author} {\bibfnamefont {B.}~\bibnamefont
  {Mukherjee}}, \bibinfo {author} {\bibfnamefont {K.}~\bibnamefont
  {Sengupta}},\ and\ \bibinfo {author} {\bibfnamefont {A.}~\bibnamefont
  {Sen}},\ }\href {https://doi.org/10.48550/ARXIV.2208.13800} {\bibinfo {title}
  {Strong hilbert space fragmentation via emergent quantum drums in two
  dimensions}} (\bibinfo {year} {2022})\BibitemShut {NoStop}%
\bibitem [{\citenamefont {Lehmann}\ \emph {et~al.}(2022)\citenamefont
  {Lehmann}, \citenamefont {Sala}, \citenamefont {Pollmann},\ and\
  \citenamefont {Rakovszky}}]{PabloNew}%
  \BibitemOpen
  \bibfield  {author} {\bibinfo {author} {\bibfnamefont {J.}~\bibnamefont
  {Lehmann}}, \bibinfo {author} {\bibfnamefont {P.}~\bibnamefont {Sala}},
  \bibinfo {author} {\bibfnamefont {F.}~\bibnamefont {Pollmann}},\ and\
  \bibinfo {author} {\bibfnamefont {T.}~\bibnamefont {Rakovszky}},\ }\href
  {https://doi.org/10.48550/ARXIV.2208.12260} {\bibinfo {title}
  {Fragmentation-induced localization and boundary charges in dimensions two
  and above}} (\bibinfo {year} {2022})\BibitemShut {NoStop}%
\bibitem [{\citenamefont {Halimeh}\ \emph
  {et~al.}(2022{\natexlab{c}})\citenamefont {Halimeh}, \citenamefont {Lang},\
  and\ \citenamefont {Hauke}}]{JadDis2}%
  \BibitemOpen
  \bibfield  {author} {\bibinfo {author} {\bibfnamefont {J.~C.}\ \bibnamefont
  {Halimeh}}, \bibinfo {author} {\bibfnamefont {H.}~\bibnamefont {Lang}},\ and\
  \bibinfo {author} {\bibfnamefont {P.}~\bibnamefont {Hauke}},\ }\bibfield
  {title} {\bibinfo {title} {Gauge protection in non-abelian lattice gauge
  theories},\ }\href {https://doi.org/10.1088/1367-2630/ac5564} {\bibfield
  {journal} {\bibinfo  {journal} {New Journal of Physics}\ }\textbf {\bibinfo
  {volume} {24}},\ \bibinfo {pages} {033015} (\bibinfo {year}
  {2022}{\natexlab{c}})}\BibitemShut {NoStop}%
\bibitem [{\citenamefont {Nicolau}\ \emph {et~al.}(2022)\citenamefont
  {Nicolau}, \citenamefont {Marques}, \citenamefont {Mompart}, \citenamefont
  {Ahufinger},\ and\ \citenamefont {Dias}}]{LOCfrag1}%
  \BibitemOpen
  \bibfield  {author} {\bibinfo {author} {\bibfnamefont {E.}~\bibnamefont
  {Nicolau}}, \bibinfo {author} {\bibfnamefont {A.~M.}\ \bibnamefont
  {Marques}}, \bibinfo {author} {\bibfnamefont {J.}~\bibnamefont {Mompart}},
  \bibinfo {author} {\bibfnamefont {V.}~\bibnamefont {Ahufinger}},\ and\
  \bibinfo {author} {\bibfnamefont {R.~G.}\ \bibnamefont {Dias}},\ }\href
  {https://doi.org/10.48550/ARXIV.2210.02429} {\bibinfo {title} {Local hilbert
  space fragmentation and weak thermalization in bose-hubbard diamond
  necklaces}} (\bibinfo {year} {2022})\BibitemShut {NoStop}%
\bibitem [{\citenamefont {Danieli}\ \emph {et~al.}(2020)\citenamefont
  {Danieli}, \citenamefont {Andreanov},\ and\ \citenamefont
  {Flach}}]{LOCfrag4}%
  \BibitemOpen
  \bibfield  {author} {\bibinfo {author} {\bibfnamefont {C.}~\bibnamefont
  {Danieli}}, \bibinfo {author} {\bibfnamefont {A.}~\bibnamefont {Andreanov}},\
  and\ \bibinfo {author} {\bibfnamefont {S.}~\bibnamefont {Flach}},\ }\bibfield
   {title} {\bibinfo {title} {Many-body flatband localization},\ }\href
  {https://doi.org/10.1103/PhysRevB.102.041116} {\bibfield  {journal} {\bibinfo
   {journal} {Phys. Rev. B}\ }\textbf {\bibinfo {volume} {102}},\ \bibinfo
  {pages} {041116} (\bibinfo {year} {2020})}\BibitemShut {NoStop}%
\bibitem [{\citenamefont {Hahn}\ \emph {et~al.}(2021)\citenamefont {Hahn},
  \citenamefont {McClarty},\ and\ \citenamefont {Luitz}}]{LOCfrag5}%
  \BibitemOpen
  \bibfield  {author} {\bibinfo {author} {\bibfnamefont {D.}~\bibnamefont
  {Hahn}}, \bibinfo {author} {\bibfnamefont {P.~A.}\ \bibnamefont {McClarty}},\
  and\ \bibinfo {author} {\bibfnamefont {D.~J.}\ \bibnamefont {Luitz}},\
  }\bibfield  {title} {\bibinfo {title} {{Information dynamics in a model with
  Hilbert space fragmentation}},\ }\href
  {https://doi.org/10.21468/SciPostPhys.11.4.074} {\bibfield  {journal}
  {\bibinfo  {journal} {SciPost Phys.}\ }\textbf {\bibinfo {volume} {11}},\
  \bibinfo {pages} {074} (\bibinfo {year} {2021})}\BibitemShut {NoStop}%
\bibitem [{\citenamefont {Chertkov}\ and\ \citenamefont
  {Clark}(2021)}]{LOCfrag6}%
  \BibitemOpen
  \bibfield  {author} {\bibinfo {author} {\bibfnamefont {E.}~\bibnamefont
  {Chertkov}}\ and\ \bibinfo {author} {\bibfnamefont {B.~K.}\ \bibnamefont
  {Clark}},\ }\bibfield  {title} {\bibinfo {title} {Motif magnetism and quantum
  many-body scars},\ }\href {https://doi.org/10.1103/PhysRevB.104.104410}
  {\bibfield  {journal} {\bibinfo  {journal} {Phys. Rev. B}\ }\textbf {\bibinfo
  {volume} {104}},\ \bibinfo {pages} {104410} (\bibinfo {year}
  {2021})}\BibitemShut {NoStop}%
\bibitem [{\citenamefont {Schulz}\ \emph {et~al.}(2019)\citenamefont {Schulz},
  \citenamefont {Hooley}, \citenamefont {Moessner},\ and\ \citenamefont
  {Pollmann}}]{Stark1}%
  \BibitemOpen
  \bibfield  {author} {\bibinfo {author} {\bibfnamefont {M.}~\bibnamefont
  {Schulz}}, \bibinfo {author} {\bibfnamefont {C.~A.}\ \bibnamefont {Hooley}},
  \bibinfo {author} {\bibfnamefont {R.}~\bibnamefont {Moessner}},\ and\
  \bibinfo {author} {\bibfnamefont {F.}~\bibnamefont {Pollmann}},\ }\bibfield
  {title} {\bibinfo {title} {Stark many-body localization},\ }\href
  {https://doi.org/10.1103/PhysRevLett.122.040606} {\bibfield  {journal}
  {\bibinfo  {journal} {Phys. Rev. Lett.}\ }\textbf {\bibinfo {volume} {122}},\
  \bibinfo {pages} {040606} (\bibinfo {year} {2019})}\BibitemShut {NoStop}%
\bibitem [{\citenamefont {Gunawardana}\ and\ \citenamefont
  {Bu\ifmmode~\check{c}\else \v{c}\fi{}a}(2022)}]{Thivan}%
  \BibitemOpen
  \bibfield  {author} {\bibinfo {author} {\bibfnamefont {T.~M.}\ \bibnamefont
  {Gunawardana}}\ and\ \bibinfo {author} {\bibfnamefont {B.}~\bibnamefont
  {Bu\ifmmode~\check{c}\else \v{c}\fi{}a}},\ }\bibfield  {title} {\bibinfo
  {title} {Dynamical l-bits and persistent oscillations in stark many-body
  localization},\ }\href {https://doi.org/10.1103/PhysRevB.106.L161111}
  {\bibfield  {journal} {\bibinfo  {journal} {Phys. Rev. B}\ }\textbf {\bibinfo
  {volume} {106}},\ \bibinfo {pages} {L161111} (\bibinfo {year}
  {2022})}\BibitemShut {NoStop}%
\bibitem [{\citenamefont {Wilming}\ \emph {et~al.}(2022)\citenamefont
  {Wilming}, \citenamefont {Osborne}, \citenamefont {Decker},\ and\
  \citenamefont {Karrasch}}]{Henrik}%
  \BibitemOpen
  \bibfield  {author} {\bibinfo {author} {\bibfnamefont {H.}~\bibnamefont
  {Wilming}}, \bibinfo {author} {\bibfnamefont {T.~J.}\ \bibnamefont
  {Osborne}}, \bibinfo {author} {\bibfnamefont {K.~S.~C.}\ \bibnamefont
  {Decker}},\ and\ \bibinfo {author} {\bibfnamefont {C.}~\bibnamefont
  {Karrasch}},\ }\href {https://doi.org/10.48550/ARXIV.2210.03153} {\bibinfo
  {title} {Reviving product states in the disordered heisenberg chain}}
  (\bibinfo {year} {2022})\BibitemShut {NoStop}%
\bibitem [{\citenamefont {Mendoza-Arenas}\ and\ \citenamefont
  {Clark}(2022)}]{Stephen}%
  \BibitemOpen
  \bibfield  {author} {\bibinfo {author} {\bibfnamefont {J.~J.}\ \bibnamefont
  {Mendoza-Arenas}}\ and\ \bibinfo {author} {\bibfnamefont {S.~R.}\
  \bibnamefont {Clark}},\ }\href {https://doi.org/10.48550/ARXIV.2209.11718}
  {\bibinfo {title} {Giant rectification in strongly-interacting
  boundary-driven tilted systems}} (\bibinfo {year} {2022})\BibitemShut
  {NoStop}%
\bibitem [{\citenamefont {Sacha}(2015)}]{Sacha_2015}%
  \BibitemOpen
  \bibfield  {author} {\bibinfo {author} {\bibfnamefont {K.}~\bibnamefont
  {Sacha}},\ }\bibfield  {title} {\bibinfo {title} {Modeling spontaneous
  breaking of time-translation symmetry},\ }\bibfield  {journal} {\bibinfo
  {journal} {Physical Review A}\ }\textbf {\bibinfo {volume} {91}},\ \href
  {https://doi.org/10.1103/physreva.91.033617} {10.1103/physreva.91.033617}
  (\bibinfo {year} {2015})\BibitemShut {NoStop}%
\bibitem [{\citenamefont {Else}\ \emph {et~al.}(2016)\citenamefont {Else},
  \citenamefont {Bauer},\ and\ \citenamefont {Nayak}}]{ElseTC}%
  \BibitemOpen
  \bibfield  {author} {\bibinfo {author} {\bibfnamefont {D.~V.}\ \bibnamefont
  {Else}}, \bibinfo {author} {\bibfnamefont {B.}~\bibnamefont {Bauer}},\ and\
  \bibinfo {author} {\bibfnamefont {C.}~\bibnamefont {Nayak}},\ }\bibfield
  {title} {\bibinfo {title} {Floquet time crystals},\ }\href
  {https://doi.org/10.1103/PhysRevLett.117.090402} {\bibfield  {journal}
  {\bibinfo  {journal} {Phys. Rev. Lett.}\ }\textbf {\bibinfo {volume} {117}},\
  \bibinfo {pages} {090402} (\bibinfo {year} {2016})}\BibitemShut {NoStop}%
\bibitem [{\citenamefont {Yao}\ and\ \citenamefont {Nayak}(2018)}]{cryptoeq}%
  \BibitemOpen
  \bibfield  {author} {\bibinfo {author} {\bibfnamefont {N.~Y.}\ \bibnamefont
  {Yao}}\ and\ \bibinfo {author} {\bibfnamefont {C.}~\bibnamefont {Nayak}},\
  }\bibfield  {title} {\bibinfo {title} {Time crystals in periodically driven
  systems},\ }\href {https://doi.org/10.1063/PT.3.4020} {\bibfield  {journal}
  {\bibinfo  {journal} {Physics Today}\ }\textbf {\bibinfo {volume} {71}},\
  \bibinfo {pages} {40} (\bibinfo {year} {2018})},\ \Eprint
  {https://arxiv.org/abs/https://doi.org/10.1063/PT.3.4020}
  {https://doi.org/10.1063/PT.3.4020} \BibitemShut {NoStop}%
\bibitem [{\citenamefont {von Keyserlingk}\ \emph {et~al.}(2016)\citenamefont
  {von Keyserlingk}, \citenamefont {Khemani},\ and\ \citenamefont
  {Sondhi}}]{Floquetlbits1}%
  \BibitemOpen
  \bibfield  {author} {\bibinfo {author} {\bibfnamefont {C.~W.}\ \bibnamefont
  {von Keyserlingk}}, \bibinfo {author} {\bibfnamefont {V.}~\bibnamefont
  {Khemani}},\ and\ \bibinfo {author} {\bibfnamefont {S.~L.}\ \bibnamefont
  {Sondhi}},\ }\bibfield  {title} {\bibinfo {title} {Absolute stability and
  spatiotemporal long-range order in floquet systems},\ }\href
  {https://doi.org/10.1103/PhysRevB.94.085112} {\bibfield  {journal} {\bibinfo
  {journal} {Phys. Rev. B}\ }\textbf {\bibinfo {volume} {94}},\ \bibinfo
  {pages} {085112} (\bibinfo {year} {2016})}\BibitemShut {NoStop}%
\bibitem [{\citenamefont {von Keyserlingk}\ and\ \citenamefont
  {Sondhi}(2016)}]{Floquetlbits2}%
  \BibitemOpen
  \bibfield  {author} {\bibinfo {author} {\bibfnamefont {C.~W.}\ \bibnamefont
  {von Keyserlingk}}\ and\ \bibinfo {author} {\bibfnamefont {S.~L.}\
  \bibnamefont {Sondhi}},\ }\bibfield  {title} {\bibinfo {title} {Phase
  structure of one-dimensional interacting floquet systems. ii. symmetry-broken
  phases},\ }\href {https://doi.org/10.1103/PhysRevB.93.245146} {\bibfield
  {journal} {\bibinfo  {journal} {Phys. Rev. B}\ }\textbf {\bibinfo {volume}
  {93}},\ \bibinfo {pages} {245146} (\bibinfo {year} {2016})}\BibitemShut
  {NoStop}%
\bibitem [{\citenamefont {Chinzei}\ and\ \citenamefont
  {Ikeda}(2020)}]{Chinzei}%
  \BibitemOpen
  \bibfield  {author} {\bibinfo {author} {\bibfnamefont {K.}~\bibnamefont
  {Chinzei}}\ and\ \bibinfo {author} {\bibfnamefont {T.~N.}\ \bibnamefont
  {Ikeda}},\ }\bibfield  {title} {\bibinfo {title} {Time crystals protected by
  floquet dynamical symmetry in hubbard models},\ }\href
  {https://doi.org/10.1103/PhysRevLett.125.060601} {\bibfield  {journal}
  {\bibinfo  {journal} {Phys. Rev. Lett.}\ }\textbf {\bibinfo {volume} {125}},\
  \bibinfo {pages} {060601} (\bibinfo {year} {2020})}\BibitemShut {NoStop}%
\bibitem [{\citenamefont {Chinzei}\ and\ \citenamefont
  {Ikeda}(2021)}]{Chinzei2}%
  \BibitemOpen
  \bibfield  {author} {\bibinfo {author} {\bibfnamefont {K.}~\bibnamefont
  {Chinzei}}\ and\ \bibinfo {author} {\bibfnamefont {T.~N.}\ \bibnamefont
  {Ikeda}},\ }\href@noop {} {\bibinfo {title} {Criticality and rigidity of
  dissipative discrete time crystals in solids}} (\bibinfo {year} {2021}),\
  \Eprint {https://arxiv.org/abs/2110.00591} {arXiv:2110.00591
  [cond-mat.stat-mech]} \BibitemShut {NoStop}%
\bibitem [{\citenamefont {Sarkar}\ and\ \citenamefont {Dubi}(2022)}]{Subhajit}%
  \BibitemOpen
  \bibfield  {author} {\bibinfo {author} {\bibfnamefont {S.}~\bibnamefont
  {Sarkar}}\ and\ \bibinfo {author} {\bibfnamefont {Y.}~\bibnamefont {Dubi}},\
  }\bibfield  {title} {\bibinfo {title} {Signatures of discrete
  time-crystallinity in transport through an open fermionic chain},\ }\href
  {https://doi.org/10.1038/s42005-022-00925-z} {\bibfield  {journal} {\bibinfo
  {journal} {Communications Physics}\ }\textbf {\bibinfo {volume} {5}},\
  \bibinfo {pages} {155} (\bibinfo {year} {2022})}\BibitemShut {NoStop}%
\bibitem [{\citenamefont {Tucker}\ \emph {et~al.}(2018)\citenamefont {Tucker},
  \citenamefont {Zhu}, \citenamefont {Lewis-Swan}, \citenamefont {Marino},
  \citenamefont {Jimenez}, \citenamefont {Restrepo},\ and\ \citenamefont
  {Rey}}]{Jamir1}%
  \BibitemOpen
  \bibfield  {author} {\bibinfo {author} {\bibfnamefont {K.}~\bibnamefont
  {Tucker}}, \bibinfo {author} {\bibfnamefont {B.}~\bibnamefont {Zhu}},
  \bibinfo {author} {\bibfnamefont {R.~J.}\ \bibnamefont {Lewis-Swan}},
  \bibinfo {author} {\bibfnamefont {J.}~\bibnamefont {Marino}}, \bibinfo
  {author} {\bibfnamefont {F.}~\bibnamefont {Jimenez}}, \bibinfo {author}
  {\bibfnamefont {J.~G.}\ \bibnamefont {Restrepo}},\ and\ \bibinfo {author}
  {\bibfnamefont {A.~M.}\ \bibnamefont {Rey}},\ }\bibfield  {title} {\bibinfo
  {title} {Shattered time: can a dissipative time crystal survive many-body
  correlations?},\ }\href@noop {} {\bibfield  {journal} {\bibinfo  {journal}
  {New Journal of Physics}\ }\textbf {\bibinfo {volume} {20}},\ \bibinfo
  {pages} {123003} (\bibinfo {year} {2018})}\BibitemShut {NoStop}%
\bibitem [{\citenamefont {Zhu}\ \emph {et~al.}(2019)\citenamefont {Zhu},
  \citenamefont {Marino}, \citenamefont {Yao}, \citenamefont {Lukin},\ and\
  \citenamefont {Demler}}]{Jamir2}%
  \BibitemOpen
  \bibfield  {author} {\bibinfo {author} {\bibfnamefont {B.}~\bibnamefont
  {Zhu}}, \bibinfo {author} {\bibfnamefont {J.}~\bibnamefont {Marino}},
  \bibinfo {author} {\bibfnamefont {N.~Y.}\ \bibnamefont {Yao}}, \bibinfo
  {author} {\bibfnamefont {M.~D.}\ \bibnamefont {Lukin}},\ and\ \bibinfo
  {author} {\bibfnamefont {E.~A.}\ \bibnamefont {Demler}},\ }\bibfield  {title}
  {\bibinfo {title} {Dicke time crystals in driven-dissipative quantum
  many-body systems},\ }\href@noop {} {\bibfield  {journal} {\bibinfo
  {journal} {New Journal of Physics}\ }\textbf {\bibinfo {volume} {21}},\
  \bibinfo {pages} {073028} (\bibinfo {year} {2019})}\BibitemShut {NoStop}%
\bibitem [{\citenamefont {Braver}\ \emph {et~al.}(2022)\citenamefont {Braver},
  \citenamefont {Fan}, \citenamefont {\ifmmode~\check{Z}\else \v{Z}\fi{}labys},
  \citenamefont {Anisimovas},\ and\ \citenamefont {Sacha}}]{Sacha}%
  \BibitemOpen
  \bibfield  {author} {\bibinfo {author} {\bibfnamefont {Y.}~\bibnamefont
  {Braver}}, \bibinfo {author} {\bibfnamefont {C.-h.}\ \bibnamefont {Fan}},
  \bibinfo {author} {\bibfnamefont {G.}~\bibnamefont {\ifmmode~\check{Z}\else
  \v{Z}\fi{}labys}}, \bibinfo {author} {\bibfnamefont {E.}~\bibnamefont
  {Anisimovas}},\ and\ \bibinfo {author} {\bibfnamefont {K.}~\bibnamefont
  {Sacha}},\ }\bibfield  {title} {\bibinfo {title} {Two-dimensional thouless
  pumping in time-space crystalline structures},\ }\href
  {https://doi.org/10.1103/PhysRevB.106.144301} {\bibfield  {journal} {\bibinfo
   {journal} {Phys. Rev. B}\ }\textbf {\bibinfo {volume} {106}},\ \bibinfo
  {pages} {144301} (\bibinfo {year} {2022})}\BibitemShut {NoStop}%
\bibitem [{\citenamefont {Lazarides}\ \emph {et~al.}(2020)\citenamefont
  {Lazarides}, \citenamefont {Roy}, \citenamefont {Piazza},\ and\ \citenamefont
  {Moessner}}]{LazaridesDissipation}%
  \BibitemOpen
  \bibfield  {author} {\bibinfo {author} {\bibfnamefont {A.}~\bibnamefont
  {Lazarides}}, \bibinfo {author} {\bibfnamefont {S.}~\bibnamefont {Roy}},
  \bibinfo {author} {\bibfnamefont {F.}~\bibnamefont {Piazza}},\ and\ \bibinfo
  {author} {\bibfnamefont {R.}~\bibnamefont {Moessner}},\ }\bibfield  {title}
  {\bibinfo {title} {Time crystallinity in dissipative floquet systems},\
  }\bibfield  {journal} {\bibinfo  {journal} {Physical Review Research}\
  }\textbf {\bibinfo {volume} {2}},\ \href
  {https://doi.org/10.1103/physrevresearch.2.022002}
  {10.1103/physrevresearch.2.022002} (\bibinfo {year} {2020})\BibitemShut
  {NoStop}%
\bibitem [{\citenamefont {Booker}\ \emph {et~al.}(2020)\citenamefont {Booker},
  \citenamefont {Buča},\ and\ \citenamefont {Jaksch}}]{Booker_2020}%
  \BibitemOpen
  \bibfield  {author} {\bibinfo {author} {\bibfnamefont {C.}~\bibnamefont
  {Booker}}, \bibinfo {author} {\bibfnamefont {B.}~\bibnamefont {Buča}},\ and\
  \bibinfo {author} {\bibfnamefont {D.}~\bibnamefont {Jaksch}},\ }\bibfield
  {title} {\bibinfo {title} {Non-stationarity and dissipative time crystals:
  Spectral properties and finite-size effects},\ }\bibfield  {journal}
  {\bibinfo  {journal} {New Journal of Physics}\ }\href
  {https://doi.org/10.1088/1367-2630/ababc4} {10.1088/1367-2630/ababc4}
  (\bibinfo {year} {2020})\BibitemShut {NoStop}%
\bibitem [{\citenamefont {Sánchez~Muñoz}\ \emph {et~al.}(2019)\citenamefont
  {Sánchez~Muñoz}, \citenamefont {Buča}, \citenamefont {Tindall},
  \citenamefont {González-Tudela}, \citenamefont {Jaksch},\ and\ \citenamefont
  {Porras}}]{Carlos}%
  \BibitemOpen
  \bibfield  {author} {\bibinfo {author} {\bibfnamefont {C.}~\bibnamefont
  {Sánchez~Muñoz}}, \bibinfo {author} {\bibfnamefont {B.}~\bibnamefont
  {Buča}}, \bibinfo {author} {\bibfnamefont {J.}~\bibnamefont {Tindall}},
  \bibinfo {author} {\bibfnamefont {A.}~\bibnamefont {González-Tudela}},
  \bibinfo {author} {\bibfnamefont {D.}~\bibnamefont {Jaksch}},\ and\ \bibinfo
  {author} {\bibfnamefont {D.}~\bibnamefont {Porras}},\ }\bibfield  {title}
  {\bibinfo {title} {Symmetries and conservation laws in quantum trajectories:
  Dissipative freezing},\ }\bibfield  {journal} {\bibinfo  {journal} {Physical
  Review A}\ }\textbf {\bibinfo {volume} {100}},\ \href
  {https://doi.org/10.1103/physreva.100.042113} {10.1103/physreva.100.042113}
  (\bibinfo {year} {2019})\BibitemShut {NoStop}%
\bibitem [{\citenamefont {Scarlatella}\ \emph {et~al.}(2019)\citenamefont
  {Scarlatella}, \citenamefont {Fazio},\ and\ \citenamefont
  {Schiró}}]{Orazio}%
  \BibitemOpen
  \bibfield  {author} {\bibinfo {author} {\bibfnamefont {O.}~\bibnamefont
  {Scarlatella}}, \bibinfo {author} {\bibfnamefont {R.}~\bibnamefont {Fazio}},\
  and\ \bibinfo {author} {\bibfnamefont {M.}~\bibnamefont {Schiró}},\
  }\bibfield  {title} {\bibinfo {title} {Emergent finite frequency criticality
  of driven-dissipative correlated lattice bosons},\ }\bibfield  {journal}
  {\bibinfo  {journal} {Physical Review B}\ }\textbf {\bibinfo {volume} {99}},\
  \href {https://doi.org/10.1103/physrevb.99.064511}
  {10.1103/physrevb.99.064511} (\bibinfo {year} {2019})\BibitemShut {NoStop}%
\bibitem [{\citenamefont {Krishna}\ \emph {et~al.}(2022)\citenamefont
  {Krishna}, \citenamefont {Solanki}, \citenamefont {Hajdušek},\ and\
  \citenamefont {Vinjanampathy}}]{MeasurementTC}%
  \BibitemOpen
  \bibfield  {author} {\bibinfo {author} {\bibfnamefont {M.}~\bibnamefont
  {Krishna}}, \bibinfo {author} {\bibfnamefont {P.}~\bibnamefont {Solanki}},
  \bibinfo {author} {\bibfnamefont {M.}~\bibnamefont {Hajdušek}},\ and\
  \bibinfo {author} {\bibfnamefont {S.}~\bibnamefont {Vinjanampathy}},\ }\href
  {https://doi.org/10.48550/ARXIV.2206.14438} {\bibinfo {title} {Measurement
  induced continuous time crystals}} (\bibinfo {year} {2022})\BibitemShut
  {NoStop}%
\bibitem [{\citenamefont {Dubois}\ \emph {et~al.}(2022)\citenamefont {Dubois},
  \citenamefont {Saalmann},\ and\ \citenamefont {Rost}}]{SymmetryInduced}%
  \BibitemOpen
  \bibfield  {author} {\bibinfo {author} {\bibfnamefont {J.}~\bibnamefont
  {Dubois}}, \bibinfo {author} {\bibfnamefont {U.}~\bibnamefont {Saalmann}},\
  and\ \bibinfo {author} {\bibfnamefont {J.~M.}\ \bibnamefont {Rost}},\ }\href
  {https://doi.org/10.48550/ARXIV.2205.10057} {\bibinfo {title}
  {Symmetry-induced decoherence-free subspaces}} (\bibinfo {year}
  {2022})\BibitemShut {NoStop}%
\bibitem [{\citenamefont {Fagotti}\ \emph {et~al.}(2022)\citenamefont
  {Fagotti}, \citenamefont {Marić},\ and\ \citenamefont {Zadnik}}]{LenartSL1}%
  \BibitemOpen
  \bibfield  {author} {\bibinfo {author} {\bibfnamefont {M.}~\bibnamefont
  {Fagotti}}, \bibinfo {author} {\bibfnamefont {V.}~\bibnamefont {Marić}},\
  and\ \bibinfo {author} {\bibfnamefont {L.}~\bibnamefont {Zadnik}},\ }\href
  {https://doi.org/10.48550/ARXIV.2205.02221} {\bibinfo {title} {Nonequilibrium
  symmetry-protected topological order: emergence of semilocal gibbs
  ensembles}} (\bibinfo {year} {2022})\BibitemShut {NoStop}%
\bibitem [{\citenamefont {Fagotti}(2022)}]{Mauriziosemilocal}%
  \BibitemOpen
  \bibfield  {author} {\bibinfo {author} {\bibfnamefont {M.}~\bibnamefont
  {Fagotti}},\ }\bibfield  {title} {\bibinfo {title} {Global quenches after
  localized perturbations},\ }\bibfield  {journal} {\bibinfo  {journal}
  {Physical Review Letters}\ }\textbf {\bibinfo {volume} {128}},\ \href
  {https://doi.org/10.1103/physrevlett.128.110602}
  {10.1103/physrevlett.128.110602} (\bibinfo {year} {2022})\BibitemShut
  {NoStop}%
\bibitem [{\citenamefont {Bidzhiev}\ \emph {et~al.}(2022)\citenamefont
  {Bidzhiev}, \citenamefont {Fagotti},\ and\ \citenamefont
  {Zadnik}}]{LenartSL2}%
  \BibitemOpen
  \bibfield  {author} {\bibinfo {author} {\bibfnamefont {K.}~\bibnamefont
  {Bidzhiev}}, \bibinfo {author} {\bibfnamefont {M.}~\bibnamefont {Fagotti}},\
  and\ \bibinfo {author} {\bibfnamefont {L.}~\bibnamefont {Zadnik}},\
  }\bibfield  {title} {\bibinfo {title} {Macroscopic effects of localized
  measurements in jammed states of quantum spin chains},\ }\bibfield  {journal}
  {\bibinfo  {journal} {Physical Review Letters}\ }\textbf {\bibinfo {volume}
  {128}},\ \href {https://doi.org/10.1103/physrevlett.128.130603}
  {10.1103/physrevlett.128.130603} (\bibinfo {year} {2022})\BibitemShut
  {NoStop}%
\bibitem [{\citenamefont {Alaeian}\ and\ \citenamefont {Bu{\v
  c}a}(2022)}]{HadisehMulti}%
  \BibitemOpen
  \bibfield  {author} {\bibinfo {author} {\bibfnamefont {H.}~\bibnamefont
  {Alaeian}}\ and\ \bibinfo {author} {\bibfnamefont {B.}~\bibnamefont {Bu{\v
  c}a}},\ }\bibfield  {title} {\bibinfo {title} {Exact multistability and
  dissipative time crystals in interacting fermionic lattices},\ }\href
  {https://doi.org/10.1038/s42005-022-01090-z} {\bibfield  {journal} {\bibinfo
  {journal} {Communications Physics}\ }\textbf {\bibinfo {volume} {5}},\
  \bibinfo {pages} {318} (\bibinfo {year} {2022})}\BibitemShut {NoStop}%
\bibitem [{\citenamefont {Zadnik}\ and\ \citenamefont
  {Fagotti}(2021)}]{LenartFolded1}%
  \BibitemOpen
  \bibfield  {author} {\bibinfo {author} {\bibfnamefont {L.}~\bibnamefont
  {Zadnik}}\ and\ \bibinfo {author} {\bibfnamefont {M.}~\bibnamefont
  {Fagotti}},\ }\bibfield  {title} {\bibinfo {title} {{The Folded Spin-1/2 XXZ
  Model: I. Diagonalisation, Jamming, and Ground State Properties}},\ }\href
  {https://doi.org/10.21468/SciPostPhysCore.4.2.010} {\bibfield  {journal}
  {\bibinfo  {journal} {SciPost Phys. Core}\ }\textbf {\bibinfo {volume} {4}},\
  \bibinfo {pages} {10} (\bibinfo {year} {2021})}\BibitemShut {NoStop}%
\bibitem [{\citenamefont {Zadnik}\ \emph {et~al.}(2021)\citenamefont {Zadnik},
  \citenamefont {Bidzhiev},\ and\ \citenamefont {Fagotti}}]{LenardFolded2}%
  \BibitemOpen
  \bibfield  {author} {\bibinfo {author} {\bibfnamefont {L.}~\bibnamefont
  {Zadnik}}, \bibinfo {author} {\bibfnamefont {K.}~\bibnamefont {Bidzhiev}},\
  and\ \bibinfo {author} {\bibfnamefont {M.}~\bibnamefont {Fagotti}},\
  }\bibfield  {title} {\bibinfo {title} {{The Folded Spin-1/2 XXZ Model: II.
  Thermodynamics and Hydrodynamics with a Minimal Set of Charges}},\ }\href
  {https://doi.org/10.21468/SciPostPhys.10.5.099} {\bibfield  {journal}
  {\bibinfo  {journal} {SciPost Phys.}\ }\textbf {\bibinfo {volume} {10}},\
  \bibinfo {pages} {99} (\bibinfo {year} {2021})}\BibitemShut {NoStop}%
\bibitem [{\citenamefont {Gravina}\ \emph {et~al.}(2022)\citenamefont
  {Gravina}, \citenamefont {Minganti},\ and\ \citenamefont
  {Savona}}]{CatQUBIT}%
  \BibitemOpen
  \bibfield  {author} {\bibinfo {author} {\bibfnamefont {L.}~\bibnamefont
  {Gravina}}, \bibinfo {author} {\bibfnamefont {F.}~\bibnamefont {Minganti}},\
  and\ \bibinfo {author} {\bibfnamefont {V.}~\bibnamefont {Savona}},\ }\href
  {https://doi.org/10.48550/ARXIV.2208.04928} {\bibinfo {title} {A critical
  schrödinger cat qubit}} (\bibinfo {year} {2022})\BibitemShut {NoStop}%
\bibitem [{\citenamefont {Mierzejewski}\ \emph {et~al.}(2015)\citenamefont
  {Mierzejewski}, \citenamefont {Prelov\ifmmode~\check{s}\else \v{s}\fi{}ek},\
  and\ \citenamefont {Prosen}}]{Prosennumerical}%
  \BibitemOpen
  \bibfield  {author} {\bibinfo {author} {\bibfnamefont {M.}~\bibnamefont
  {Mierzejewski}}, \bibinfo {author} {\bibfnamefont {P.}~\bibnamefont
  {Prelov\ifmmode~\check{s}\else \v{s}\fi{}ek}},\ and\ \bibinfo {author}
  {\bibfnamefont {T.~c.~v.}\ \bibnamefont {Prosen}},\ }\bibfield  {title}
  {\bibinfo {title} {Identifying local and quasilocal conserved quantities in
  integrable systems},\ }\href {https://doi.org/10.1103/PhysRevLett.114.140601}
  {\bibfield  {journal} {\bibinfo  {journal} {Phys. Rev. Lett.}\ }\textbf
  {\bibinfo {volume} {114}},\ \bibinfo {pages} {140601} (\bibinfo {year}
  {2015})}\BibitemShut {NoStop}%
\bibitem [{\citenamefont {Sarkar}\ and\ \citenamefont
  {Buča}(2022)}]{SarkarBuca}%
  \BibitemOpen
  \bibfield  {author} {\bibinfo {author} {\bibfnamefont {S.}~\bibnamefont
  {Sarkar}}\ and\ \bibinfo {author} {\bibfnamefont {B.}~\bibnamefont {Buča}},\
  }\href {https://doi.org/10.48550/ARXIV.2204.13354} {\bibinfo {title}
  {Protecting coherence from the environment via stark many-body localization
  in a quantum-dot simulator}} (\bibinfo {year} {2022})\BibitemShut {NoStop}%
\bibitem [{\citenamefont {Lange}\ \emph {et~al.}(2018)\citenamefont {Lange},
  \citenamefont {Lenar\ifmmode \check{c}\else
  \v{c}\fi{}i\ifmmode~\check{c}\else \v{c}\fi{}},\ and\ \citenamefont
  {Rosch}}]{Zala}%
  \BibitemOpen
  \bibfield  {author} {\bibinfo {author} {\bibfnamefont {F.}~\bibnamefont
  {Lange}}, \bibinfo {author} {\bibfnamefont {Z.}~\bibnamefont {Lenar\ifmmode
  \check{c}\else \v{c}\fi{}i\ifmmode~\check{c}\else \v{c}\fi{}}},\ and\
  \bibinfo {author} {\bibfnamefont {A.}~\bibnamefont {Rosch}},\ }\bibfield
  {title} {\bibinfo {title} {Time-dependent generalized gibbs ensembles in open
  quantum systems},\ }\href {https://doi.org/10.1103/PhysRevB.97.165138}
  {\bibfield  {journal} {\bibinfo  {journal} {Phys. Rev. B}\ }\textbf {\bibinfo
  {volume} {97}},\ \bibinfo {pages} {165138} (\bibinfo {year}
  {2018})}\BibitemShut {NoStop}%
\bibitem [{\citenamefont {Fagotti}(2014)}]{Fagotti_2014}%
  \BibitemOpen
  \bibfield  {author} {\bibinfo {author} {\bibfnamefont {M.}~\bibnamefont
  {Fagotti}},\ }\bibfield  {title} {\bibinfo {title} {On conservation laws,
  relaxation and pre-relaxation after a quantum quench},\ }\href
  {https://doi.org/10.1088/1742-5468/2014/03/P03016} {\bibfield  {journal}
  {\bibinfo  {journal} {Journal of Statistical Mechanics: Theory and
  Experiment}\ }\textbf {\bibinfo {volume} {2014}},\ \bibinfo {pages} {P03016}
  (\bibinfo {year} {2014})}\BibitemShut {NoStop}%
\bibitem [{\citenamefont {Ampelogiannis}\ and\ \citenamefont
  {Doyon}(2021{\natexlab{b}})}]{Doyon2}%
  \BibitemOpen
  \bibfield  {author} {\bibinfo {author} {\bibfnamefont {D.}~\bibnamefont
  {Ampelogiannis}}\ and\ \bibinfo {author} {\bibfnamefont {B.}~\bibnamefont
  {Doyon}},\ }\href {https://doi.org/10.48550/ARXIV.2112.12730} {\bibinfo
  {title} {Almost everywhere ergodicity in quantum lattice models}} (\bibinfo
  {year} {2021}{\natexlab{b}})\BibitemShut {NoStop}%
\bibitem [{\citenamefont {Knill}\ \emph {et~al.}(1998)\citenamefont {Knill},
  \citenamefont {Laflamme},\ and\ \citenamefont {Zurek}}]{threshold}%
  \BibitemOpen
  \bibfield  {author} {\bibinfo {author} {\bibfnamefont {E.}~\bibnamefont
  {Knill}}, \bibinfo {author} {\bibfnamefont {R.}~\bibnamefont {Laflamme}},\
  and\ \bibinfo {author} {\bibfnamefont {W.~H.}\ \bibnamefont {Zurek}},\
  }\bibfield  {title} {\bibinfo {title} {Resilient quantum computation},\
  }\href {https://doi.org/10.1126/science.279.5349.342} {\bibfield  {journal}
  {\bibinfo  {journal} {Science}\ }\textbf {\bibinfo {volume} {279}},\ \bibinfo
  {pages} {342} (\bibinfo {year} {1998})},\ \Eprint
  {https://arxiv.org/abs/https://www.science.org/doi/pdf/10.1126/science.279.5349.342}
  {https://www.science.org/doi/pdf/10.1126/science.279.5349.342} \BibitemShut
  {NoStop}%
\bibitem [{\citenamefont {Friedman}\ \emph {et~al.}(2022)\citenamefont
  {Friedman}, \citenamefont {Yin}, \citenamefont {Hong},\ and\ \citenamefont
  {Lucas}}]{LocalityEnt}%
  \BibitemOpen
  \bibfield  {author} {\bibinfo {author} {\bibfnamefont {A.~J.}\ \bibnamefont
  {Friedman}}, \bibinfo {author} {\bibfnamefont {C.}~\bibnamefont {Yin}},
  \bibinfo {author} {\bibfnamefont {Y.}~\bibnamefont {Hong}},\ and\ \bibinfo
  {author} {\bibfnamefont {A.}~\bibnamefont {Lucas}},\ }\href
  {https://doi.org/10.48550/ARXIV.2206.09929} {\bibinfo {title} {Locality and
  error correction in quantum dynamics with measurement}} (\bibinfo {year}
  {2022})\BibitemShut {NoStop}%
\bibitem [{\citenamefont {Kim}\ and\ \citenamefont
  {Kastoryano}(2017)}]{LocalityEnt2}%
  \BibitemOpen
  \bibfield  {author} {\bibinfo {author} {\bibfnamefont {I.~H.}\ \bibnamefont
  {Kim}}\ and\ \bibinfo {author} {\bibfnamefont {M.~J.}\ \bibnamefont
  {Kastoryano}},\ }\bibfield  {title} {\bibinfo {title} {Entanglement
  renormalization, quantum error correction, and bulk causality},\ }\href
  {https://doi.org/10.1007/JHEP04(2017)040} {\bibfield  {journal} {\bibinfo
  {journal} {Journal of High Energy Physics}\ }\textbf {\bibinfo {volume}
  {2017}},\ \bibinfo {pages} {40} (\bibinfo {year} {2017})}\BibitemShut
  {NoStop}%
\bibitem [{\citenamefont {Ippoliti}\ and\ \citenamefont
  {Ho}(2022)}]{DeepTherm1}%
  \BibitemOpen
  \bibfield  {author} {\bibinfo {author} {\bibfnamefont {M.}~\bibnamefont
  {Ippoliti}}\ and\ \bibinfo {author} {\bibfnamefont {W.~W.}\ \bibnamefont
  {Ho}},\ }\href {https://doi.org/10.48550/ARXIV.2204.13657} {\bibinfo {title}
  {Dynamical purification and the emergence of quantum state designs from the
  projected ensemble}} (\bibinfo {year} {2022})\BibitemShut {NoStop}%
\bibitem [{\citenamefont {Lucas}\ \emph {et~al.}(2022)\citenamefont {Lucas},
  \citenamefont {Piroli}, \citenamefont {De~Nardis},\ and\ \citenamefont
  {De~Luca}}]{DeepTherm2}%
  \BibitemOpen
  \bibfield  {author} {\bibinfo {author} {\bibfnamefont {M.}~\bibnamefont
  {Lucas}}, \bibinfo {author} {\bibfnamefont {L.}~\bibnamefont {Piroli}},
  \bibinfo {author} {\bibfnamefont {J.}~\bibnamefont {De~Nardis}},\ and\
  \bibinfo {author} {\bibfnamefont {A.}~\bibnamefont {De~Luca}},\ }\href
  {https://doi.org/10.48550/ARXIV.2207.13628} {\bibinfo {title} {Generalized
  deep thermalization for free fermions}} (\bibinfo {year} {2022})\BibitemShut
  {NoStop}%
\bibitem [{\citenamefont {Claeys}\ and\ \citenamefont
  {Lamacraft}(2022)}]{Claeys2022emergentquantum}%
  \BibitemOpen
  \bibfield  {author} {\bibinfo {author} {\bibfnamefont {P.~W.}\ \bibnamefont
  {Claeys}}\ and\ \bibinfo {author} {\bibfnamefont {A.}~\bibnamefont
  {Lamacraft}},\ }\bibfield  {title} {\bibinfo {title} {Emergent quantum state
  designs and biunitarity in dual-unitary circuit dynamics},\ }\href
  {https://doi.org/10.22331/q-2022-06-15-738} {\bibfield  {journal} {\bibinfo
  {journal} {{Quantum}}\ }\textbf {\bibinfo {volume} {6}},\ \bibinfo {pages}
  {738} (\bibinfo {year} {2022})}\BibitemShut {NoStop}%
\bibitem [{\citenamefont {Claeys}(2023)}]{Claeys2023universalityin}%
  \BibitemOpen
  \bibfield  {author} {\bibinfo {author} {\bibfnamefont {P.~W.}\ \bibnamefont
  {Claeys}},\ }\bibfield  {title} {\bibinfo {title} {Universality in quantum
  snapshots},\ }\href {https://doi.org/10.22331/qv-2023-01-27-71} {\bibfield
  {journal} {\bibinfo  {journal} {{Quantum Views}}\ }\textbf {\bibinfo {volume}
  {7}},\ \bibinfo {pages} {71} (\bibinfo {year} {2023})}\BibitemShut {NoStop}%
\bibitem [{\citenamefont {Kuwahara}\ and\ \citenamefont
  {Saito}(2021)}]{LiebRobinsonBosons}%
  \BibitemOpen
  \bibfield  {author} {\bibinfo {author} {\bibfnamefont {T.}~\bibnamefont
  {Kuwahara}}\ and\ \bibinfo {author} {\bibfnamefont {K.}~\bibnamefont
  {Saito}},\ }\bibfield  {title} {\bibinfo {title} {Lieb-robinson bound and
  almost-linear light cone in interacting boson systems},\ }\href
  {https://doi.org/10.1103/PhysRevLett.127.070403} {\bibfield  {journal}
  {\bibinfo  {journal} {Phys. Rev. Lett.}\ }\textbf {\bibinfo {volume} {127}},\
  \bibinfo {pages} {070403} (\bibinfo {year} {2021})}\BibitemShut {NoStop}%
\bibitem [{\citenamefont {Matsuta}\ \emph {et~al.}(2016)\citenamefont
  {Matsuta}, \citenamefont {Koma},\ and\ \citenamefont
  {Nakamura}}]{Matsuta_2016}%
  \BibitemOpen
  \bibfield  {author} {\bibinfo {author} {\bibfnamefont {T.}~\bibnamefont
  {Matsuta}}, \bibinfo {author} {\bibfnamefont {T.}~\bibnamefont {Koma}},\ and\
  \bibinfo {author} {\bibfnamefont {S.}~\bibnamefont {Nakamura}},\ }\bibfield
  {title} {\bibinfo {title} {Improving the lieb{\textendash}robinson bound for
  long-range interactions},\ }\href {https://doi.org/10.1007/s00023-016-0526-1}
  {\bibfield  {journal} {\bibinfo  {journal} {Annales Henri Poincar{\'{e}}}\
  }\textbf {\bibinfo {volume} {18}},\ \bibinfo {pages} {519} (\bibinfo {year}
  {2016})}\BibitemShut {NoStop}%
\bibitem [{\citenamefont {Passarelli}\ \emph {et~al.}(2022)\citenamefont
  {Passarelli}, \citenamefont {Lucignano}, \citenamefont {Fazio},\ and\
  \citenamefont {Russomanno}}]{Passarelli_2022}%
  \BibitemOpen
  \bibfield  {author} {\bibinfo {author} {\bibfnamefont {G.}~\bibnamefont
  {Passarelli}}, \bibinfo {author} {\bibfnamefont {P.}~\bibnamefont
  {Lucignano}}, \bibinfo {author} {\bibfnamefont {R.}~\bibnamefont {Fazio}},\
  and\ \bibinfo {author} {\bibfnamefont {A.}~\bibnamefont {Russomanno}},\
  }\bibfield  {title} {\bibinfo {title} {Dissipative time crystals with
  long-range lindbladians},\ }\bibfield  {journal} {\bibinfo  {journal}
  {Physical Review B}\ }\textbf {\bibinfo {volume} {106}},\ \href
  {https://doi.org/10.1103/physrevb.106.224308} {10.1103/physrevb.106.224308}
  (\bibinfo {year} {2022})\BibitemShut {NoStop}%
\bibitem [{\citenamefont {Dogra}\ \emph {et~al.}(2019)\citenamefont {Dogra},
  \citenamefont {Landini}, \citenamefont {Kroeger}, \citenamefont {Hruby},
  \citenamefont {Donner},\ and\ \citenamefont {Esslinger}}]{Esslinger}%
  \BibitemOpen
  \bibfield  {author} {\bibinfo {author} {\bibfnamefont {N.}~\bibnamefont
  {Dogra}}, \bibinfo {author} {\bibfnamefont {M.}~\bibnamefont {Landini}},
  \bibinfo {author} {\bibfnamefont {K.}~\bibnamefont {Kroeger}}, \bibinfo
  {author} {\bibfnamefont {L.}~\bibnamefont {Hruby}}, \bibinfo {author}
  {\bibfnamefont {T.}~\bibnamefont {Donner}},\ and\ \bibinfo {author}
  {\bibfnamefont {T.}~\bibnamefont {Esslinger}},\ }\bibfield  {title} {\bibinfo
  {title} {Dissipation-induced structural instability and chiral dynamics in a
  quantum gas},\ }\href {https://doi.org/10.1126/science.aaw4465} {\bibfield
  {journal} {\bibinfo  {journal} {Science}\ }\textbf {\bibinfo {volume}
  {366}},\ \bibinfo {pages} {1496} (\bibinfo {year} {2019})},\ \Eprint
  {https://arxiv.org/abs/https://science.sciencemag.org/content/366/6472/1496.full.pdf}
  {https://science.sciencemag.org/content/366/6472/1496.full.pdf} \BibitemShut
  {NoStop}%
\bibitem [{\citenamefont {Zupancic}\ \emph {et~al.}(2019)\citenamefont
  {Zupancic}, \citenamefont {Dreon}, \citenamefont {Li}, \citenamefont
  {Baumg\"artner}, \citenamefont {Morales}, \citenamefont {Zheng},
  \citenamefont {Cooper}, \citenamefont {Esslinger},\ and\ \citenamefont
  {Donner}}]{Esslinger2}%
  \BibitemOpen
  \bibfield  {author} {\bibinfo {author} {\bibfnamefont {P.}~\bibnamefont
  {Zupancic}}, \bibinfo {author} {\bibfnamefont {D.}~\bibnamefont {Dreon}},
  \bibinfo {author} {\bibfnamefont {X.}~\bibnamefont {Li}}, \bibinfo {author}
  {\bibfnamefont {A.}~\bibnamefont {Baumg\"artner}}, \bibinfo {author}
  {\bibfnamefont {A.}~\bibnamefont {Morales}}, \bibinfo {author} {\bibfnamefont
  {W.}~\bibnamefont {Zheng}}, \bibinfo {author} {\bibfnamefont {N.~R.}\
  \bibnamefont {Cooper}}, \bibinfo {author} {\bibfnamefont {T.}~\bibnamefont
  {Esslinger}},\ and\ \bibinfo {author} {\bibfnamefont {T.}~\bibnamefont
  {Donner}},\ }\bibfield  {title} {\bibinfo {title} {$p$-band induced
  self-organization and dynamics with repulsively driven ultracold atoms in an
  optical cavity},\ }\href {https://doi.org/10.1103/PhysRevLett.123.233601}
  {\bibfield  {journal} {\bibinfo  {journal} {Phys. Rev. Lett.}\ }\textbf
  {\bibinfo {volume} {123}},\ \bibinfo {pages} {233601} (\bibinfo {year}
  {2019})}\BibitemShut {NoStop}%
\bibitem [{\citenamefont {Iemini}\ \emph {et~al.}(2018)\citenamefont {Iemini},
  \citenamefont {Russomanno}, \citenamefont {Keeling}, \citenamefont
  {Schir\`o}, \citenamefont {Dalmonte},\ and\ \citenamefont {Fazio}}]{Fazio}%
  \BibitemOpen
  \bibfield  {author} {\bibinfo {author} {\bibfnamefont {F.}~\bibnamefont
  {Iemini}}, \bibinfo {author} {\bibfnamefont {A.}~\bibnamefont {Russomanno}},
  \bibinfo {author} {\bibfnamefont {J.}~\bibnamefont {Keeling}}, \bibinfo
  {author} {\bibfnamefont {M.}~\bibnamefont {Schir\`o}}, \bibinfo {author}
  {\bibfnamefont {M.}~\bibnamefont {Dalmonte}},\ and\ \bibinfo {author}
  {\bibfnamefont {R.}~\bibnamefont {Fazio}},\ }\bibfield  {title} {\bibinfo
  {title} {Boundary time crystals},\ }\href
  {https://doi.org/10.1103/PhysRevLett.121.035301} {\bibfield  {journal}
  {\bibinfo  {journal} {Phys. Rev. Lett.}\ }\textbf {\bibinfo {volume} {121}},\
  \bibinfo {pages} {035301} (\bibinfo {year} {2018})}\BibitemShut {NoStop}%
\bibitem [{\citenamefont {Nakanishi}\ and\ \citenamefont
  {Sasamoto}(2023)}]{PTDisTC}%
  \BibitemOpen
  \bibfield  {author} {\bibinfo {author} {\bibfnamefont {Y.}~\bibnamefont
  {Nakanishi}}\ and\ \bibinfo {author} {\bibfnamefont {T.}~\bibnamefont
  {Sasamoto}},\ }\bibfield  {title} {\bibinfo {title} {Dissipative time
  crystals originating from parity-time symmetry},\ }\href
  {https://doi.org/10.1103/PhysRevA.107.L010201} {\bibfield  {journal}
  {\bibinfo  {journal} {Phys. Rev. A}\ }\textbf {\bibinfo {volume} {107}},\
  \bibinfo {pages} {L010201} (\bibinfo {year} {2023})}\BibitemShut {NoStop}%
\bibitem [{\citenamefont {Carollo}\ and\ \citenamefont
  {Lesanovsky}(2021)}]{Lesanovsky}%
  \BibitemOpen
  \bibfield  {author} {\bibinfo {author} {\bibfnamefont {F.}~\bibnamefont
  {Carollo}}\ and\ \bibinfo {author} {\bibfnamefont {I.}~\bibnamefont
  {Lesanovsky}},\ }\href@noop {} {\bibinfo {title} {Exact solution of a
  boundary time-crystal phase transition: time-translation symmetry breaking
  and non-markovian dynamics of correlations}} (\bibinfo {year} {2021}),\
  \Eprint {https://arxiv.org/abs/2110.00030} {arXiv:2110.00030
  [cond-mat.stat-mech]} \BibitemShut {NoStop}%
\bibitem [{\citenamefont {Cabot}\ \emph
  {et~al.}(2022{\natexlab{a}})\citenamefont {Cabot}, \citenamefont {Muhle},
  \citenamefont {Carollo},\ and\ \citenamefont {Lesanovsky}}]{Federico1}%
  \BibitemOpen
  \bibfield  {author} {\bibinfo {author} {\bibfnamefont {A.}~\bibnamefont
  {Cabot}}, \bibinfo {author} {\bibfnamefont {L.~S.}\ \bibnamefont {Muhle}},
  \bibinfo {author} {\bibfnamefont {F.}~\bibnamefont {Carollo}},\ and\ \bibinfo
  {author} {\bibfnamefont {I.}~\bibnamefont {Lesanovsky}},\ }\href
  {https://doi.org/10.48550/ARXIV.2212.06460} {\bibinfo {title} {Quantum
  trajectories of dissipative time-crystals}} (\bibinfo {year}
  {2022}{\natexlab{a}})\BibitemShut {NoStop}%
\bibitem [{\citenamefont {Cabot}\ \emph
  {et~al.}(2022{\natexlab{b}})\citenamefont {Cabot}, \citenamefont {Muhle},
  \citenamefont {Carollo},\ and\ \citenamefont {Lesanovsky}}]{Federico2}%
  \BibitemOpen
  \bibfield  {author} {\bibinfo {author} {\bibfnamefont {A.}~\bibnamefont
  {Cabot}}, \bibinfo {author} {\bibfnamefont {L.~S.}\ \bibnamefont {Muhle}},
  \bibinfo {author} {\bibfnamefont {F.}~\bibnamefont {Carollo}},\ and\ \bibinfo
  {author} {\bibfnamefont {I.}~\bibnamefont {Lesanovsky}},\ }\href
  {https://doi.org/10.48550/ARXIV.2212.06460} {\bibinfo {title} {Quantum
  trajectories of dissipative time-crystals}} (\bibinfo {year}
  {2022}{\natexlab{b}})\BibitemShut {NoStop}%
\bibitem [{\citenamefont {Dreon}\ \emph {et~al.}(2022)\citenamefont {Dreon},
  \citenamefont {Baumg{\"a}rtner}, \citenamefont {Li}, \citenamefont
  {Hertlein}, \citenamefont {Esslinger},\ and\ \citenamefont
  {Donner}}]{Esslinger3}%
  \BibitemOpen
  \bibfield  {author} {\bibinfo {author} {\bibfnamefont {D.}~\bibnamefont
  {Dreon}}, \bibinfo {author} {\bibfnamefont {A.}~\bibnamefont
  {Baumg{\"a}rtner}}, \bibinfo {author} {\bibfnamefont {X.}~\bibnamefont {Li}},
  \bibinfo {author} {\bibfnamefont {S.}~\bibnamefont {Hertlein}}, \bibinfo
  {author} {\bibfnamefont {T.}~\bibnamefont {Esslinger}},\ and\ \bibinfo
  {author} {\bibfnamefont {T.}~\bibnamefont {Donner}},\ }\bibfield  {title}
  {\bibinfo {title} {Self-oscillating pump in a topological dissipative
  atom--cavity system},\ }\href {https://doi.org/10.1038/s41586-022-04970-0}
  {\bibfield  {journal} {\bibinfo  {journal} {Nature}\ }\textbf {\bibinfo
  {volume} {608}},\ \bibinfo {pages} {494} (\bibinfo {year}
  {2022})}\BibitemShut {NoStop}%
\bibitem [{\citenamefont {Moroder}\ \emph {et~al.}(2022)\citenamefont
  {Moroder}, \citenamefont {Grundner}, \citenamefont {Damanet}, \citenamefont
  {Schollwöck}, \citenamefont {Mardazad}, \citenamefont {Flannigan},
  \citenamefont {Köhler},\ and\ \citenamefont {Paeckel}}]{Francois1}%
  \BibitemOpen
  \bibfield  {author} {\bibinfo {author} {\bibfnamefont {M.}~\bibnamefont
  {Moroder}}, \bibinfo {author} {\bibfnamefont {M.}~\bibnamefont {Grundner}},
  \bibinfo {author} {\bibfnamefont {F.}~\bibnamefont {Damanet}}, \bibinfo
  {author} {\bibfnamefont {U.}~\bibnamefont {Schollwöck}}, \bibinfo {author}
  {\bibfnamefont {S.}~\bibnamefont {Mardazad}}, \bibinfo {author}
  {\bibfnamefont {S.}~\bibnamefont {Flannigan}}, \bibinfo {author}
  {\bibfnamefont {T.}~\bibnamefont {Köhler}},\ and\ \bibinfo {author}
  {\bibfnamefont {S.}~\bibnamefont {Paeckel}},\ }\href
  {https://doi.org/10.48550/ARXIV.2207.08243} {\bibinfo {title} {Metallicity in
  the dissipative hubbard-holstein model: Markovian and non-markovian
  tensor-network methods for open quantum many-body systems}} (\bibinfo {year}
  {2022})\BibitemShut {NoStop}%
\bibitem [{\citenamefont {Link}\ \emph {et~al.}(2022)\citenamefont {Link},
  \citenamefont {M\"uller}, \citenamefont {Lena}, \citenamefont {Luoma},
  \citenamefont {Damanet}, \citenamefont {Strunz},\ and\ \citenamefont
  {Daley}}]{Francois2}%
  \BibitemOpen
  \bibfield  {author} {\bibinfo {author} {\bibfnamefont {V.}~\bibnamefont
  {Link}}, \bibinfo {author} {\bibfnamefont {K.}~\bibnamefont {M\"uller}},
  \bibinfo {author} {\bibfnamefont {R.~G.}\ \bibnamefont {Lena}}, \bibinfo
  {author} {\bibfnamefont {K.}~\bibnamefont {Luoma}}, \bibinfo {author}
  {\bibfnamefont {F.~m.~c.}\ \bibnamefont {Damanet}}, \bibinfo {author}
  {\bibfnamefont {W.~T.}\ \bibnamefont {Strunz}},\ and\ \bibinfo {author}
  {\bibfnamefont {A.~J.}\ \bibnamefont {Daley}},\ }\bibfield  {title} {\bibinfo
  {title} {Non-markovian quantum dynamics in strongly coupled multimode
  cavities conditioned on continuous measurement},\ }\href
  {https://doi.org/10.1103/PRXQuantum.3.020348} {\bibfield  {journal} {\bibinfo
   {journal} {PRX Quantum}\ }\textbf {\bibinfo {volume} {3}},\ \bibinfo {pages}
  {020348} (\bibinfo {year} {2022})}\BibitemShut {NoStop}%
\bibitem [{\citenamefont {Flannigan}\ \emph {et~al.}(2022)\citenamefont
  {Flannigan}, \citenamefont {Damanet},\ and\ \citenamefont
  {Daley}}]{Francois3}%
  \BibitemOpen
  \bibfield  {author} {\bibinfo {author} {\bibfnamefont {S.}~\bibnamefont
  {Flannigan}}, \bibinfo {author} {\bibfnamefont {F.}~\bibnamefont {Damanet}},\
  and\ \bibinfo {author} {\bibfnamefont {A.~J.}\ \bibnamefont {Daley}},\
  }\bibfield  {title} {\bibinfo {title} {Many-body quantum state diffusion for
  non-markovian dynamics in strongly interacting systems},\ }\href
  {https://doi.org/10.1103/PhysRevLett.128.063601} {\bibfield  {journal}
  {\bibinfo  {journal} {Phys. Rev. Lett.}\ }\textbf {\bibinfo {volume} {128}},\
  \bibinfo {pages} {063601} (\bibinfo {year} {2022})}\BibitemShut {NoStop}%
\bibitem [{\citenamefont {Imbrie}(2016)}]{Imbrie}%
  \BibitemOpen
  \bibfield  {author} {\bibinfo {author} {\bibfnamefont {J.~Z.}\ \bibnamefont
  {Imbrie}},\ }\bibfield  {title} {\bibinfo {title} {On many-body localization
  for quantum spin chains},\ }\href {https://doi.org/10.1007/s10955-016-1508-x}
  {\bibfield  {journal} {\bibinfo  {journal} {Journal of Statistical Physics}\
  }\textbf {\bibinfo {volume} {163}},\ \bibinfo {pages} {998–1048} (\bibinfo
  {year} {2016})}\BibitemShut {NoStop}%
\bibitem [{\citenamefont {Serbyn}\ \emph {et~al.}(2014)\citenamefont {Serbyn},
  \citenamefont {Papi\ifmmode~\acute{c}\else \'{c}\fi{}},\ and\ \citenamefont
  {Abanin}}]{PapicMBL}%
  \BibitemOpen
  \bibfield  {author} {\bibinfo {author} {\bibfnamefont {M.}~\bibnamefont
  {Serbyn}}, \bibinfo {author} {\bibfnamefont {Z.}~\bibnamefont
  {Papi\ifmmode~\acute{c}\else \'{c}\fi{}}},\ and\ \bibinfo {author}
  {\bibfnamefont {D.~A.}\ \bibnamefont {Abanin}},\ }\bibfield  {title}
  {\bibinfo {title} {Quantum quenches in the many-body localized phase},\
  }\href {https://doi.org/10.1103/PhysRevB.90.174302} {\bibfield  {journal}
  {\bibinfo  {journal} {Phys. Rev. B}\ }\textbf {\bibinfo {volume} {90}},\
  \bibinfo {pages} {174302} (\bibinfo {year} {2014})}\BibitemShut {NoStop}%
\bibitem [{\citenamefont {Yoshinaga}\ \emph {et~al.}(2022)\citenamefont
  {Yoshinaga}, \citenamefont {Hakoshima}, \citenamefont {Imoto}, \citenamefont
  {Matsuzaki},\ and\ \citenamefont {Hamazaki}}]{LOCfrag2}%
  \BibitemOpen
  \bibfield  {author} {\bibinfo {author} {\bibfnamefont {A.}~\bibnamefont
  {Yoshinaga}}, \bibinfo {author} {\bibfnamefont {H.}~\bibnamefont
  {Hakoshima}}, \bibinfo {author} {\bibfnamefont {T.}~\bibnamefont {Imoto}},
  \bibinfo {author} {\bibfnamefont {Y.}~\bibnamefont {Matsuzaki}},\ and\
  \bibinfo {author} {\bibfnamefont {R.}~\bibnamefont {Hamazaki}},\ }\bibfield
  {title} {\bibinfo {title} {Emergence of hilbert space fragmentation in ising
  models with a weak transverse field},\ }\href
  {https://doi.org/10.1103/PhysRevLett.129.090602} {\bibfield  {journal}
  {\bibinfo  {journal} {Phys. Rev. Lett.}\ }\textbf {\bibinfo {volume} {129}},\
  \bibinfo {pages} {090602} (\bibinfo {year} {2022})}\BibitemShut {NoStop}%
\bibitem [{\citenamefont {Castro-Alvaredo}\ \emph {et~al.}(2020)\citenamefont
  {Castro-Alvaredo}, \citenamefont {Lencs\'es}, \citenamefont {Sz\'ecs\'enyi},\
  and\ \citenamefont {Viti}}]{Olalla1}%
  \BibitemOpen
  \bibfield  {author} {\bibinfo {author} {\bibfnamefont {O.~A.}\ \bibnamefont
  {Castro-Alvaredo}}, \bibinfo {author} {\bibfnamefont {M.}~\bibnamefont
  {Lencs\'es}}, \bibinfo {author} {\bibfnamefont {I.~M.}\ \bibnamefont
  {Sz\'ecs\'enyi}},\ and\ \bibinfo {author} {\bibfnamefont {J.}~\bibnamefont
  {Viti}},\ }\bibfield  {title} {\bibinfo {title} {Entanglement oscillations
  near a quantum critical point},\ }\href
  {https://doi.org/10.1103/PhysRevLett.124.230601} {\bibfield  {journal}
  {\bibinfo  {journal} {Phys. Rev. Lett.}\ }\textbf {\bibinfo {volume} {124}},\
  \bibinfo {pages} {230601} (\bibinfo {year} {2020})}\BibitemShut {NoStop}%
\bibitem [{\citenamefont {Castro-Alvaredo}\ \emph {et~al.}(2019)\citenamefont
  {Castro-Alvaredo}, \citenamefont {Lencsés}, \citenamefont {Szécsényi},\
  and\ \citenamefont {Viti}}]{Olalla2}%
  \BibitemOpen
  \bibfield  {author} {\bibinfo {author} {\bibfnamefont {O.~A.}\ \bibnamefont
  {Castro-Alvaredo}}, \bibinfo {author} {\bibfnamefont {M.}~\bibnamefont
  {Lencsés}}, \bibinfo {author} {\bibfnamefont {I.~M.}\ \bibnamefont
  {Szécsényi}},\ and\ \bibinfo {author} {\bibfnamefont {J.}~\bibnamefont
  {Viti}},\ }\bibfield  {title} {\bibinfo {title} {Entanglement dynamics after
  a quench in ising field theory: a branch point twist field approach},\
  }\bibfield  {journal} {\bibinfo  {journal} {Journal of High Energy Physics}\
  }\textbf {\bibinfo {volume} {2019}},\ \href
  {https://doi.org/10.1007/jhep12(2019)079} {10.1007/jhep12(2019)079} (\bibinfo
  {year} {2019})\BibitemShut {NoStop}%
\bibitem [{\citenamefont {Marino}\ \emph {et~al.}(2022)\citenamefont {Marino},
  \citenamefont {Eckstein}, \citenamefont {Foster},\ and\ \citenamefont
  {Rey}}]{JamirDPT}%
  \BibitemOpen
  \bibfield  {author} {\bibinfo {author} {\bibfnamefont {J.}~\bibnamefont
  {Marino}}, \bibinfo {author} {\bibfnamefont {M.}~\bibnamefont {Eckstein}},
  \bibinfo {author} {\bibfnamefont {M.~S.}\ \bibnamefont {Foster}},\ and\
  \bibinfo {author} {\bibfnamefont {A.~M.}\ \bibnamefont {Rey}},\ }\bibfield
  {title} {\bibinfo {title} {Dynamical phase transitions in the collisionless
  pre-thermal states of isolated quantum systems: theory and experiments},\
  }\href {https://doi.org/10.1088/1361-6633/ac906c} {\bibfield  {journal}
  {\bibinfo  {journal} {Reports on Progress in Physics}\ }\textbf {\bibinfo
  {volume} {85}},\ \bibinfo {pages} {116001} (\bibinfo {year}
  {2022})}\BibitemShut {NoStop}%
\bibitem [{\citenamefont {Kormos}\ \emph {et~al.}(2016)\citenamefont {Kormos},
  \citenamefont {Collura}, \citenamefont {Tak{\'{a}}cs},\ and\ \citenamefont
  {Calabrese}}]{Kormos_2016}%
  \BibitemOpen
  \bibfield  {author} {\bibinfo {author} {\bibfnamefont {M.}~\bibnamefont
  {Kormos}}, \bibinfo {author} {\bibfnamefont {M.}~\bibnamefont {Collura}},
  \bibinfo {author} {\bibfnamefont {G.}~\bibnamefont {Tak{\'{a}}cs}},\ and\
  \bibinfo {author} {\bibfnamefont {P.}~\bibnamefont {Calabrese}},\ }\bibfield
  {title} {\bibinfo {title} {Real-time confinement following a quantum quench
  to a non-integrable model},\ }\href {https://doi.org/10.1038/nphys3934}
  {\bibfield  {journal} {\bibinfo  {journal} {Nature Physics}\ }\textbf
  {\bibinfo {volume} {13}},\ \bibinfo {pages} {246} (\bibinfo {year}
  {2016})}\BibitemShut {NoStop}%
\bibitem [{\citenamefont {Engel}\ \emph {et~al.}(2000)\citenamefont {Engel},
  \citenamefont {Nagel},\ and\ \citenamefont {Brendle}}]{booksemigroup}%
  \BibitemOpen
  \bibfield  {author} {\bibinfo {author} {\bibfnamefont {K.-J.}\ \bibnamefont
  {Engel}}, \bibinfo {author} {\bibfnamefont {R.}~\bibnamefont {Nagel}},\ and\
  \bibinfo {author} {\bibfnamefont {S.}~\bibnamefont {Brendle}},\ }\href@noop
  {} {\emph {\bibinfo {title} {One-parameter semigroups for linear evolution
  equations}}},\ Vol.\ \bibinfo {volume} {194}\ (\bibinfo  {publisher}
  {Springer},\ \bibinfo {year} {2000})\BibitemShut {NoStop}%
\bibitem [{\citenamefont {Araki}(1969)}]{Araki}%
  \BibitemOpen
  \bibfield  {author} {\bibinfo {author} {\bibfnamefont {H.}~\bibnamefont
  {Araki}},\ }\bibfield  {title} {\bibinfo {title} {Gibbs states of a one
  dimensional quantum lattice},\ }\href {https://doi.org/10.1007/BF01645134}
  {\bibfield  {journal} {\bibinfo  {journal} {Communications in Mathematical
  Physics}\ }\textbf {\bibinfo {volume} {14}},\ \bibinfo {pages} {120}
  (\bibinfo {year} {1969})}\BibitemShut {NoStop}%
\bibitem [{\citenamefont {Kliesch}\ \emph
  {et~al.}(2014{\natexlab{b}})\citenamefont {Kliesch}, \citenamefont {Gogolin},
  \citenamefont {Kastoryano}, \citenamefont {Riera},\ and\ \citenamefont
  {Eisert}}]{EisertLocality}%
  \BibitemOpen
  \bibfield  {author} {\bibinfo {author} {\bibfnamefont {M.}~\bibnamefont
  {Kliesch}}, \bibinfo {author} {\bibfnamefont {C.}~\bibnamefont {Gogolin}},
  \bibinfo {author} {\bibfnamefont {M.~J.}\ \bibnamefont {Kastoryano}},
  \bibinfo {author} {\bibfnamefont {A.}~\bibnamefont {Riera}},\ and\ \bibinfo
  {author} {\bibfnamefont {J.}~\bibnamefont {Eisert}},\ }\bibfield  {title}
  {\bibinfo {title} {Locality of temperature},\ }\href
  {https://doi.org/10.1103/PhysRevX.4.031019} {\bibfield  {journal} {\bibinfo
  {journal} {Phys. Rev. X}\ }\textbf {\bibinfo {volume} {4}},\ \bibinfo {pages}
  {031019} (\bibinfo {year} {2014}{\natexlab{b}})}\BibitemShut {NoStop}%
\end{thebibliography}%

\end{document}